\documentclass[12pt]{article}

\usepackage{setspace}

\usepackage{geometry}
\newtheorem{lemma}{Lemma}
\newtheorem{proposition}{Proposition}
\newtheorem{theorem}{Theorem}

\newcommand{\proof}{\vspace{-.5ex}\textit{Proof:} }

\newcommand{\qed}{\hfill$\Box$\smallskip}

\usepackage[utf8]{inputenc}
\usepackage[english]{babel}
\usepackage{epic}

\usepackage{amsmath,amssymb,wasysym}
\usepackage{algorithm}
\usepackage{eucal}
\usepackage{xcolor}
\usepackage{booktabs}
\usepackage{multirow}
\usepackage{graphicx}
\usepackage{xspace}
\usepackage{array}

\usepackage[round]{natbib}

\usepackage{algorithm}
\usepackage{algorithmic}

\newcolumntype{C}[1]{>{\centering\let\newline\\\arraybackslash\hspace{0pt}}m{#1}}

\usepackage{mathtools}

\usepackage{url}

\newcommand\adj{adj}
\newcommand\id{id}
\newcommand\cp{\circ}
\newcommand\ext{\gamma}
\newcommand\cext[2]{\cp_{\!#1}^{#2}}

\newcommand{\adjedges}[1]{E_{\adj}^{#1}}
\newcommand{\selfedges}[1]{E_{\id}^{#1}}

\newcommand\tel{}
\newcommand\MR{M\!R}
\newcommand\MRC{\MR_{\circ}}
\renewcommand\AC{A_{\circ}}
\newcommand\BC{B_{\circ}}
\newcommand\AM{A^{\ddagger}}
\newcommand\BM{B^{\ddagger}}
\newcommand\AMi[1]{A^{\ddagger_{#1}}}
\newcommand\BMi[1]{B^{\ddagger_{#1}}}

\def\no{\varepsilon}
\newcommand\AB{A\!B}
\newcommand\BA{B\!A}
\newcommand\BB{B\!B}
\renewcommand\AA{A\!A}

\newcommand\tinysub[1]{\mspace{2mu}\!\scalebox{.65}{#1}}

\newcommand\AAno{\AA_\no}
\newcommand\ABno{\AB_\no}
\newcommand\BBno{\BB_\no}
\newcommand\AAaa{\AA_{\ma}}
\newcommand\AAbb{\AA_{\mb}}
\newcommand\AAab{\AA_{\mab}}
\newcommand\AAba{\AA_{\mba}}
\newcommand\BBaa{\BB_{\!\ma}}
\newcommand\BBbb{\BB_{\mb}}
\newcommand\BBab{\BB_{\!\mab}}
\newcommand\BBba{\BB_{\mba}}
\newcommand\ABaa{\AB_{\!\ma}}
\newcommand\ABbb{\AB_{\mb}}
\newcommand\ABab{\AB_{\!\mab}}
\newcommand\ABba{\AB_{\mba}}
\newcommand\BAab{\BA_{\mab}}
\newcommand\BAba{\BA_{\mba}}
\newcommand\ABend{\bullet}

\newcommand\ma{\mathcal A}
\newcommand\mb{\mathcal B}
\newcommand\mg{\mathcal G}
\newcommand\mtu{\mathcal U}
\newcommand\mab{\ma\!\mb}
\newcommand\mba{\mb\!\ma}

\newcommand\tl[1]{#1^{\z\!t}}
\newcommand\hd[1]{#1^{\z\!h}}
\newcommand\rev[1]{\overline{#1}}
\newcommand\occ{\Phi}
\newcommand\occof[1]{\occ_{\!#1}}
\newcommand\diff{\Delta\occ}
\newcommand\ddcj{d_{\scriptscriptstyle D\z\!C\!J}}
\newcommand\ddcjid{\ddcj^{\scriptscriptstyle id}}
\newcommand\T[1]{{\tt #1}}
\newcommand\Ti[2]{\T{#1}_{\tinysub{\T {#2}}}}
\newcommand\mk{m}
\newcommand\z{\mspace{2mu}}

\newcommand\W{\tt{W}}
\newcommand\M{\tt{M}}
\newcommand\WA{\overline{\W}}
\newcommand\WB{\underline{\W}}
\newcommand\MA{\overline{\M}}
\newcommand\MB{\underline{\M}}
\newcommand\N{\tt{N}}
\newcommand\Z{\tt{Z}}

\begin{document}

{\bf\large Computing the rearrangement distance of natural genomes}
\vspace{5ex}

Leonard Bohnenk\"amper, 
Mar\'ilia D. V. Braga, 
Daniel Doerr, 
and
Jens Stoye$^{\star}$ 
\vspace{5ex}

Faculty of Technology and Center for Biotechnology (CeBiTec)\\ Bielefeld University\\
Postfach 10 01 31\\ 
33501 Bielefeld\\ 
Germany\\
Phone: +49 521 106 3840\\
Fax: +49 521 106 6495

\vspace{5ex}

$^\star$
Corresponding author\\
Email: \url{jens.stoye@uni-bielefeld.de}
\vspace{5ex}

Running title:\\
Rearrangement distance of natural genomes
\vspace{5ex}

\newpage

\begin{abstract}
The computation of genomic distances has been a very active field of computational comparative genomics over the last 25 years.
Substantial results include the polynomial-time computability of the inversion distance by Hannenhalli and Pevzner in 1995 and the introduction of the double-cut and join (DCJ) distance by Yancopoulos, Attie and Friedberg in 2005.
Both results, however, rely on the assumption that the genomes under comparison contain the same set of unique \emph{markers} (syntenic genomic regions, sometimes also referred to as \emph{genes}).
In 2015, Shao, Lin and Moret relax this condition by allowing for duplicate markers in the analysis. This generalized version of the genomic distance problem is NP-hard, and they give an ILP solution that is efficient enough to be applied to real-world datasets. A restriction of their approach is that it can be applied only to \emph{balanced} genomes, that have equal numbers of duplicates of any marker. Therefore it still needs a delicate preprocessing of the input data in which excessive copies of unbalanced markers have to be removed.

In this paper we present an algorithm solving the genomic distance problem for \emph{natural} genomes, in which any marker may occur an arbitrary number of times. Our method is based on a new graph data structure, the \emph{multi-relational diagram}, that allows an elegant extension of the ILP by Shao, Lin and Moret to count \emph{runs} of markers that are under- or over-represented in one genome with respect to the other and need to be inserted or deleted, respectively. With this extension, previous restrictions on the genome configurations are lifted, for the first time enabling an uncompromising rearrangement analysis. Any marker sequence can directly be used for the distance calculation.

The evaluation of our approach shows that it can be used to analyze genomes with up to a few ten thousand markers, which we demonstrate on simulated and real data.
Source code and test data are available from~\url{https://gitlab.ub.uni-bielefeld.de/gi/ding}.

\end{abstract}


\newpage

\section{Introduction}

The study of genome rearrangements has a long tradition in comparative genomics.
A central question is how many (and what kind of) mutations have occurred between the genomic sequences of two individual genomes. 
In order to avoid disturbances due to minor local effects, often the basic units in such comparisons are syntenic regions identified between the genomes under study, much larger than the individual DNA bases.
We refer to such regions as \emph{genomic markers}, or simply \emph{markers}, although often one also finds the term \emph{genes}.

Following the initial statement as an edit distance problem~\citep{SAN-1992}, a comprehensive trail of literature has addressed the problem of computing the number of rearrangements between two genomes.
In a seminal paper in 1995, \citet{HAN-PEV-1999} introduced the first polynomial time algorithm for the computation of the inversion distance of transforming one 
chromosome into another one by means of segmental inversions. Later, the same authors generalized their results to the HP model~\citep{HAN-PEV-1995} which is capable of handling multi-chromosomal genomes and accounts for additional genome rearrangements. Another breakthrough was the introduction of the double cut and join (DCJ) model~\citep{YAN-ATT-FRI-2005,BER-MIX-STO-2006},
that is able to capture many genome rearrangements and whose genomic distance is computable in linear time. 
The model is based on a simple operation in which the genome sequence is cut twice between two consecutive markers and re-assembled by joining the resulting four loose cut-ends in a different combination. 

A prerequisite for applying the DCJ model in practice is that their genomic marker sets must be identical and that any marker occurs exactly once in each genome. This severely limits its applicability in practice. Linear time extensions of the DCJ model allow markers to occur exclusively in one of the two genomes, computing a genomic \emph{DCJ-indel distance} that minimizes the sum of DCJ and insertion/deletion (indel) events~\citep{BRA-WIL-STO-2011,COM-2013}.
Still, markers are required to be \emph{singleton}, i.e., no duplicates can occur. 
When duplicates are allowed, the problem is more intrincate and
all approaches proposed so far are NP-hard, see for
instance~\citep{SAN-1999,BRY-2000,BUL-JIA-2013,ANG-FER-RUS-THE-VIA-2009,MAR-FEI-BRA-STO-2015,SHA-LIN-MOR-2015,YIN-TAN-SCH-BAD-2016}.
From the practical side, more recently, \citet{SHA-LIN-MOR-2015} presented an integer linear programming (ILP) formulation for computing the DCJ distance in presence of duplicates, but restricted to \emph{balanced genomes}, where both genomes have equal numbers of duplicates. \citet{YIN-TAN-SCH-BAD-2016} then developed 
a branch and bound approach to compute the DCJ-indel distance of \emph{quasi-balanced genomes}, that have equal number of duplicated common markers but also markers that occur exclusively in one of the two genomes.
An ILP that computes the DCJ-indel distance of \emph{unbalanced genomes} was later presented by \citet{LYU-GER-GOR-2017}, but their approach does not seem to be applicable to real data sets, see Section~\ref{sec:ilp-details} 
for details.

In this paper we present the first feasible exact algorithm for solving the NP-hard problem of computing the distance under a general genome model where any marker may occur an arbitrary number of times in any of the two genomes, called \emph{natural genomes}.
Specifically, we adopt the \emph{maximal matches} model where only markers appearing more often in one genome than in the other can be deleted or inserted. Our ILP formulation is based on the one from \citet{SHA-LIN-MOR-2015}, but with an efficient extension that allows to count \emph{runs} of markers that are under- or over-represented in one genome with respect to the other, so that the pre-existing model of minimizing the distance allowing DCJ and indel operations~\citep{BRA-WIL-STO-2011} can be adapted to our problem. With this extension, once we have the genome markers, no other restriction on the genome configurations is imposed.

The evaluation of our approach 
shows that it can be used to analyze genomes with up to a few ten thousand markers, provided the number of duplicates is not too large.
The complete source code of our ILP implementation and the simulation software used for generating the benchmarking data in Section~\ref{sec:eva} are available from~\url{https://gitlab.ub.uni-bielefeld.de/gi/ding}.

This paper is an extended version of earlier work that was presented at RECOMB~2020~\citep{BOH-BRA-DOE-STO-2020a}.

\section{Preliminaries}\label{sec:preliminaries}

A \emph{genome} is a set of \emph{chromosomes} and each chromosome can be linear or circular.
Each \emph{marker} in a chromosome is an oriented DNA fragment.
The representation of a marker~$\mk$ in a chromosome can be
the symbol~$\mk$ itself, if it is read in direct orientation, or the
symbol~$\rev{\mk}$, if it is read in reverse orientation.
We represent a chromosome $S$ of a genome~$A$ by a string~$s$, obtained by the concatenation of all symbols in~$S$, read in any of the two directions. If $S$ is circular, we can start to read it at any marker and the string $s$ is flanked by parentheses. 

Given two genomes~$A$ and~$B$, let $\mtu$ be the set of all markers that occur in any of the two genomes.
For each marker $\mk \in \mtu$,
let $\occof{A}(\mk)$ be the number of occurrences of $\mk$ in genome $A$ and $\occof{B}(\mk)$ be the number of occurrences of $\mk$ in genome $B$.
We can then define
$\diff(\mk)=\occof{A}(\mk)-\occof{B}(\mk)$. 
If both $\occof{A}(\mk)>0$ and $\occof{B}(\mk)>0$, $\mk$ is called a {\em common marker}. We denote by
$\mg\subseteq\mtu$ the set of common markers of $A$ and $B$.
The markers in $\mtu\backslash\mg$ are called {\em exclusive markers}.  
For example, if
we have two unichromosomal linear genomes ~$A=\{\tel\T{1}\z\T{3}\z\T{2}\z\rev{\T{5}}\z\rev{\T{4}}\z\T{3}\z\T{5}\z\T{4}\tel\}$ and~$B=\{\tel\T{1}\z\T{6}\z\T{2}\z\T{3}\z\T{1}\z\T{7}\z\T{3}\z\T{4}\z\T{1}\z\T{3}\tel\}$, 
then~$\mtu=\{\T{1},\T{2},\T{3},\T{4},\T{5},\T{6},\T{7}\}$ and $\mg=\{\T{1},\T{2},\T{3},\T{4}\}$.
Furthermore, $\diff(\T{1})\!=\!1\!-\!3\!=\!-2$, $\diff(\T{2})\!=\!1\!-\!1\!=\!0$, $\diff(\T{3})\!=\!2\!-\!3\!=\!-1$,
$\diff(\T{4})\!=\!2\!-\!1\!=\!1$, 
$\diff(\T{5})\!=\!2$, 
and $\diff(\T{6})=\diff(\T{7})=-1$.

\subsection{The DCJ-indel model}
A genome can be transformed or \emph{sorted} into another genome
with the following types of mutations:

\begin{itemize}
  \item A \emph{double-cut-and-join} (DCJ) is the operation
that cuts a genome at two different positions (possibly in two different
chromosomes), creating four open ends, and joins these
open ends in a different way. This can represent many different rearrangements,
such as inversions, translocations, fusions and fissions. 
For example, a DCJ can cut linear chromosome~$\tel\T{1}\z\T{2}\z\rev{\T{4}}\z\rev{\T{3}}\z\T{5}\z\T{6}\tel$ 
before and after~$\rev{\T{4}}\z\rev{\T{3}}$, creating the segments~$\tel\T{1}\z\T{2}\bullet$,
 $\bullet \rev{\T{4}}\z\rev{\T{3}}\bullet$ and $\bullet\T{5}\z\T{6}\tel$, where the symbol~$\bullet$ represents the
open ends. By joining the first with the third and the second with the
fourth open end, we invert $\rev{\T{4}}\z\rev{\T{3}}$ and obtain $\tel\T{1}\z\T{2}\z\T{3}\z\T{4}\z\T{5}\z\T{6}\tel$.

\item Since the genomes can have distinct multiplicity of markers, we also need
to consider {\em insertions} and {\em deletions} of segments of contiguous
markers~\citep{YAN-FRI-2009,BRA-WIL-STO-2011,COM-2013}.
We refer to insertions and deletions collectively as \emph{indels}. 
For example, the deletion of segment~$\T{5}\z\T{2}\z\T{6}\z\T{2}$
from linear chromosome~$\tel\T{1}\z\T{2}\z\T{3}\z\T{5}\z\T{2}\z\T{6}\z\T{2}\z\T{4}\tel$
results in~$\tel\T{1}\z\T{2}\z\T{3}\z\T{4}\tel$.
Indels have two restrictions:
(i) only markers that  have positive $\diff$ can be deleted; and (ii) only markers that have negative $\diff$ can be inserted. 
\end{itemize}

\noindent
In this paper, we are interested in
computing the \emph{DCJ-indel distance} between two genomes $A$ and~$B$, that is
denoted by $\ddcjid(A,B)$ and corresponds to the minimum number of DCJs and
indels required to sort $A$ into $B$.
We separate the instances of the problem in four types:

\begin{enumerate}
  \item \emph{Singular genomes}: the genomes contain no duplicate markers, that is,
  each common marker %
  is singular in each genome. (The exclusive markers are not restricted to be singular, because it is mathematically trivial to transform them into singular markers when they occur in multiple copies.)
  Formally, we have that, for each $\mk \in
  \mg$, $\occof{A}(\mk)=\occof{B}(\mk)=1$. 
  The distance between singular genomes can be easily computed in linear
  time~\citep{BER-MIX-STO-2006,BRA-WIL-STO-2011,COM-2013}.
  \item \emph{Balanced genomes}: the genomes contain no exclusive markers, but can have duplicates,
  and the number of duplicates in each genome is the same. Formally, we have $\mg = \mtu$ and,
  for each $\mk \in \mtu$,
  $\occof{A}(\mk)=\occof{B}(\mk)$. Computing the distance for this set of instances
  is NP-hard, and an ILP formulation was given in~\citep{SHA-LIN-MOR-2015}. 
  \item \emph{Quasi-balanced genomes}: the genomes contain exclusive markers and can have duplicates,
  but still the number of duplicates in each genome is the same. Formally, we have $\mg \subseteq \mtu$ and,
  for each $\mk \in \mg$,
  $\occof{A}(\mk)=\occof{B}(\mk)$. Computing the distance for this set of instances
  is also NP-hard, and a branch and bound approach was given in~\citep{YIN-TAN-SCH-BAD-2016}. 
  \item \emph{Natural genomes}: these genomes can have exclusive markers and duplicates, with
  no restrictions on the number of copies. Since these are
  generalizations of balanced genomes, computing the distance for this set of
  instances is also NP-hard.
  In the present work we present an efficient ILP formulation for computing the distance in this case.
\end{enumerate} 

\section{DCJ-indel distance of singular genomes}
  
First we recall the problem when common duplicates do not occur, that is, when we have singular genomes.
We will summarize the linear time approach to compute the DCJ-indel distance in this
case that was presented in~\citep{BRA-WIL-STO-2011}, already adapted to the notation required for presenting the new results of this paper. 

\subsection{Relational diagram}

For computing a genomic distance it is useful to represent the relation between two genomes in some graph structure~\citep{HAN-PEV-1995,BER-MIX-STO-2006,FRI-DAR-YAN-2008,BRA-WIL-STO-2011}. Here we adopt a variation of this structure, defined as follows.
For each marker $\mk$, denote its two extremities by~$\tl{\mk}$ (tail) and~$\hd{\mk}$ (head).
Given two singular genomes~$A$ and~$B$, the \emph{relational diagram} $R(A,B)$ has a set of vertices $V = V(A) \cup V(B)$, where $V(A)$ has a vertex for each extremity of each marker of genome $A$ and $V(B)$ has a vertex for each extremity of each marker of genome $B$.
Due to the 1-to-1 correspondence between the
vertices of~$R(A,B)$ and the occurrences of marker extremities in~$A$ and~$B$, we
can identify each extremity with its corresponding vertex.
It is convenient to represent vertices in $V(A)$ in an upper line, respecting the order in which they appear in each chromosome of $A$, and the vertices in $V(B)$ in a lower line, respecting the order in which they appear in each chromosome of $B$.

If the marker extremities $\gamma_1$ and $\gamma_2$ are adjacent in a
chromosome of $A$, we have an \emph{adjacency edge} connecting them. Similarly,
if the marker extremities $\gamma_1'$ and $\gamma_2'$ are adjacent in a
chromosome of $B$, we have an adjacency edge connecting them.
Marker extremities located at chromosome ends are called \emph{telomeres}
and are not connected to any adjacency edge. 
In contrast, each extremity that is not a telomere is connected to exactly one adjacency edge.
Denote by $\adjedges{A}$ and by $\adjedges{B}$ the adjacency edges in $A$ and in
$B$, respectively.
In addition, for each common marker~$\mk\in\mg$,
we have two \emph{extremity edges}, one connecting the vertex $\hd{\mk}$ from~$V(A)$ to the vertex $\hd{\mk}$ from~$V(B)$ and the other connecting the vertex~$\tl{\mk}$ from $V(A)$ to the vertex $\tl{\mk}$ from~$V(B)$.
Denote by $E_{\ext}$ the set of extremity edges.
Finally, 
for each occurrence of an exclusive marker in~$\mtu\backslash\mg$,
we have an \emph{indel edge} 
connecting the  vertices representing its two extremities.
Denote by $\selfedges{A}$ and by $\selfedges{B}$ the indel edges in $A$ and in
$B$. Each vertex is then connected either to an extremity edge or to an indel edge.

All vertices have degree one or two, therefore~$R(A, B)$ is a simple collection of cycles and paths.
A path that has one endpoint 
in genome~$A$ and the other in genome~$B$ is called an {\em AB-path}. In the same 
way, both endpoints of an {\em AA-path} are in~$A$ and both 
endpoints of
a {\em BB-path} are in~$B$. A cycle contains either zero or an even number of extremity edges. When a cycle has at least two extremity edges, it is called an {\em AB-cycle}.
Moreover, a path (respectively cycle) of~$R(A, B)$ composed exclusively of indel and adjacency edges in one of the two genomes corresponds to a whole linear (respectively circular) chromosome and is called a {\em linear} (respectively {\em circular}) {\em
singleton} in that genome. Actually, linear singletons are particular cases of~$\AA$-paths or~$\BB$-paths. 
An example of a relational diagram for singular is given in Figure~\ref{fig:bp-adj-graph}.


\begin{figure}
  \begin{center}
    \scriptsize
    \setlength{\unitlength}{.7pt}
    \begin{picture}(500,80)

      \put(0,0){\color{blue}
        \qbezier(0,70)(20,40)(40,10)
      }
      
      \dottedline[$\cdot$]{3}(30,70)(45,70)
      \dottedline[$\cdot$]{3}(75,70)(90,70)
      \dottedline[$\cdot$]{3}(295,70)(310,70)
      \qbezier(30,70)(50,40)(70,10)
      \qbezier(90,70)(178,40)(265,10)
      \dottedline[$\cdot$]{3}(70,10)(85,10)
      \dottedline[$\cdot$]{3}(115,10)(130,10)
      \dottedline[$\cdot$]{3}(250,10)(265,10)
      \qbezier(310,70)(220,40)(130,10)
      \drawline(295,70)(250,10)
      \drawline(40,70)(80,70)
      \drawline(80,10)(120,10)
      
      \put(0,0){\color{blue}
        \dottedline[$\cdot$]{3}(120,70)(135,70)
        \qbezier(135,70)(155,40)(175,10)
        \qbezier(120,70)(208,40)(295,10)
        \qbezier(340,70)(250,40)(160,10)
        \dottedline[$\cdot$]{3}(160,10)(175,10)
      }
      \put(0,0){\color{red}
      \drawline(165,70)(205,10)
      \drawline(265,70)(220,10)
      \dottedline[$\cdot$]{3}(205,10)(220,10)
      }
      
      \put(0,0){\color{red}
      \dottedline[$\cdot$]{3}(430,10)(445,10)
      \drawline(400,10)(430,10)
      \drawline(445,10)(475,10)
      }
      
      \put(  0,70){\circle*{5}}
      \put( 30,70){\circle*{5}}
      \put( 45,70){\circle*{5}}
      \put( 75,70){\circle*{5}}
      \put( 90,70){\circle*{5}}
      \put(120,70){\circle*{5}}
      \put(135,70){\circle*{5}}
      \put(165,70){\circle*{5}}
      
      \put(265,70){\circle*{5}}
      \put(295,70){\circle*{5}}
      \put(310,70){\circle*{5}}
      \put(340,70){\circle*{5}}
     
      \put(-40,82){\makebox(0,0){$A$}}
      \put( 0,82){\makebox(0,0){$\tl{\T{1}}$}}
      \put( 30,82){\makebox(0,0){$\hd{\T{1}}$}}
      \put( 45,82){\makebox(0,0){$\hd{\T{6}}$}}
      \put( 75,82){\makebox(0,0){$\tl{\T{6}}$}}
      \put( 90,82){\makebox(0,0){$\tl{\T{5}}$}}
      \put(120,82){\makebox(0,0){$\hd{\T{5}}$}}
      \put(135,82){\makebox(0,0){$\tl{\T{3}}$}}
      \put(165,82){\makebox(0,0){$\hd{\T{3}}$}}
      \put(265,82){\makebox(0,0){$\tl{\T{4}}$}}
      \put(295,82){\makebox(0,0){$\hd{\T{4}}$}}
      \put(310,82){\makebox(0,0){$\tl{\T{2}}$}}
      \put(340,82){\makebox(0,0){$\hd{\T{2}}$}}
      
      \put( 40,10){\circle*{5}}
      \put( 70,10){\circle*{5}}
      \put( 85,10){\circle*{5}}
      \put(115,10){\circle*{5}}
      \put(130,10){\circle*{5}}
      \put(160,10){\circle*{5}}
      \put(175,10){\circle*{5}}
      \put(205,10){\circle*{5}}
      \put(220,10){\circle*{5}}
      \put(250,10){\circle*{5}}
      \put(265,10){\circle*{5}}
      \put(295,10){\circle*{5}}
      
      \put(400,10){\circle*{5}}
      \put(430,10){\circle*{5}}
      \put(445,10){\circle*{5}}
      \put(475,10){\circle*{5}}

      \put(-40,-2){\makebox(0,0){$B$}}
      \put( 40,-2){\makebox(0,0){$\tl{\T{1}}$}}
      \put( 70,-2){\makebox(0,0){$\hd{\T{1}}$}}
      \put( 85,-2){\makebox(0,0){$\tl{\T{7}}$}}
      \put(115,-2){\makebox(0,0){$\hd{\T{7}}$}}
      \put(130,-2){\makebox(0,0){$\tl{\T{2}}$}}
      \put(160,-2){\makebox(0,0){$\hd{\T{2}}$}}
      \put(175,-2){\makebox(0,0){$\tl{\T{3}}$}}
      \put(205,-2){\makebox(0,0){$\hd{\T{3}}$}}
      \put(220,-2){\makebox(0,0){$\tl{\T{4}}$}}
      \put(250,-2){\makebox(0,0){$\hd{\T{4}}$}}
      \put(265,-2){\makebox(0,0){$\tl{\T{5}}$}}
      \put(295,-2){\makebox(0,0){$\hd{\T{5}}$}}
      \put(400,-2){\makebox(0,0){$\tl{\T{7}}$}}
      \put(430,-2){\makebox(0,0){$\hd{\T{7}}$}}
      \put(445,-2){\makebox(0,0){$\hd{\T{8}}$}}
      \put(475,-2){\makebox(0,0){$\tl{\T{8}}$}}
      
    \end{picture}
  \end{center}
  \caption{\label{fig:bp-adj-graph}For genomes $A=\{\tel \T{1}\z\rev{\T{6}}\z\T{5}\z\T{3}\tel, \tel \T{4}\z\T{2}\tel \}$ and $B=\{\tel \T{1}\z\T{7}\z\T{2}\z\T{3}\z\T{4}\z\T{5}\tel, \tel \T{7}\z\rev{\T{8}}\tel\}$, the relational diagram contains one cycle, two $\AB$-paths (represented in blue), one $\AA$-path and one $\BB$-path (both represented in red).
  Short dotted horizontal edges are adjacency edges, long horizontal edges are indel edges, top-down edges are extremity edges.}
\end{figure}
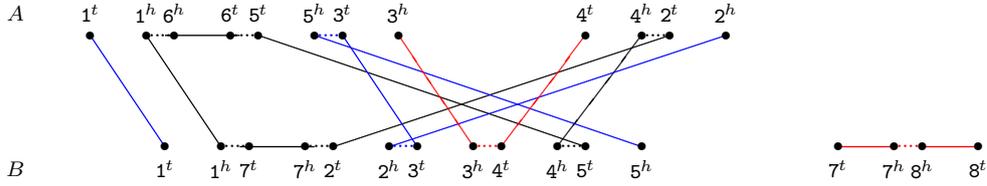


The numbers of telomeres and of $\AB$-paths in $R(A, B)$ are even. 
The \emph{DCJ-cost}~\citep{BRA-WIL-STO-2011} of a DCJ operation $\rho$, denoted by $\|\rho\|$, is defined as follows. 
If it either increases the number of $\AB$-cycles by
one, or the number of $\AB$-paths by two, $\rho$ is \emph{optimal} and has $\|\rho\|=0$. If it does not affect the number of $\AB$-cycles and
$\AB$-paths in the diagram, $\rho$ is \emph{neutral} and has $\|\rho\|=1$. If it either decreases the number of $\AB$-cycles by one, or the number of
$\AB$-paths by two, $\rho$ is \emph{counter-optimal} and has $\|\rho\|=2$.

\subsection{Runs and indel-potential}

The approach that uses DCJ operations to group exclusive markers for minimizing indels depends on the following concepts.

Given two genomes~$A$ and~$B$ and a component~$C$ of~$R(A,B)$, a \emph{run}~\citep{BRA-WIL-STO-2011} is a maximal subpath of~$C$ in which the first and the last edges are indel edges and all indel edges belong to the same genome. 
It can be an $\ma$-run when its indel edges are in genome~$A$, or a $\mb$-run when its indel edges are in genome~$B$. 
We denote by $\Lambda(C)$ the number of runs in component~$C$. If $\Lambda(C)\geq 1$ the component $C$ is said to be \emph{indel-enclosing}, otherwise $\Lambda(C)=0$ and $C$ is said to be \emph{indel-free}.

While sorting components separately with optimal DCJs only, runs can be \emph{merged} (when two runs become a single one), and also \emph{accumulated} together (when all its indel edges alternate with adjacency edges only and the run 
can be inserted or deleted at once)~\citep{BRA-WIL-STO-2011}. 
The \emph{indel-potential} of a component~$C$, denoted by~$\lambda(C)$, is
the minimum number of indels derived from $C$ after this process
and can be directly computed from~$\Lambda(C)$:
\[
\lambda(C) = 
\begin{cases}
~~~~~0\:, & \mbox{ if $\Lambda(C)=0$~~~($C$ is indel-free)}; \\[1mm]
\left\lceil \frac{\Lambda(C)+1}{2}\right\rceil\:, & \mbox{ if $\Lambda(C) \geq 1$~~~($C$ is indel-enclosing)}.
\end{cases}
\] 

Figure~\ref{fig:run} shows a $\BB$-path with 4 runs, and how its indel-potential can be achieved.

\begin{figure}
\begin{center}
\tiny
\setlength{\unitlength}{.6pt}
\begin{tabular}{ccc}
{\bf (i)} &&{\bf (ii)}\\
  \begin{picture}(310,90)

    \drawline(10,70)(10,30)
    \dottedline{1}(10,70)(20,70)
    \put(0,0){\color{red}\drawline(20,70)(40,70)}
    \dottedline{1}(40,70)(50,70)
    \drawline(50,70)(50,30)
    \dottedline{1}(50,30)(60,30)
    \drawline(60,70)(60,30)
    \dottedline{1}(60,70)(70,70)
    \drawline(70,70)(70,30)
    \dottedline{1}(70,30)(80,30)
    \put(0,0){\color{red}\drawline(80,30)(100,30)}
    \dottedline{1}(100,30)(110,30)
    \drawline(110,70)(110,30)
    \dottedline{1}(110,70)(120,70)
    \put(0,0){\color{red}\drawline(120,70)(140,70)}
    \dottedline{1}(140,70)(150,70)
    \drawline(150,70)(150,30)
    \dottedline{1}(150,30)(160,30)
    \drawline(160,70)(160,30)
    \dottedline{1}(160,70)(170,70)
    \put(0,0){\color{red}\drawline(170,70)(190,70)}
    \dottedline{1}(190,70)(200,70)
    \drawline(200,70)(200,30)
    \dottedline{1}(200,30)(210,30)
    \put(0,0){\color{red}\drawline(210,30)(230,30)}
    \dottedline{1}(230,30)(240,30)
    \drawline(240,70)(240,30)
    \dottedline{1}(240,70)(250,70)
    \drawline(250,70)(250,30)
    \dottedline{1}(250,30)(260,30)
    \drawline(260,70)(260,30)
    \dottedline{1}(260,70)(270,70)
    \drawline(270,70)(270,30)
    \dottedline{1}(270,30)(280,30)
    \put(0,0){\color{red}\drawline(280,30)(300,30)}
    
    \put( 10,70){\circle*{3}}
    \put( 20,70){\circle*{3}}
    \put( 40,70){\circle*{3}}
    \put( 50,70){\circle*{3}}
    \put( 60,70){\circle*{3}}
    \put( 70,70){\circle*{3}}
    \put(110,70){\circle*{3}}
    \put(120,70){\circle*{3}}
    \put(140,70){\circle*{3}}
    \put(150,70){\circle*{3}}
    \put(160,70){\circle*{3}}
    \put(170,70){\circle*{3}}
    \put(190,70){\circle*{3}}
    \put(200,70){\circle*{3}}
    \put(240,70){\circle*{3}}
    \put(250,70){\circle*{3}}
    \put(260,70){\circle*{3}}
    \put(270,70){\circle*{3}}
    
    \put( 30,82){\makebox(0,0){$e_1$}}
    \put(130,82){\makebox(0,0){$e_3$}}
    \put(115,70){\makebox(0,0){$\backslash$}}
    \put(180,82){\makebox(0,0){$e_4$}}
    \put(195,70){\makebox(0,0){$/$}}
    \put(270,82){\makebox(0,0){}}
    \put(330,82){\makebox(0,0){}}
    
    \put( 10,30){\circle*{3}}
    \put( 50,30){\circle*{3}}
    \put( 60,30){\circle*{3}}
    \put( 70,30){\circle*{3}}
    \put( 80,30){\circle*{3}}
    \put(100,30){\circle*{3}}
    \put(110,30){\circle*{3}}
    \put(150,30){\circle*{3}}
    \put(160,30){\circle*{3}}
    \put(200,30){\circle*{3}}
    \put(210,30){\circle*{3}}
    \put(230,30){\circle*{3}}
    \put(240,30){\circle*{3}}
    \put(250,30){\circle*{3}}
    \put(260,30){\circle*{3}}
    \put(270,30){\circle*{3}}
    \put(280,30){\circle*{3}}
    \put(300,30){\circle*{3}}

    \put(30,-2){\makebox(0,0){$\underbrace{~~~}_{\mathcal
    A\mbox{-run}}$}} 
    \put(90,2){\makebox(0,0){$\underbrace{~e_2~}_{\mathcal
    B\mbox{-run}}$}}
    \put(155,-4){\makebox(0,0){$\underbrace{~~~~~~~~~~~~~~~~}_{\mathcal
    A\mbox{-run}}$}} 
    \put(257,2){\makebox(0,0){$\underbrace{e_5~~~~~~~~~~~~~e_6}_{\mathcal
    B\mbox{-run}}$}}
    
  \end{picture}  &
    \begin{picture}(100,90)
     \put(50,60){\makebox(0,0){$\rightarrow$}}
     \put(50,40){\makebox(0,0){optimal}}
     \put(50,20){\makebox(0,0){DCJ}}
    \end{picture} &
    \begin{picture}(310,90)
    
    \drawline(10,70)(10,30)
    \dottedline{1}(10,70)(20,70)
    \put(0,0){\color{red}\drawline(20,70)(40,70)}
    \dottedline{1}(40,70)(50,70)
    \drawline(50,70)(50,30)
    \dottedline{1}(50,30)(60,30)
    \drawline(60,70)(60,30)
    \dottedline{1}(60,70)(70,70)
    \drawline(70,70)(70,30)
    \dottedline{1}(70,30)(80,30)
    \put(0,0){\color{red}\drawline(80,30)(100,30)}
    \dottedline{1}(100,30)(110,30)
    \drawline(110,70)(110,30)
    \dottedline{1}(110,70)(120,70)
    \drawline(120,70)(120,30)
    \dottedline{1}(120,30)(130,30)
    \put(0,0){\color{red}\drawline(130,30)(150,30)}
    \dottedline{1}(150,30)(160,30)
    \drawline(160,70)(160,30)
    \dottedline{1}(160,70)(170,70)
    \drawline(170,70)(170,30)
    \dottedline{1}(170,30)(180,30)
    \drawline(180,70)(180,30)
    \dottedline{1}(180,70)(190,70)
    \drawline(190,70)(190,30)
    \dottedline{1}(190,30)(200,30)
    \put(0,0){\color{red}\drawline(200,30)(220,30)}
    
    \put( 10,70){\circle*{3}}
    \put( 20,70){\circle*{3}}
    \put( 40,70){\circle*{3}}
    \put( 50,70){\circle*{3}}
    \put( 60,70){\circle*{3}}
    \put( 70,70){\circle*{3}}
    \put(110,70){\circle*{3}}
    \put(120,70){\circle*{3}}
    \put(160,70){\circle*{3}}
    \put(170,70){\circle*{3}}
    \put(180,70){\circle*{3}}
    \put(190,70){\circle*{3}}
    
    \put( 30,82){\makebox(0,0){$e_1$}}
    \put(150,82){\makebox(0,0){}}
    \put(210,82){\makebox(0,0){}}
    \put(270,82){\makebox(0,0){}}
    
    \put( 10,30){\circle*{3}}
    \put( 50,30){\circle*{3}}
    \put( 60,30){\circle*{3}}
    \put( 70,30){\circle*{3}}
    \put( 80,30){\circle*{3}}
    \put(100,30){\circle*{3}}
    \put(110,30){\circle*{3}}
    \put(115,70){\circle{7}}
    \put(120,30){\circle*{3}}
    \put(130,30){\circle*{3}}
    \put(150,30){\circle*{3}}
    \put(160,30){\circle*{3}}
    \put(170,30){\circle*{3}}
    \put(180,30){\circle*{3}}
    \put(190,30){\circle*{3}}
    \put(200,30){\circle*{3}}
    \put(220,30){\circle*{3}}
    
    \put(30,-2){\makebox(0,0){$\underbrace{~~~}_{\mathcal
    A\mbox{-run}}$}} 
 
    \put(150,2){\makebox(0,0){$\underbrace{e_2~~\;~~~~~~e_5~~~~~~~~~~~~~~e_6}_{\mathcal
    B\mbox{-run}}$}}
    
    \put(0,0){\color{red}\drawline(250,70)(270,70)}
    \put(0,0){\color{red}\drawline(280,70)(300,70)}
    \qbezier(240,70)(255,50)(270,30)
    \qbezier(310,70)(295,50)(280,30)
    \dottedline{1}(240,70)(250,70)
    \dottedline{1}(270,70)(280,70)
    \dottedline{1}(300,70)(310,70)
    \dottedline{1}(270,30)(280,30)
    
    \put(240,70){\circle*{3}}
    \put(250,70){\circle*{3}}
    \put(270,70){\circle*{3}}
    \put(275,70){\circle{7}}
    \put(280,70){\circle*{3}}
    \put(300,70){\circle*{3}}
    \put(310,70){\circle*{3}}
    \put(270,30){\circle*{3}}
    \put(280,30){\circle*{3}}
    \put(260,85){\makebox(0,0){$e_4$}}
    \put(290,85){\makebox(0,0){$e_3$}}

    \put(275,-4){\makebox(0,0){$\underbrace{~~~~~~~~~~~}_{\ma\mbox{-run}}$}}
    \end{picture}
  \end{tabular}
\end{center}
\caption{\label{fig:run}(i) A $\BB$-path with 4 runs. 
(ii) After an optimal DCJ that creates a new cycle, one $\ma$-run is accumulated (between edges $e_4$ and $e_3$ there is only an adjacency edge) and two $\mb$-runs are merged ($e_2$ is in the same run with $e_5$ and $e_6$).
Indeed, the indel-potential of the original $\BB$-path is three.}
\end{figure}
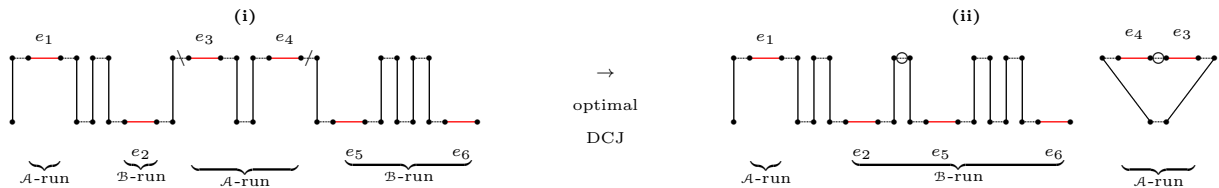

The indel-potential allows to state an upper bound for the DCJ-indel distance:
 
\begin{lemma}[from~\citep{BRA-WIL-STO-2011,BER-MIX-STO-2006}]\label{lem:distance-upper}
Given two singular
genomes $A$ and $B$, whose relational diagram $R(A,B)$ has $c$ $\AB$-cycles and $i$ $\AB$-paths, we have 

\[
\ddcjid(A,B) \;\leq\; |\mg| - c -
\frac{i}{2}~+ \!\!\!\! \sum_{C \in R(A,B)} \!\!\!\!\!\! \lambda(C)\,.
\]
\end{lemma}

Let $\lambda_0$ and $\lambda_1$ be,
respectively, the sum of the indel-potentials for the components of the
relational diagram before and after a DCJ $\rho$.
The \emph{indel-cost} of $\rho$ is then $\Delta \lambda(\rho) = \lambda_1
- \lambda_0$, and the DCJ-indel cost of $\rho$ is defined as
$\Delta d(\rho) = \|\rho\| + \Delta \lambda(\rho)$.
While sorting components separately, 
it has been shown that by using neutral or counter-optimal DCJs one can 
never achieve $\Delta d < 0$, therefore we cannot decrease the upper bound stated above with DCJ operations that act on a single component of the diagram~\citep{BRA-WIL-STO-2011}. 

\subsection{Distance of circular genomes}

For singular circular genomes, the diagram $R(A,B)$ is composed of cycles only. 
In this case the upper bound given by Lemma~\ref{lem:distance-upper} is tight and leads to a simplified formula~\citep{BRA-WIL-STO-2011}:
\[
\ddcjid(A,B) \;=\; |\mg| - c~+ \!\!\!\! \sum_{C \in R(A,B)} \!\!\!\!\!\! \lambda(C)\,.
\]

\subsection{Recombinations and linear genomes}

For singular linear genomes,
the upper bound given by Lemma~\ref{lem:distance-upper} is achieved
when the components of $R(A,B)$ are sorted separately. However, there is another type of DCJ operation, called {\em recombination}, whose cuts are applied on two distinct components. These two components are called {\em sources}, while the components obtained after the joinings are called {\em resultants}. In particular, some recombinations whose both sources are paths have $\Delta d<0$ and are then said to be \emph{deducting}.
The total number of types of deducting recombinations is relatively small. By exhaustively exploring the space of recombination types, it is possible to identify groups of chained recombinations, 
so that the sources of each group are the original paths of the diagram. In other words, a path that is a resultant of a group is never a source of another group.
This results in a greedy approach that optimally finds the value to be deducted, as we will describe in the following

\paragraph{Deducting recombinations.}

Any recombination whose sources are an $\AA$-path and a $\BB$-path is optimal.
A recombination whose sources are two different $\AB$-paths can be either neutral, when the resultants are also $\AB$-paths, or counter-optimal, when the resultants are an $\AA$-path and a $\BB$-path.
Any recombination whose sources are an $\AB$-path and an $\AA$- or a $\BB$-path is neutral~\citep{BRA-STO-2010,BRA-WIL-STO-2011}.

Let $\ma$ (respectively
$\mb$) be a sequence with an odd ($\geq1$) number of runs, starting and ending
with an $\ma$-run (respectively $\mb$-run). We can then make any combination of
$\ma$ and $\mb$, such as $\mab$, that is a sequence with an even ($\geq2$)
number of runs, starting with an $\ma$-run and ending with a $\mb$-run. An
empty sequence (with no run) is represented by $\no$. Then each one of the
notations $\AAno$, $\AAaa$, $\AAbb$,
$\AAab \!\equiv\! \AAba$, $\BBno$,
$\BBaa$, $\BBbb$, $\BBab \!\equiv\!
\BBba$, $\ABno$, $\ABaa$, $\ABbb$,
$\ABab$ and $\ABba$ represents a particular type of path
($\AA$, $\BB$ or $\AB$) with a particular structure of runs ($\no$, $\ma$,
$\mb$, $\mab$ or $\mba$). By convention, an $\AB$-path is always read from $A$
to $B$. These notations were adopted due to the observation that, besides the DCJ
type of the recombination (optimal, neutral or counter-optimal), the only
properties that matter are whether the paths have an odd or an even number
of runs and whether the first run is in genome~$A$ or in genome~$B$~\citep{BRA-WIL-STO-2011}. An example of a deducting recombination is given in Figure~\ref{fig:recomb}. 

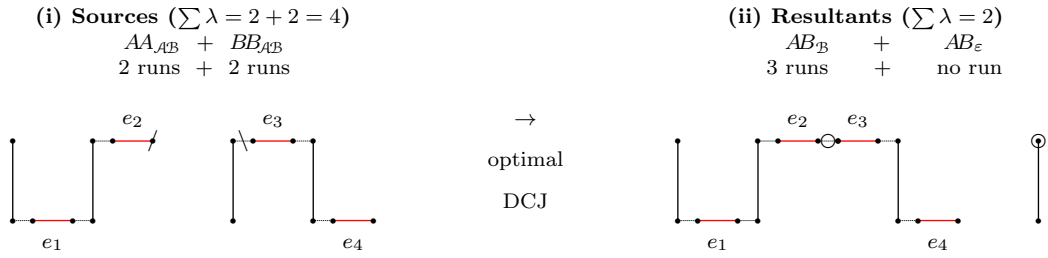
\begin{figure}
\begin{center}

\scriptsize
\setlength{\unitlength}{.75pt}
\begin{tabular}{ccc}

{\bf (i) Sources ($\sum\lambda = 2+2=4$)} &&~~{\bf (ii) Resultants ($\sum\lambda = 2$)}\\
{~~~$\AAab$~~+~~$\BBab$} & & {~~~~~~~$\ABbb$~~~~~\,+~~~~~~~$\ABno$}\\
{~~~2 runs~~+~~2 runs} & & {~~~~~~~3 runs~~~~~~+~~~~~~no run}\\
\begin{picture}(200,100)

    \drawline(10,70)(10,30)
    \dottedline{1}(10,30)(20,30)
    \dottedline{1}(40,30)(50,30)
    \drawline(50,70)(50,30)
    \dottedline{1}(50,70)(60,70)
        {\color{red}
          \drawline(60,70)(80,70)
          \drawline(130,70)(150,70)
          \drawline(20,30)(40,30)
          \drawline(170,30)(190,30)
        }
        
        \drawline(120,70)(120,30)
        \dottedline{1}(120,70)(130,70)
        \dottedline{1}(150,70)(160,70)
        \drawline(160,70)(160,30)
        \dottedline{1}(160,30)(170,30)
        
        \put( 10,70){\circle*{3}}
        \put( 50,70){\circle*{3}}
        \put( 60,70){\circle*{3}}
        \put( 80,70){\circle*{3}}
        \put( 70,80){\makebox(0,0){$e_2$}}
        \put( 80,70){\makebox(0,0){$/$}}
        
        \put(120,70){\circle*{3}}
        \put(130,70){\circle*{3}}
        \put(150,70){\circle*{3}}
        \put(160,70){\circle*{3}}
        \put(140,80){\makebox(0,0){$e_3$}}
        \put(125,70){\makebox(0,0){$\backslash$}}
        
        \put( 10,30){\circle*{3}}
        \put( 20,30){\circle*{3}}
        \put( 40,30){\circle*{3}}
        \put( 50,30){\circle*{3}}
        \put( 30,18){\makebox(0,0){$e_1$}}
        \put(120,30){\circle*{3}}
        \put(160,30){\circle*{3}}
        \put(170,30){\circle*{3}}
        \put(190,30){\circle*{3}}
        \put(180,18){\makebox(0,0){$e_4$}}
        
\end{picture} &
\begin{picture}(100,100)
     \put(50,80){\makebox(0,0){$\rightarrow$}}
     \put(50,60){\makebox(0,0){optimal}}
     \put(50,40){\makebox(0,0){DCJ}}
    \end{picture} &
    \begin{picture}(200,100)
    \drawline(10,70)(10,30)
    \dottedline{1}(10,30)(20,30)
    \dottedline{1}(40,30)(50,30)
    \drawline(50,70)(50,30)
    \dottedline{1}(50,70)(60,70)
    \dottedline{1}(80,70)(90,70)
        {\color{red}
        \drawline(20,30)(40,30)
          \drawline(60,70)(80,70)
          \drawline(90,70)(110,70)
          \drawline(130,30)(150,30)
        }
        \dottedline{1}(110,70)(120,70)
        \drawline(120,70)(120,30)
        \dottedline{1}(120,30)(130,30)

        \put( 10,70){\circle*{3}}
        \put( 50,70){\circle*{3}}
        \put( 60,70){\circle*{3}}
        \put( 80,70){\circle*{3}}
        \put( 85,70){\circle{7}}
        \put( 90,70){\circle*{3}}
        \put( 70,80){\makebox(0,0){$e_2$}}
        \put( 90,70){\circle*{3}}
        \put(110,70){\circle*{3}}
        \put(120,70){\circle*{3}}
        \put(100,80){\makebox(0,0){$e_3$}}
        \put( 10,30){\circle*{3}}
        \put( 20,30){\circle*{3}}
        \put( 40,30){\circle*{3}}
        \put( 50,30){\circle*{3}}
        \put( 30,18){\makebox(0,0){$e_1$}}
        
        \put(120,30){\circle*{3}}
        \put(130,30){\circle*{3}}
        \put(150,30){\circle*{3}}
        \put(140,18){\makebox(0,0){$e_4$}}
        
        \drawline(190,70)(190,30)
        \put(190,70){\circle*{3}}
        \put(190,70){\circle{7}}
        \put(190,30){\circle*{3}}
     \end{picture}
\end{tabular}
\end{center}
\vspace{-7mm}
\caption{\label{fig:recomb}An optimal recombination with $\Delta d=\Delta\lambda=-2$.}
\end{figure}

The complete set of path recombinations with $\Delta d \leq -1$ is given in Table~\ref{tab:recomb}. 
In Table~\ref{tab:recomb-0} we also list recombinations with $\Delta d =0$ that create at least one source of recombinations of Table~\ref{tab:recomb}.
We denote by $\ABend$ an $\AB$-path that can not be a source of a recombination in Tables~\ref{tab:recomb} and ~\ref{tab:recomb-0}, such as $\ABno$, $\ABaa$ and $\ABbb$.

\medskip

\begin{table}
\caption{Path recombinations that have $\Delta d \leq -1$ and allow the best
reuse of the resultants.}
\label{tab:recomb}
\begin{center}
\scriptsize
\setlength{\tabcolsep}{2pt}
\begin{tabular}{ccc@{~~~~~~~~}ccc@{~~~~~~~~}r@{~~~~}c@{~~~~}r}
\hline
\multicolumn{3}{l}{\bf ~~~~~sources} & \multicolumn{3}{l}{\bf resultants}
& \boldmath$\Delta \lambda$ & \boldmath$\|\rho\|$ & \boldmath$\Delta d$ \\
\hline
\hline
$\AAab$ &+ &$\BBab$ & $\ABend$ &+ &$\ABend$
& $-2$ & $0$ & $-2$\\ 
\hline
\hline
$\AAab$ &+ &$\BBaa$ & $\ABend$ &+ &$\ABba$
& $-1$ & $0$ & $-1$\\  
$\AAab$ &+ &$\BBbb$ & $\ABend$ &+ &$\ABab$
& $-1$ & $0$ & $-1$\\ 
\hline
\hline
$\AAaa$ &+ &$\BBab$ & $\ABend$ &+ &$\ABab$
& $-1$ & $0$ & $-1$\\ 
$\AAbb$ &+ &$\BBab$ & $\ABend$ &+ &$\ABba$
& $-1$ & $0$ & $-1$\\ 
\hline
\hline
$\AAaa$ &+ &$\BBaa$ & $\ABend$ &+ &$\ABend$
& $-1$ & $0$ & $-1$\\
$\AAbb$ &+ &$\BBbb$ & $\ABend$ &+ &$\ABend$
& $-1$ & $0$ & $-1$\\
\hline
\end{tabular} 
\hspace{8mm}
\setlength{\tabcolsep}{2pt}
\begin{tabular}{ccc@{~~~~~~~~}ccc@{~~~~~~~~}r@{~~~~}c@{~~~~}r}
\hline
\multicolumn{3}{l}{\bf ~~~~~sources} & \multicolumn{3}{l}{\bf resultants}
& \boldmath$\Delta \lambda$ & \boldmath$\|\rho\|$ & \boldmath$\Delta d$ \\
\hline
\hline
$\AAab$ &+ &$\AAab$ & $\AAaa$ &+ &$\AAbb$
& $-2$ & $+1$ & $-1$\\ 
$\BBab$ &+ &$\BBab$ & $\BBaa$ &+ &$\BBbb$
& $-2$ & $+1$ & $-1$\\
\hline 
\hline
$\AAab$ &+ &$\ABab$ & $\ABend$ &+ &$\AAaa$
& $-2$ & $+1$ & $-1$\\ 
$\AAab$ &+ &$\ABba$ & $\ABend$ &+ &$\AAbb$
& $-2$ & $+1$ & $-1$\\
\hline
\hline
$\BBab$ &+ &$\ABab$ & $\ABend$ &+ &$\BBbb$
& $-2$ & $+1$ & $-1$\\
$\BBab$ &+ &$\ABba$ & $\ABend$ &+ &$\BBaa$
& $-2$ & $+1$ & $-1$\\
\hline
\hline
$\ABab$ &+ &$\ABba$ & $\ABend$ &+ &$\ABend$
& $-2$ & $+1$ & $-1$\\ 
\hline
\end{tabular}
\end{center}
\end{table}

\begin{table}
\caption{Path recombinations with $\Delta d =0$ creating 
resultants that can be used in recombinations with $\Delta d \leq -1$.}
\label{tab:recomb-0}
\begin{center}
\scriptsize
\setlength{\tabcolsep}{2pt}
\begin{tabular}{ccc@{~~~~~~~~}ccc@{~~~~~~~~}r@{~~~~}c@{~~~~}r}
\hline
\multicolumn{3}{l}{\bf ~~~~sources} & \multicolumn{3}{l}{\bf resultants}
& \boldmath$\Delta \lambda$ & \boldmath$\|\rho\|$ & \boldmath$\Delta d$ \\
\hline
\hline
$\AAaa$ &+ &$\ABba$ & $\ABend$ &+ &$\AAab$
& $-1$ & $+1$ & $0$\\ 
$\AAbb$ &+ &$\ABab$ & $\ABend$ &+ &$\AAab$
& $-1$ & $+1$ & $0$\\ 
\hline
\hline
$\BBaa$ &+ &$\ABab$ & $\ABend$ &+ &$\BBab$
& $-1$ & $+1$ & $0$\\ 
$\BBbb$ &+ &$\ABba$ & $\ABend$ &+ &$\BBab$
& $-1$ & $+1$ & $0$\\ 
\hline
\end{tabular}
\hspace{8mm}
\setlength{\tabcolsep}{2pt}
\begin{tabular}{ccc@{~~~~~~~~}ccc@{~~~~~~~~}r@{~~~~}c@{~~~~}r}
\hline
\multicolumn{3}{l}{\bf ~~~~sources} & \multicolumn{3}{l}{\bf resultants}
& \boldmath$\Delta \lambda$ & \boldmath$\|\rho\|$ & \boldmath$\Delta d$ \\
\hline
\hline
$\AAaa$ &+ &$\BBbb$ & $\ABend$ &+ &$\ABab$
& $0$ & $0$ & $0$\\ 
$\AAbb$ &+ &$\BBaa$ & $\ABend$ &+ &$\ABba$
& $0$ & $0$ & $0$\\ 
\hline
\hline
$\ABab$ &+ &$\ABab$ & $\AAaa$ &+ &$\BBbb$
& $-2$ & $+2$ & $0$\\ 
$\ABba$ &+ &$\ABba$ & $\AAbb$ &+ &$\BBaa$
& $-2$ & $+2$ & $0$\\ 
\hline
\end{tabular}
\end{center}
\end{table}

The two sources of a recombination can also be called {\em partners}.
Looking at Table~\ref{tab:recomb} we observe that all partners of $\ABab$- and $\ABba$-paths are also partners of $\AAab$- and $\BBab$-paths, all partners of $\AAaa$- and $\AAbb$-paths are also partners of $\AAab$-paths and all partners of $\BBaa$- and $\BBbb$-paths are also partners of $\BBab$-paths.
Moreover, in some cases deducting recombinations are \emph{chained}, that is, resultants from deducting recombinations in Tables~\ref{tab:recomb} and~\ref{tab:recomb-0} are sources of other deducting recombinations, as shown in Figure~\ref{fig:chained-recomb}.
These observations allow the identification of groups of chained recombinations, as listed in Table~\ref{tab:recomb-group}.  

\begin{figure}
\begin{center}

\tiny
\setlength{\unitlength}{.55pt}
\begin{tabular}{cc}
\begin{tabular}{cccc}

~~~{~~~$\AAab$~~+~~$\BBaa$} & & {$\ABno$~~~~+~~~$\ABba$~~~~~~~}\\
~~~{~~~2 runs~~~+~~1 run} & & {no run~+~~~2 runs~~~~~~~}\\
\begin{picture}(140,100)
 
        \put(0,0){\color{red}
          \drawline(50,70)(70,70)
          \drawline(110,70)(130,70)
          \drawline(10,30)(30,30)
        }

        \dottedline{1}(0,30)(10,30)
        \dottedline{1}(30,30)(40,30)
        \drawline(0,70)(0,30)
        \drawline(40,70)(40,30)
        \dottedline{1}(40,70)(50,70)
        \drawline(100,70)(100,30)
        \dottedline{1}(100,70)(110,70)
        \dottedline{1}(130,70)(140,70)
        \drawline(140,70)(140,30)
        
        \put(  0,70){\circle*{3}}
        \put( 40,70){\circle*{3}}
        \put( 50,70){\circle*{3}}
        \put( 70,70){\circle*{3}}
        \put( 60,80){\makebox(0,0){$e_2$}}
        \put( 70,70){\makebox(0,0){$/$}}
        
        \put(100,70){\circle*{3}}
        \put(110,70){\circle*{3}}
        \put(130,70){\circle*{3}}
        \put(140,70){\circle*{3}}
        \put(120,80){\makebox(0,0){$e_3$}}
        \put(105,70){\makebox(0,0){$\backslash$}}
        
        \put(  0,30){\circle*{3}}
        \put( 10,30){\circle*{3}}
        \put( 30,30){\circle*{3}}
        \put( 40,30){\circle*{3}}
        \put( 20,18){\makebox(0,0){$e_1$}}
        \put(100,30){\circle*{3}}
        \put(140,30){\circle*{3}}
\end{picture} &
\begin{picture}(60,100)
     \put(30,80){\makebox(0,0){$\rightarrow$}}
     \put(30,60){\makebox(0,0){optimal}}
     \put(30,40){\makebox(0,0){DCJ}}
    \end{picture} &
    \begin{picture}(160,100)
        \drawline(20,70)(20,30)
        \put(20,70){\circle*{3}}
        \put(20,70){\circle{7}}
        \put(20,30){\circle*{3}}
        
        \put(0,0){\color{red}
          \drawline(100,70)(120,70)
          \drawline(130,70)(150,70)
          \drawline(60,30)(80,30)
        }
        
        \drawline(50,70)(50,30)
        \dottedline{1}(50,30)(60,30)
        \dottedline{1}(80,30)(90,30)
        \drawline(90,70)(90,30)
        \dottedline{1}(90,70)(100,70)
        \dottedline{1}(120,70)(130,70)
        \dottedline{1}(150,70)(160,70)
        \drawline(160,70)(160,30)
        
        \put( 50,70){\circle*{3}}
        \put( 90,70){\circle*{3}}
        \put(100,70){\circle*{3}}
        \put(110,80){\makebox(0,0){$e_2$}}
        \put(120,70){\circle*{3}}
        \put(130,70){\circle*{3}}
        \put(125,70){\circle{7}}
        \put(140,80){\makebox(0,0){$e_3$}}
        \put(150,70){\circle*{3}}
        \put(160,70){\circle*{3}}
        \put(155,70){\makebox(0,0){$/$}}
        
        \put( 50,30){\circle*{3}}
        \put( 60,30){\circle*{3}}
        \put( 80,30){\circle*{3}}
        \put( 90,30){\circle*{3}}
        \put( 70,18){\makebox(0,0){$e_1$}}
        \put(160,30){\circle*{3}}
        
     \end{picture}&
\begin{picture}(80,100)
     \put(40,40){\makebox(0,0){$\searrow$}}
    \end{picture}\\
     {~~~$\AAab$~~+~~$\BBbb$} & & {$\ABno$~~~~+~~~$\ABab$~~~~~~~}&~~~neutral~$ _\rightarrow$\\
{~~~2 runs~~+~~1 run} & & {no run~+~~~2 runs~~~~~~~}&DCJ~~\\
     \begin{picture}(140,100)
     \put(0,0){\color{red}
          \drawline(50,70)(70,70)
          \drawline(10,30)(30,30)
        }

        \dottedline{1}(0,30)(10,30)
        \dottedline{1}(30,30)(40,30)
        \drawline(0,70)(0,30)
        \drawline(40,70)(40,30)
        \dottedline{1}(40,70)(50,70)
        \drawline(100,70)(100,30)
        \dottedline{1}(100,70)(110,70)
        \drawline(110,70)(110,30)
        \dottedline{1}(110,30)(120,30)
        
        \put(  0,70){\circle*{3}}
        \put( 40,70){\circle*{3}}
        \put( 50,70){\circle*{3}}
        \put( 70,70){\circle*{3}}
        \put( 60,80){\makebox(0,0){$e_5$}}
        \put(  0,70){\makebox(0,0){$/$}}
        
        \put(100,70){\circle*{3}}
        \put(110,70){\circle*{3}}
        \put(105,70){\makebox(0,0){$\backslash$}}
        
        \put(  0,30){\circle*{3}}
        \put( 10,30){\circle*{3}}
        \put( 30,30){\circle*{3}}
        \put( 40,30){\circle*{3}}
        \put( 20,18){\makebox(0,0){$e_4$}}
        \put(100,30){\circle*{3}}
        \put(110,30){\circle*{3}}
        \put(120,30){\circle*{3}}
        \put(140,30){\circle*{3}}
     
        \put(0,0){\color{red}
          \drawline(120,30)(140,30)
        }
        
        \put(130,18){\makebox(0,0){$e_6$}}
        
\end{picture} &
\begin{picture}(60,100)
     \put(30,80){\makebox(0,0){$\rightarrow$}}
     \put(30,60){\makebox(0,0){optimal}}
     \put(30,40){\makebox(0,0){DCJ~}}
    \end{picture} &
    \begin{picture}(160,100)
        \drawline(20,70)(20,30)
        \put(20,70){\circle*{3}}
        \put(20,70){\circle{7}}
        \put(20,30){\circle*{3}}
    
        \put(0,0){\color{red}
          \drawline(50,70)(70,70)
          \drawline(90,30)(110,30)
          \drawline(140,30)(160,30)
        }
        
        \dottedline{1}(70,70)(80,70)
        \drawline(80,70)(80,30)
        \dottedline{1}(80,30)(90,30)
        \dottedline{1}(110,30)(120,30)
        \drawline(120,70)(120,30)
        \dottedline{1}(120,70)(130,70)
        \drawline(130,70)(130,30)
        \dottedline{1}(130,30)(140,30)
        
        \put( 50,70){\circle*{3}}
        \put( 70,70){\circle*{3}}
        \put( 80,70){\circle*{3}}
        \put( 60,80){\makebox(0,0){$e_5$}}
        \put( 50,70){\makebox(0,0){$\backslash$}}
        \put(120,70){\circle*{3}}
        \put(130,70){\circle*{3}}
        \put(125,70){\circle{7}}
        
        \put( 80,30){\circle*{3}}
        \put( 90,30){\circle*{3}}
        \put(110,30){\circle*{3}}
        \put(120,30){\circle*{3}}
        \put(100,18){\makebox(0,0){$e_4$}}
        \put(130,30){\circle*{3}}
        \put(140,30){\circle*{3}}
        \put(160,30){\circle*{3}}
        \put(150,18){\makebox(0,0){$e_6$}}

     \end{picture}&
\begin{picture}(80,100)
     \put(40,80){\makebox(0,0){$\nearrow$}}
    \end{picture}
\end{tabular} &
\begin{tabular}{c}
\hfill\\[1mm]
{$\ABno$~~~~~+~~~~~~$\ABbb$~~~~~~~~~~~~~~~~~~~~~~}\\
{no run~~~~+~~~~~~3 runs~~~~~~~~~~~~~~~~~~~~~~}\\
    \begin{picture}(270,100)
        \drawline(20,70)(20,30)
        \put(20,70){\circle*{3}}
        \put(20,70){\circle{7}}
        \put(20,30){\circle*{3}}
    
        \put(0,0){\color{red}
          \drawline(100,70)(120,70)
          \drawline(130,70)(150,70)
          \drawline(60,30)(80,30)
        }
        
        \drawline(50,70)(50,30)
        \dottedline{1}(50,30)(60,30)
        \dottedline{1}(80,30)(90,30)
        \drawline(90,70)(90,30)
        \dottedline{1}(90,70)(100,70)
        \dottedline{1}(120,70)(130,70)
        \dottedline{1}(150,70)(160,70)
        
        \put( 50,70){\circle*{3}}
        \put( 90,70){\circle*{3}}
        \put(100,70){\circle*{3}}
        \put(110,80){\makebox(0,0){$e_2$}}
        \put(120,70){\circle*{3}}
        \put(130,70){\circle*{3}}
        \put(140,80){\makebox(0,0){$e_3$}}
        \put(150,70){\circle*{3}}
        \put(155,70){\circle{7}}
        
        \put( 50,30){\circle*{3}}
        \put( 60,30){\circle*{3}}
        \put( 80,30){\circle*{3}}
        \put( 90,30){\circle*{3}}
        \put( 70,18){\makebox(0,0){$e_1$}}
    
        \put(0,0){\color{red}
          \drawline(160,70)(180,70)
          \drawline(200,30)(220,30)
          \drawline(250,30)(270,30)
        }
        
        \dottedline{1}(180,70)(190,70)
        \drawline(190,70)(190,30)
        \dottedline{1}(190,30)(200,30)
        \dottedline{1}(220,30)(230,30)
        \drawline(230,70)(230,30)
        \dottedline{1}(230,70)(240,70)
        \drawline(240,70)(240,30)
        \dottedline{1}(240,30)(250,30)
        
        \put(160,70){\circle*{3}}
        \put(180,70){\circle*{3}}
        \put(190,70){\circle*{3}}
        \put(170,80){\makebox(0,0){$e_5$}}
        \put(230,70){\circle*{3}}
        \put(240,70){\circle*{3}}
        
        \put(190,30){\circle*{3}}
        \put(200,30){\circle*{3}}
        \put(220,30){\circle*{3}}
        \put(230,30){\circle*{3}}
        \put(210,18){\makebox(0,0){$e_4$}}
        \put(240,30){\circle*{3}}
        \put(250,30){\circle*{3}}
        \put(270,30){\circle*{3}}
        \put(260,18){\makebox(0,0){$e_6$}}
        
     \end{picture}
\end{tabular}
\end{tabular}
\end{center}
\caption{\label{fig:chained-recomb}Chained recombinations transforming four paths ($2\times \AAab + \BBaa + \BBbb$) into four other paths ($3 \times \ABno + \ABbb$) with overall $\Delta d=-3$.}
\end{figure}
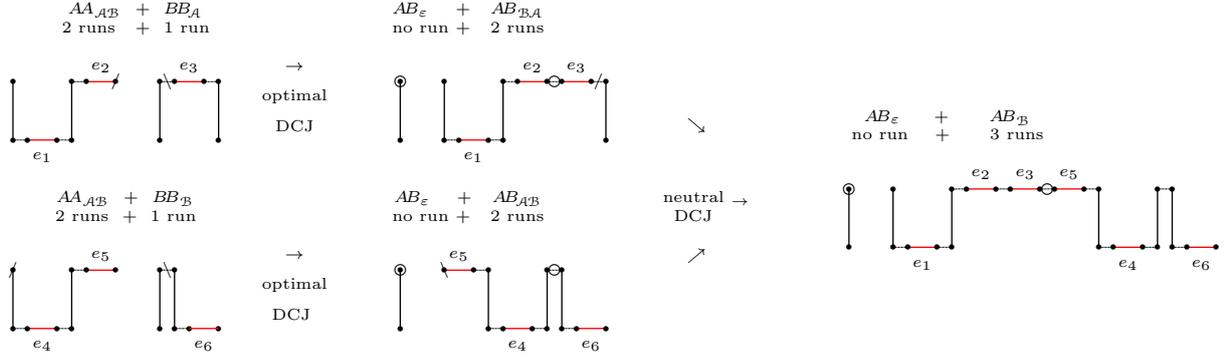

Each group is represented by a combination of letters, where:
\begin{itemize}
\item $\W$ represents an $\AAab$, $\WA$ represents an $\AAaa$ and $\WB$ represents an $\AAbb$;
\item $\M$ represents a $\BBab$, $\MA$ represents a $\BBaa$ and $\MB$ represents a $\BBbb$;
\item$\Z$ represents an $\ABab$ and $\N$ represents an $\ABba$.
\end{itemize}
  
Although some groups have reusable resultants, those are actually never reused. 
(If groups that are lower in the table use as sources resultants from higher groups, the sources of all referred groups would be previously consumed in groups that occupy even higher positions in the table.)
Due to this fact, the number of occurrences in each group depends only on the initial number of each type of component. 
 
\begin{table}
\caption{Chained recombination groups obtained from Tables~\ref{tab:recomb} and~\ref{tab:recomb-0}. 
The column {\bf scr} indicates the contribution of each path
to the distance decrease (the table is sorted in descending order with respect
to this column).}
\label{tab:recomb-group}
\begin{center}
\scriptsize
\setlength{\tabcolsep}{4pt}
\begin{tabular}[t]{|lc|ccccc|ccccccr|r|r|}
\hline
& {\bf id} & \multicolumn{5}{c}{\bf sources} & \multicolumn{7}{c}{\bf resultants} & \boldmath$\Delta d$ & {\bf scr}\\
\hline
\hline
$\mathcal P$ 
& $\W\M$ & $\AAab$ &~&$\BBab$ &&-----&-----&&-----&&-----&& $2\times\ABend$ & $-2$ & $-1$\\
\hline
\hline
& & & & && & && && &&\\[-3mm]
$\mathcal Q$ 
& $\W\W\MA\MB$ & $2\times\AAab$ &~&$\BBaa\!+\!\BBbb$ & &-----&-----&&-----&&-----&& $4\times\ABend$ & $-3$ & $-3/4$\\[1mm]
& $\M\M\WA\WB$ &  $\AAaa\!+\!\AAbb$ &~& $2\times\BBab$ &&----- &-----&&-----&&-----&& $4\times\ABend$ & $-3$ & $-3/4$\\
\hline
\hline
& & & & && & && && &&\\[-3mm]
$\mathcal T$ 
& $\W\Z\MA$ & $\AAab$ &~&$\BBaa$ &~&$\ABab$ &-----&&-----&&-----&& $3\times\ABend$ & $-2$ & $-2/3$\\
& $\W\W\MA$& $2\times\AAab$ &~&$\BBaa$ &&-----&$\AAbb$ &&-----&&-----&~& $2\times\ABend$  & $-2$ & $-2/3$\\[1mm]
& $\W\N\MB$ & $\AAab$ &~&$\BBbb$ &~&$\ABba$ &-----&&-----&&-----&& $3\times\ABend$ & $-2$ & $-2/3$\\
& $\W\W\MB$& $2\times\AAab$ &~&$\BBbb$ &&----- &$\AAaa$ &&-----&&-----&~& $2\times\ABend$  & $-2$ & $-2/3$\\[1mm]
& $\M\N\WA$ & $\AAaa$ &~& $\BBab$ &~&$\ABba$ &-----&&-----&&-----&& $3\times\ABend$ & $-2$ & $-2/3$\\
& $\M\M\WA$& $\AAaa$ &~& $2\times\BBab$&&-----&-----&&$\BBbb$&&-----&~& $2\times\ABend$ & $-2$ & $-2/3$\\[1mm]
& $\M\Z\WB$ & $\AAbb$ &~&$\BBab$ &~&$\ABab$ &-----&&-----&&-----&& $3\times\ABend$ & $-2$ & $-2/3$\\
& $\M\M\WB$& $\AAbb$ &~&$2\times\BBab$ &&-----&-----&&$\BBaa$&&-----&~& $2\times\ABend$ & $-2$ & $-2/3$\\[1mm]
\hline
\hline
$\mathcal S$ 
& $\Z\N$ & -----&&-----&&$\ABab\!+\!\ABba$ &-----&&-----&&-----&&  $2\times\ABend$ & $-1$ & $-1/2$\\[1mm]
& $\WA\MA$ & $\AAaa$ &~&$\BBaa$ &&-----&-----&&-----&&-----&&  $2\times\ABend$ & $-1$ & $-1/2$\\
& $\WB\MB$ & $\AAbb$ &~&$\BBbb$ &&-----&-----&&-----&&-----&& $2\times\ABend$ & $-1$ & $-1/2$\\[1mm]
&$\W\MA$ & $\AAab$ &~&$\BBaa$ &&-----&-----&&-----&& $\ABba$&~&$\ABend$ & $-1$ & $-1/2$\\[1mm]
&$\W\MB$ & $\AAab$ &~&$\BBbb$ &&-----&-----&&-----&& $\ABab$&~&$\ABend$ & $-1$ & $-1/2$\\[1mm]
&$\W\Z$ & $\AAab$ &&-----&~&$\ABab$ & $\AAaa$ &&-----&&-----&~& $\ABend$ & $-1$ & $-1/2$\\[1mm]
&$\W\N$ & $\AAab$ &&-----&~&$\ABba$ & $\AAbb$ &&-----&&-----&~& $\ABend$ & $-1$ & $-1/2$\\[1mm] 
& $\W\W$& $2\times\AAab$ &&-----&&-----& $\AAbb\!+\!\AAaa$&&-----&&-----&&----- & $-1$ & $-1/2$\\[1mm]
& $\M\WA$ & $\AAaa$ &~& $\BBab$&&-----&-----&&-----&& $\ABab$&~&$\ABend$ & $-1$ & $-1/2$\\[1mm] 
& $\M\WB$ & $\AAbb$ &~& $\BBab$&&-----&-----&&-----&& $\ABba$&~&$\ABend$ & $-1$ & $-1/2$\\[1mm]
& $\M\Z$ & -----&& $\BBab$ &~&$\ABab$ &  -----&&$\BBbb$&&-----&~& $\ABend$ & $-1$ & $-1/2$\\[1mm]
& $\M\N$ & -----&& $\BBab$ &~&$\ABba$ &  -----&&$\BBaa$&&-----&~& $\ABend$ & $-1$ & $-1/2$\\[1mm]
& $\M\M$&  -----&& $2\times\BBab$ &&-----&-----&& $\BBbb\!+\!\BBaa$ &&-----&&-----& $-1$ & $-1/2$\\
\hline
\hline
& & & & && & && && &&\\[-3mm]
$\mathcal M$ 
& $\Z\Z\WB\MA$ & $\AAbb$ &~&$\BBaa$ &~&$2\times\ABab$ &-----&&-----&&-----&& $4\times\ABend$ & $-2$ & $-1/2$\\[1mm]
& $\N\N\WA\MB$ & $\AAaa$ &~&$\BBbb$ &~&$2\times\ABba$ &-----&&-----&&-----&& $4\times\ABend$ & $-2$ & $-1/2$\\
\hline
\hline
& & & & && & && && &&\\[-3mm]
$\mathcal  N$ 
& $\Z\WB\MA$ & $\AAbb$ &~&$\BBaa$ &~&$\ABab$ &-----&&-----&& $\ABba$&~&$2\times\ABend$ & $-1$ & $-1/3$\\[1mm]
& $\Z\Z\WB$& $\AAbb$ &&-----&~& $2\times\ABab$ & $\AAaa$ &&-----&&-----&~& $2\times\ABend$ & $-1$ & $-1/3$\\[1mm]
& $\Z\Z\MA$& ----- &&$\BBaa$ &~&$2\times\ABab$ & -----&&$\BBbb$&&-----&~& $2\times\ABend$ & $-1$ & $-1/3$\\[1mm]
&$\N\WA\MB$ & $\AAaa$ &~&$\BBbb$ &~&$\ABba$ &-----&&-----&& $\ABab$&~&$2\times\ABend$ & $-1$ & $-1/3$\\[1mm]
& $\N\N\WA$ & $\AAaa$ &&----- &~&$2\times\ABba$ &$\AAbb$ &&-----&&-----&~& $2\times\ABend$ & $-1$ & $-1/3$\\[1mm]
& $\N\N\MB$& ----- &&$\BBbb$ &~&$2\times\ABba$ &-----&&$\BBaa$&&-----&~& $2\times\ABend$ & $-1$ & $-1/3$\\
\hline
\end{tabular}
\end{center}
\end{table}

The deductions shown in Table~\ref{tab:recomb-group} can be computed with an approach that greedily maximizes the groups in $\mathcal P$, $\mathcal Q$, $\mathcal T$, $\mathcal S$, $\mathcal M$ and $\mathcal N$ in this order.
The $\mathcal P$ part contains only one operation and is thus very simple.
The same happens with $\mathcal Q$, since the two groups in this part are exclusive
after applying $\mathcal P$. The four subparts of $\mathcal T$ are also exclusive after applying $\mathcal Q$. (Note that groups $\W\W\MA$, $\W\W\MB$, $\M\M\WA$ and $\M\M\WB$ of $\mathcal T$ are simply subgroups of $\mathcal Q$.) The groups in $\mathcal S$ correspond to the simple application of all possible
remaining operations with $\Delta d=-1$. After applying operations of type $\Z\N$, $\WA\MA$ and $\WB\MB$, the remaining operations in $\mathcal S$ are all exclusive. After $\mathcal S$, the two groups
in $\mathcal M$ are exclusive and then the same happens to the six groups in $\mathcal N$ (that are simply  subgroups of $\mathcal M$).

We can now write the theorem that gives the exact formula for the DCJ-indel distance of linear singular genomes:

\begin{theorem}[from~\citep{BRA-WIL-STO-2011}]\label{thm:exact-distance}
Given two singular linear
genomes $A$ and $B$, whose relational diagram $R(A,B)$ has $c$ $\AB$-cycles and $i$ $\AB$-paths, we have
\[\ddcjid(A,B) = |\mg| -c -\frac{i}{2} +\!\!\!\sum_{C \in R(A,B)}\!\!\! \lambda(C) - \delta,\] 
where $\delta = 2\mathcal P + 3\mathcal Q +2\mathcal T + \mathcal S + 2\mathcal M + \mathcal N$ and $\mathcal P$, $\mathcal Q$, $\mathcal T$, $\mathcal S$, $\mathcal M$ and $\mathcal N$ here refer to the number of deductions in the corresponding chained recombination groups.
\end{theorem}

\section{DCJ-indel distance of natural genomes}

Based on the results presented so far, we develop an approach for computing the DCJ-indel distance of natural genomes $A$ and $B$.
First we note that it is possible to transform $A$ and $B$ into \emph{matched} singular genomes $\AM$ and $\BM$ as follows.
For each common marker $\mk \in \mg$, if $\occof{A} \leq \occof{B}$, we should determine which occurrence of~$\mk$ in~$B$ matches each occurrence of~$\mk$ in~$A$, or if $\occof{B} < \occof{A}$, which occurrence of~$\mk$ in~$A$ matches each occurrence of~$\mk$ in~$B$.
The matched occurrences receive the same identifier (for example, by adding the same \emph{index}) in~$\AM$ and in~$\BM$.
Examples are given in Figure~\ref{fig:multi-adj-graph} (top and center).
Observe that, after this procedure, the number of common markers between any pair of matched genomes $\AM$ and $\BM$ is
\[
n_*=\sum_{\mk \in \mg} \min\{\occof{A}(\mk),\occof{B}(\mk)\}\,.
\]

Let $\mathbb{M}$ be the set of all possible pairs of matched singular genomes obtained from natural genomes $A$ and~$B$. The DCJ-indel distance of $A$ and $B$ is then defined as
\begin{equation*}
\label{eq:dcj-indel}
\ddcjid(A, B) = \min_{(\AM,\BM) \in \mathbb{M}}\{ \ddcjid(\AM,\BM) \} \,.
\end{equation*}

\begin{figure}
  \begin{center}
    \tiny
    \setlength{\unitlength}{.65pt}
    \begin{tabular}{c}
    \begin{picture}(500,100)

      {\color{purple}
      \dottedline[$\cdot$]{2}(65,70)(80,70)
      \dottedline[$\cdot$]{2}(145,70)(160,70)
      \dottedline[$\cdot$]{2}(185,70)(200,70)
      \dottedline[$\cdot$]{2}(225,70)(240,70)
      \dottedline[$\cdot$]{2}(265,70)(280,70)
      \dottedline[$\cdot$]{2}(305,70)(320,70)
      
      \dottedline[$\cdot$]{2}(105,10)(120,10)
      \dottedline[$\cdot$]{2}(185,10)(200,10)
      \dottedline[$\cdot$]{2}(225,10)(240,10)
      \dottedline[$\cdot$]{2}(265,10)(280,10)
      \dottedline[$\cdot$]{2}(305,10)(320,10)
      \dottedline[$\cdot$]{2}(345,10)(360,10)
      
      \drawline(80,70)(120,10)
      \qbezier(65,70)(125,40)(185,10)
      \qbezier(145,70)(125,40)(105,10)
      
      \drawline(160,70)(185,70)
      \drawline(200,10)(225,10)
      
      \qbezier(200,70)(253,40)(305,10)
      \qbezier(225,70)(253,40)(280,10)
      \drawline(240,70)(240,10)
      \drawline(265,70)(265,10)
      
      \drawline(280,70)(305,70)
      \drawline(320,10)(345,10)
      \drawline(320,70)(345,70)
      \drawline(360,10)(385,10)
      }

      \put(0,0){\color{orange}
        \dottedline[$\cdot$]{2}(25,10)(40,10)
        \dottedline[$\cdot$]{2}(65,10)(80,10)
        \dottedline[$\cdot$]{2}(105,70)(120,70)
        \dottedline[$\cdot$]{2}(145,10)(160,10)
        
        \qbezier(40,70)(100,40)(160,10)
        \qbezier(120,70)(100,40)(80,10)
        \drawline(0,10)(25,10)
        \drawline(40,10)(65,10)
        \drawline(105,70)(145,10)
      }
      
      \put( 40,70){\circle*{5}}
      \put( 65,70){\circle*{5}}
      \put( 80,70){\circle*{5}}
      \put(105,70){\circle*{5}}
      \put(120,70){\circle*{5}}
      \put(145,70){\circle*{5}}
      \put(160,70){\circle*{5}}
      \put(185,70){\circle*{5}}
      \put(200,70){\circle*{5}}
      \put(225,70){\circle*{5}}
      \put(240,70){\circle*{5}}
      \put(265,70){\circle*{5}}
      \put(280,70){\circle*{5}}
      \put(305,70){\circle*{5}}
      \put(320,70){\circle*{5}}
      \put(345,70){\circle*{5}}

      \put( 40,82){\makebox(0,0){$\tl{\Ti{1}{1}}$}}
      \put( 65,82){\makebox(0,0){$\hd{\Ti{1}{1}}$}}
      \put( 80,82){\makebox(0,0){$\tl{\Ti{3}{1}}$}}
      \put(105,82){\makebox(0,0){$\hd{\Ti{3}{1}}$}}
      \put(120,82){\makebox(0,0){$\tl{\Ti{2}{1}}$}}
      \put(145,82){\makebox(0,0){$\hd{\Ti{2}{1}}$}}
      \put(160,82){\makebox(0,0){$\hd{\T{5}}$}}
      \put(185,82){\makebox(0,0){$\tl{\T{5}}$}}
      \put(200,82){\makebox(0,0){$\hd{\Ti{4}{1}}$}}
      \put(225,82){\makebox(0,0){$\tl{\Ti{4}{1}}$}}
      \put(240,82){\makebox(0,0){$\tl{\Ti{3}{2}}$}}
      \put(265,82){\makebox(0,0){$\hd{\Ti{3}{2}}$}}
      \put(280,82){\makebox(0,0){$\tl{\T{5}}$}}
      \put(305,82){\makebox(0,0){$\hd{\T{5}}$}}
      \put(320,82){\makebox(0,0){$\tl{\T{4}}$}}
      \put(345,82){\makebox(0,0){$\hd{\T{4}}$}}
      
      \put(  0,10){\circle*{5}}
      \put( 25,10){\circle*{5}}
      \put( 40,10){\circle*{5}}
      \put( 65,10){\circle*{5}}
      \put( 80,10){\circle*{5}}
      \put(105,10){\circle*{5}}
      \put(120,10){\circle*{5}}
      \put(145,10){\circle*{5}}
      \put(160,10){\circle*{5}}
      \put(185,10){\circle*{5}}
      \put(200,10){\circle*{5}}
      \put(225,10){\circle*{5}}
      \put(240,10){\circle*{5}}
      \put(265,10){\circle*{5}}
      \put(280,10){\circle*{5}}
      \put(305,10){\circle*{5}}
      \put(320,10){\circle*{5}}
      \put(345,10){\circle*{5}}
      \put(360,10){\circle*{5}}
      \put(385,10){\circle*{5}}

      \put( 0,-2){\makebox(0,0){$\tl{\T{1}}$}}
      \put( 25,-2){\makebox(0,0){$\hd{\T{1}}$}}
      \put( 40,-2){\makebox(0,0){$\tl{\T{6}}$}}
      \put( 65,-2){\makebox(0,0){$\hd{\T{6}}$}}
      \put( 80,-2){\makebox(0,0){$\tl{\Ti{2}{1}}$}}
      \put(105,-2){\makebox(0,0){$\hd{\Ti{2}{1}}$}}
      \put(120,-2){\makebox(0,0){$\tl{\Ti{3}{1}}$}}
      \put(145,-2){\makebox(0,0){$\hd{\Ti{3}{1}}$}}
      \put(160,-2){\makebox(0,0){$\tl{\Ti{1}{1}}$}}
      \put(185,-2){\makebox(0,0){$\hd{\Ti{1}{1}}$}}
      \put(200,-2){\makebox(0,0){$\tl{\T{7}}$}}
      \put(225,-2){\makebox(0,0){$\hd{\T{7}}$}}
      \put(240,-2){\makebox(0,0){$\tl{\Ti{3}{2}}$}}
      \put(265,-2){\makebox(0,0){$\hd{\Ti{3}{2}}$}}
      \put(280,-2){\makebox(0,0){$\tl{\Ti{4}{1}}$}}
      \put(305,-2){\makebox(0,0){$\hd{\Ti{4}{1}}$}}
      \put(320,-2){\makebox(0,0){$\tl{\T{1}}$}}
      \put(345,-2){\makebox(0,0){$\hd{\T{1}}$}}
      \put(360,-2){\makebox(0,0){$\tl{\T{3}}$}}
      \put(385,-2){\makebox(0,0){$\hd{\T{3}}$}}
      
      \put(450,40){\makebox(0,0){$\ddcjid
    =5-0-\frac{2}{2}+1+3=8$}}
    \end{picture}\\[5mm]
    \begin{picture}(500,100)

      {\color{blue}
      \dottedline[$\cdot$]{2}(65,70)(80,70)
      \dottedline[$\cdot$]{2}(105,70)(120,70)
      \dottedline[$\cdot$]{2}(265,70)(280,70)
      \dottedline[$\cdot$]{2}(305,70)(320,70)
      \dottedline[$\cdot$]{2}(25,10)(40,10)
      \dottedline[$\cdot$]{2}(65,10)(80,10)
      \dottedline[$\cdot$]{2}(145,10)(160,10)
      \dottedline[$\cdot$]{2}(185,10)(200,10)
      \dottedline[$\cdot$]{2}(225,10)(240,10)
      \dottedline[$\cdot$]{2}(265,10)(280,10)
      \qbezier(65,70)(45,40)(25,10)
      \qbezier(80,70)(160,40)(240,10)
      \drawline(40,10)(65,10)
      \qbezier(105,70)(180,40)(265,10)
      \qbezier(120,70)(100,40)(80,10)
      \drawline(160,10)(185,10)
      \drawline(200,10)(225,10)
      \qbezier(265,70)(205,40)(145,10)
      \drawline(280,70)(305,70)
      \qbezier(320,70)(300,40)(280,10)
      }

      \put(0,0){\color{cyan}
      \dottedline[$\cdot$]{2}(145,70)(160,70)
      \dottedline[$\cdot$]{2}(185,70)(200,70)
      \dottedline[$\cdot$]{2}(225,70)(240,70)
      \dottedline[$\cdot$]{2}(105,10)(120,10)
      \dottedline[$\cdot$]{2}(305,10)(320,10)
      \dottedline[$\cdot$]{2}(345,10)(360,10)
        \qbezier(40,70)(20,40)(0,10)
        \qbezier(145,70)(125,40)(105,10)
        \drawline(160,70)(185,70)
        \drawline(200,70)(225,70)
        \qbezier(240,70)(180,40)(120,10)
        
       \qbezier(345,70)(325,40)(305,10)
       \drawline(320,10)(345,10)
       \drawline(360,10)(385,10)
      }
      
      \put( 40,70){\circle*{5}}
      \put( 65,70){\circle*{5}}
      \put( 80,70){\circle*{5}}
      \put(105,70){\circle*{5}}
      \put(120,70){\circle*{5}}
      \put(145,70){\circle*{5}}
      \put(160,70){\circle*{5}}
      \put(185,70){\circle*{5}}
      \put(200,70){\circle*{5}}
      \put(225,70){\circle*{5}}
      \put(240,70){\circle*{5}}
      \put(265,70){\circle*{5}}
      \put(280,70){\circle*{5}}
      \put(305,70){\circle*{5}}
      \put(320,70){\circle*{5}}
      \put(345,70){\circle*{5}}

      \put( 40,82){\makebox(0,0){$\tl{\Ti{1}{1}}$}}
      \put( 65,82){\makebox(0,0){$\hd{\Ti{1}{1}}$}}
      \put( 80,82){\makebox(0,0){$\tl{\Ti{3}{1}}$}}
      \put(105,82){\makebox(0,0){$\hd{\Ti{3}{1}}$}}
      \put(120,82){\makebox(0,0){$\tl{\Ti{2}{1}}$}}
      \put(145,82){\makebox(0,0){$\hd{\Ti{2}{1}}$}}
      \put(160,82){\makebox(0,0){$\hd{\T{5}}$}}
      \put(185,82){\makebox(0,0){$\tl{\T{5}}$}}
      \put(200,82){\makebox(0,0){$\hd{\T{4}}$}}
      \put(225,82){\makebox(0,0){$\tl{\T{4}}$}}
      \put(240,82){\makebox(0,0){$\tl{\Ti{3}{2}}$}}
      \put(265,82){\makebox(0,0){$\hd{\Ti{3}{2}}$}}
      \put(280,82){\makebox(0,0){$\tl{\T{5}}$}}
      \put(305,82){\makebox(0,0){$\hd{\T{5}}$}}
      \put(320,82){\makebox(0,0){$\tl{\Ti{4}{1}}$}}
      \put(345,82){\makebox(0,0){$\hd{\Ti{4}{1}}$}}
      
      \put(  0,10){\circle*{5}}
      \put( 25,10){\circle*{5}}
      \put( 40,10){\circle*{5}}
      \put( 65,10){\circle*{5}}
      \put( 80,10){\circle*{5}}
      \put(105,10){\circle*{5}}
      \put(120,10){\circle*{5}}
      \put(145,10){\circle*{5}}
      \put(160,10){\circle*{5}}
      \put(185,10){\circle*{5}}
      \put(200,10){\circle*{5}}
      \put(225,10){\circle*{5}}
      \put(240,10){\circle*{5}}
      \put(265,10){\circle*{5}}
      \put(280,10){\circle*{5}}
      \put(305,10){\circle*{5}}
      \put(320,10){\circle*{5}}
      \put(345,10){\circle*{5}}
      \put(360,10){\circle*{5}}
      \put(385,10){\circle*{5}}
      
      \put( 0,-2){\makebox(0,0){$\tl{\Ti{1}{1}}$}}
      \put( 25,-2){\makebox(0,0){$\hd{\Ti{1}{1}}$}}
      \put( 40,-2){\makebox(0,0){$\tl{\T{6}}$}}
      \put( 65,-2){\makebox(0,0){$\hd{\T{6}}$}}
      \put( 80,-2){\makebox(0,0){$\tl{\Ti{2}{1}}$}}
      \put(105,-2){\makebox(0,0){$\hd{\Ti{2}{1}}$}}
      \put(120,-2){\makebox(0,0){$\tl{\Ti{3}{2}}$}}
      \put(145,-2){\makebox(0,0){$\hd{\Ti{3}{2}}$}}
      \put(160,-2){\makebox(0,0){$\tl{\T{1}}$}}
      \put(185,-2){\makebox(0,0){$\hd{\T{1}}$}}
      \put(200,-2){\makebox(0,0){$\tl{\T{7}}$}}
      \put(225,-2){\makebox(0,0){$\hd{\T{7}}$}}
      \put(240,-2){\makebox(0,0){$\tl{\Ti{3}{1}}$}}
      \put(265,-2){\makebox(0,0){$\hd{\Ti{3}{1}}$}}
      \put(280,-2){\makebox(0,0){$\tl{\Ti{4}{1}}$}}
      \put(305,-2){\makebox(0,0){$\hd{\Ti{4}{1}}$}}
      \put(320,-2){\makebox(0,0){$\tl{\T{1}}$}}
      \put(345,-2){\makebox(0,0){$\hd{\T{1}}$}}
      \put(360,-2){\makebox(0,0){$\tl{\T{3}}$}}
      \put(385,-2){\makebox(0,0){$\hd{\T{3}}$}}
      
      \put(450,40){\makebox(0,0){$\ddcjid
      =5-2-\frac{2}{2}+1+2+2=7$}}
    \end{picture}\\[10mm]
    \begin{picture}(570,120)
    
      \put(0,0){\color{cyan}
      \drawline(240,99)(280,99)
      \drawline(300,100)(340,100)
      }
     \put(0,0){\color{blue}
      \drawline(420,99)(460,99)
      }

      \put(0,0){\color{purple}
      \drawline(240,101)(280,101)
      \drawline(420,101)(460,101)
      \drawline(480,100)(520,100)
      }
      
      \put(0,0){\color{orange}
      \drawline(0,10)(40,10)
      \drawline(60,11)(100,11)
      }
      \put(0,0){\color{purple}
      \drawline(300,11)(340,11)
      \drawline(480,11)(520,11)
      \drawline(540,11)(580,11)
      }
      
      \put(0,0){\color{blue}
      \drawline(60,9)(100,9)
      \drawline(240,10)(280,10)
      \drawline(300,9)(340,9)
      }
      \put(0,0){\color{gray}
      \drawline(180,10)(220,10)
      \drawline(360,10)(400,10)
      }
      \put(0,0){\color{cyan}
      \drawline(480,9)(520,9)
      \drawline(540,9)(580,9)
      }
      
      \put(0,0){\color{gray}
      \qbezier(60,100)(270,55)(480,10)
      \qbezier(100,100)(300,52)(520,10)
      \qbezier(120,100)(340,58)(540,10)
      \qbezier(160,100)(370,55)(580,10)
      \qbezier(360,100)(450,55)(540,10)
      \qbezier(400,100)(490,55)(580,10)
      }
      \put(0,0){\color{cyan}
      \qbezier(60,100)(30,55)(0,10)
      }
      \put(0,0){\color{orange}
      \qbezier(60,100)(150,55)(240,10)
      }
      \put(0,0){\color{blue}
      \qbezier(100,100)(70,55)(40,10)
      }
      \put(0,0){\color{purple}
      \qbezier(100,100)(190,55)(280,10)
      \qbezier(120,100)(150,55)(180,10)
      }
      \put(0,0){\color{blue}
      \qbezier(120,100)(240,55)(360,10)
      }
      
      \put(0,0){\color{orange}
      \qbezier(179,100)(149,55)(119,10)
      }
      \put(0,0){\color{blue}
      \qbezier(181,100)(151,55)(121,10)
      }
      \put(0,0){\color{orange}
      \qbezier(160,100)(190,55)(220,10)
      }
      \put(0,0){\color{blue}
      \qbezier(160,100)(280,55)(400,10)
      }
      
      \put(0,0){\color{purple}
      \qbezier(219,100)(189,55)(159,10)
      }
      \put(0,0){\color{cyan}
      \qbezier(221,100)(191,55)(161,10)
      }
      \put(0,0){\color{purple}
      \qbezier(300,100)(380,55)(460,10)
      }
      \put(0,0){\color{cyan}
      \qbezier(360,100)(270,55)(180,10)
      }
      \put(0,0){\color{purple}
      \qbezier(340,100)(380,55)(420,10)
      \drawline(360,100)(360,10)
      }
      
      \put(0,0){\color{blue}
      \qbezier(400,100)(310,55)(220,10)
      }
      \put(0,0){\color{purple}
      \drawline(400,100)(400,10)
      }
      \put(0,0){\color{blue}
      \qbezier(480,100)(450,55)(420,10)
      }
      
      \put(0,0){\color{cyan}
      \qbezier(520,100)(490,55)(460,10)
      }
      
      \put(0,0){\color{purple}
      \dottedline[$\cdot$]{2}(100,101)(120,101)
      \dottedline[$\cdot$]{2}(220,101)(240,101)
      \dottedline[$\cdot$]{2}(280,101)(300,101)
      \dottedline[$\cdot$]{2}(340,101)(360,101)
      \dottedline[$\cdot$]{2}(400,101)(420,101)
      \dottedline[$\cdot$]{2}(460,101)(480,101)
      \dottedline[$\cdot$]{2}(160,11)(180,11)
      \dottedline[$\cdot$]{2}(280,11)(300,11)
      \dottedline[$\cdot$]{2}(340,11)(360,11)
      \dottedline[$\cdot$]{2}(400,11)(420,11)
      \dottedline[$\cdot$]{2}(460,11)(480,11)
      \dottedline[$\cdot$]{2}(520,11)(540,11)
      }
      \put(0,0){\color{blue}
      \dottedline[$\cdot$]{2}(100,99)(120,99)
      \dottedline[$\cdot$]{2}(160,99)(180,99)
      \dottedline[$\cdot$]{2}(400,99)(420,99)
      \dottedline[$\cdot$]{2}(460,99)(480,99)
      \dottedline[$\cdot$]{2}(40,9)(60,9)
      \dottedline[$\cdot$]{2}(100,9)(120,9)
      \dottedline[$\cdot$]{2}(220,9)(240,9)
      \dottedline[$\cdot$]{2}(280,9)(300,9)
      \dottedline[$\cdot$]{2}(340,9)(360,9)
      \dottedline[$\cdot$]{2}(400,9)(420,9)
      }
      \put(0,0){\color{orange}
      \dottedline[$\cdot$]{2}(160,101)(180,101)
      \dottedline[$\cdot$]{2}(40,11)(60,11)
      \dottedline[$\cdot$]{2}(100,11)(120,11)
      \dottedline[$\cdot$]{2}(220,11)(240,11)
      }
      \put(0,0){\color{cyan}
      \dottedline[$\cdot$]{2}(220,99)(240,99)
      \dottedline[$\cdot$]{2}(280,99)(300,99)
      \dottedline[$\cdot$]{2}(340,99)(360,99)
      \dottedline[$\cdot$]{2}(160,9)(180,9)
      \dottedline[$\cdot$]{2}(460,9)(480,9)
      \dottedline[$\cdot$]{2}(520,9)(540,9)
      }
      
      \put( 60,100){\circle*{5}}
      \put(100,100){\circle*{5}}
      \put(120,100){\circle*{5}}
      \put(160,100){\circle*{5}}
      \put(180,100){\circle*{5}}
      \put(220,100){\circle*{5}}
      \put(240,100){\circle*{5}}
      \put(280,100){\circle*{5}}
      \put(300,100){\circle*{5}}
      \put(340,100){\circle*{5}}
      \put(360,100){\circle*{5}}
      \put(400,100){\circle*{5}}
      \put(420,100){\circle*{5}}
      \put(460,100){\circle*{5}}
      \put(480,100){\circle*{5}}
      \put(520,100){\circle*{5}}

      \put(-25,112){\makebox(0,0){$A$}}
      \put( 60,112){\makebox(0,0){$\tl{\T{1}}$}}
      \put(100,112){\makebox(0,0){$\hd{\T{1}}$}}
      \put(120,112){\makebox(0,0){$\tl{\T{3}}$}}
      \put(160,112){\makebox(0,0){$\hd{\T{3}}$}}
      \put(180,112){\makebox(0,0){$\tl{\T{2}}$}}
      \put(220,112){\makebox(0,0){$\hd{\T{2}}$}}
      \put(240,112){\makebox(0,0){$\hd{\T{5}}$}}
      \put(280,112){\makebox(0,0){$\tl{\T{5}}$}}
      \put(300,112){\makebox(0,0){$\hd{\T{4}}$}}
      \put(340,112){\makebox(0,0){$\tl{\T{4}}$}}
      \put(360,112){\makebox(0,0){$\tl{\T{3}}$}}
      \put(400,112){\makebox(0,0){$\hd{\T{3}}$}}
      \put(420,112){\makebox(0,0){$\tl{\T{5}}$}}
      \put(460,112){\makebox(0,0){$\hd{\T{5}}$}}
      \put(480,112){\makebox(0,0){$\tl{\T{4}}$}}
      \put(520,112){\makebox(0,0){$\hd{\T{4}}$}}
      
      \put(  0,10){\circle*{5}}
      \put( 40,10){\circle*{5}}
      \put( 60,10){\circle*{5}}
      \put(100,10){\circle*{5}}
      \put(120,10){\circle*{5}}
      \put(160,10){\circle*{5}}
      \put(180,10){\circle*{5}}
      \put(220,10){\circle*{5}}
      \put(240,10){\circle*{5}}
      \put(280,10){\circle*{5}}
      \put(300,10){\circle*{5}}
      \put(340,10){\circle*{5}}
      \put(360,10){\circle*{5}}
      \put(400,10){\circle*{5}}
      \put(420,10){\circle*{5}}
      \put(460,10){\circle*{5}}
      \put(480,10){\circle*{5}}
      \put(520,10){\circle*{5}}
      \put(540,10){\circle*{5}}
      \put(580,10){\circle*{5}}

      \put(-25,-2){\makebox(0,0){$B$}}
      \put(  0,-2){\makebox(0,0){$\tl{\T{1}}$}}
      \put( 40,-2){\makebox(0,0){$\hd{\T{1}}$}}
      \put( 60,-2){\makebox(0,0){$\tl{\T{6}}$}}
      \put(100,-2){\makebox(0,0){$\hd{\T{6}}$}}
      \put(120,-2){\makebox(0,0){$\tl{\T{2}}$}}
      \put(160,-2){\makebox(0,0){$\hd{\T{2}}$}}
      \put(180,-2){\makebox(0,0){$\tl{\T{3}}$}}
      \put(220,-2){\makebox(0,0){$\hd{\T{3}}$}}
      \put(240,-2){\makebox(0,0){$\tl{\T{1}}$}}
      \put(280,-2){\makebox(0,0){$\hd{\T{1}}$}}
      \put(300,-2){\makebox(0,0){$\tl{\T{7}}$}}
      \put(340,-2){\makebox(0,0){$\hd{\T{7}}$}}
      \put(360,-2){\makebox(0,0){$\tl{\T{3}}$}}
      \put(400,-2){\makebox(0,0){$\hd{\T{3}}$}}
      \put(420,-2){\makebox(0,0){$\tl{\T{4}}$}}
      \put(460,-2){\makebox(0,0){$\hd{\T{4}}$}}
      \put(480,-2){\makebox(0,0){$\tl{\T{1}}$}}
      \put(520,-2){\makebox(0,0){$\hd{\T{1}}$}}
      \put(540,-2){\makebox(0,0){$\tl{\T{3}}$}}
      \put(580,-2){\makebox(0,0){$\hd{\T{3}}$}}
    \end{picture}
    \end{tabular}
  \end{center}
  \caption{\label{fig:multi-adj-graph}Natural genomes $A=\tel\T{1}\z\T{3}\z\T{2}\z\rev{\T{5}}\z\rev{\T{4}}\z\T{3}\z\T{5}\z\T{4}\tel$ and~$B=\tel\T{1}\z\T{6}\z\T{2}\z\T{3}\z\T{1}\z\T{7}\z\T{3}\z\T{4}\z\T{1}\z\T{3}\tel$ can give rise to many distinct pairs of matched singular genomes.
  The relational diagrams of two of these pairs are represented here, in the top and center.
  In the bottom we show the multi-relational diagram $\MR(A,B)$.
  The decomposition that gives the diagram in the top is represented in red/orange.
  Similarly, the decomposition that gives the diagram in the center is represented in blue/cyan.
  Edges that are in both decompositions have two colors.
  }
 \end{figure}

\subsection{Multi-relational diagram}

While the original relational diagram clearly depends on the singularity of common markers, when they appear in multiple copies we can obtain a data structure that integrates the properties of all possible relational diagrams of matched genomes. 
The \emph{multi-relational diagram} $\MR(A,B)$ of two natural genomes $A$ and $B$
also has a set $V(A)$ with a vertex for each of the two extremities of each marker occurrence of genome $A$ and a set $V(B)$ with a vertex for each of the two extremities of each marker occurrence of genome $B$.

Again, sets $\adjedges{A}$ and $\adjedges{B}$ contain adjacency edges
connecting adjacent extremities of markers in $A$ and in $B$.
But here the set $E_{\ext}$ contains, for each marker $\mk\in\mg$, an
extremity edge connecting each vertex in $V(A)$ that represents an occurrence of
$\tl{\mk}$ to each vertex in $V(B)$ that represents an occurrence of $\tl{\mk}$, 
and an extremity edge connecting each vertex in $V(A)$ that represents an
occurrence of $\hd{\mk}$ to each vertex in $V(B)$ that represents an occurrence
of $\hd{\mk}$.
Furthermore, for each marker $\mk\in\mtu$ 
with $\occof{A}(\mk)>\occof{B}(\mk)$,
the set $\selfedges{A}$ contains one indel edge connecting the vertices
representing the two extremities of the same occurrence of~$\mk$ in $A$. 
Similarly, for each marker $\mk'\in\mtu$ with $\occof{B}(\mk')>\occof{A}(\mk')$,
the set $\selfedges{B}$ contains one indel edge connecting the vertices
representing the two extremities of the same occurrence of~$\mk'$ in $B$.
An example of a multi-relational diagram is given in Figure~\ref{fig:multi-adj-graph} (bottom).

\paragraph{Consistent decompositions.}
Note that if $A$ and $B$ are singular genomes, $\MR(A,B)$ reduces to the ordinary $R(A,B)$. On the other hand, in the presence of duplicate common markers, $\MR(A,B)$
may contain vertices of degree larger than two. A \emph{decomposition}
is a collection of vertex-disjoint \emph{components}, that can be cycles and/or paths, covering all vertices of $\MR(A,B)$. 
There can be multiple ways of selecting a decomposition,  
and we need to find one that 
allows to match occurrences of a marker in genome $A$ with 
occurrences of the same marker in genome $B$.

Let $\mk_{(A)}$ and $\mk_{(B)}$ be, respectively, occurrences of the same marker $\mk$ in genomes $A$ and $B$.
The extremity edge that connects $\hd{\mk_{(A)}}$ to $\hd{\mk_{(B)}}$ and the extremity edge that connects $\tl{\mk_{(A)}}$ to~$\tl{\mk_{(B)}}$ are called \emph{siblings}. 
A set $S \subseteq E_{\ext}$ is a \emph{sibling set} if it is exclusively composed of pairs of siblings and does not contain any pair of incident edges.
In other words, there is a clear bijection between matchings of occurrences and sibling sets of $\MR(A,B)$.
In particular, a \emph{maximal} sibling set $S$ corresponds to a maximal matching of occurrences of common markers in both genomes and we denote by $\AMi{S}$ and $\BMi{S}$ the matched singular genomes corresponding to the sibling set $S$.

The set of edges~$D[S]$ \emph{induced} by a maximal sibling set $S$ is said to be a \emph{consistent decomposition} of $\MR(A,B)$ and can be obtained by the following two steps: (i) in the beginning, $D[S]$ is the union of $S$ with the sets of adjacency edges $\adjedges{A}$ and $\adjedges{B}$; (ii) for each indel edge $e$, if its two endpoints have degree one or zero in $D[S]$, then~$e$ is added to $D[S]$.
Note that the consistent decomposition $D[S]$ covers all vertices of $\MR(A,B)$ and is composed of cycles and paths, allowing us to
compute the value
\[
\ddcjid(D) = n_* - c_D - \frac{i_D}{2} +\sum_{C \in D}\!\! \lambda(C) -\delta_D \,,
\] 
where $c_D$ and
$i_D$ are the numbers of $\AB$-cycles and $\AB$-paths in $D[S]$,
respectively, and $\delta_D$ is the optimal deduction of recombinations of paths from $D[S]$.
Since 
$n_*$ is constant for any consistent decomposition, we can separate the part of the formula that depends on $D[S]$, called \emph{weight} of $D[S]$:
\begin{equation*}\label{eq:decomposition-weight}
w(D[S]) = c_D + \frac{i_D}{2} -\sum_{C \in D[S]}\!\! \lambda(C) +\delta_D \,.
\end{equation*}

Given two natural genomes $A$ and $B$, the DCJ-indel distance of $A$ and $B$ can then be computed by the following equation:
\[ 
\ddcjid(A, B) = \min_{S \in \mathbb{S}_\textsc{max}}\{\ddcjid(D[S]) \} = n_* - \max_{S \in \mathbb{S}_\textsc{max}}\{w(D[S])\}\,,
\]  
where $\mathbb{S}_\textsc{max}$ is the set of all maximal sibling sets of $\MR(A,B)$.

A consistent decomposition $D[S]$ such that $\ddcjid(D[S])=\ddcjid(A, B)$ is
said to be \emph{optimal}.
Computing the DCJ-indel distance between two natural genomes $A$ and $B$, or, equivalently, finding an optimal consistent decomposition of $\MR(A,B)$ is an NP-hard problem. In Section~\ref{sec:ilp} we will describe an efficient ILP formulation to solve it. Before that, we need to introduce a transformation of~$\MR(A,B)$ that is necessary for our ILP.

\section{\label{app:capping}Capping}

The ends of linear chromosomes produce some difficulties for the decomposition.
Fortunately there is an elegant technique to overcome this problem, called \emph{capping}~\citep{HAN-PEV-1995}.
It consists of modifying the genomes by adding \emph{artificial} singular common markers, called \emph{caps}, that circularize all linear chromosomes, so that their relational diagram is composed of cycles only, but, if the capping is optimal, the genomic distance is preserved.

\subsection{Capping of canonical genomes}

When two singular genomes $A$ and $B$ have no exclusive markers, they are called \emph{canonical genomes}.
  
\smallskip

The diagram $R(A,B)$ of canonical genomes $A$ and $B$ has no indel edges and the indel-potential of any component $C$ is $\lambda(C)=0$. In this case, the upper bound given by Lemma~\ref{lem:distance-upper} is tight, and the distance formula can be simplified to $\ddcjid(A,B) = |\mg| - c -
\frac{i}{2}$, as it was already shown in~\citep{BER-MIX-STO-2006}.

\smallskip

Also, obtaining an optimal capping of canonical genomes is quite straightforward~\citep{HAN-PEV-1995,YAN-ATT-FRI-2005,BRA-STO-2010}, as shown in Table \ref{tab:capping-canonical}: the
caps should guarantee that each $\AB$-path is closed into a
separate $\AB$-cycle, and each pair composed of an $\AA$- and a $\BB$-path is closed
into an $\AB$-cycle by linking each extremity of the $\AA$-path to one of the two
extremities of the $\BB$-path (there are two possibilities of linking, and any
of the two is optimal). If the numbers of linear chromosomes in~$A$ and in $B$ are
different, there will be some $\AA$- or $\BB$-paths remaining. For each of these
an \emph{artificial} adjacency between caps is created in the
genome with less linear chromosomes, and each artificial adjacency closes
each remaining $\AA$- or $\BB$-path into a separate
$\AB$-cycle.

\begin{table}
\caption{\label{tab:capping-canonical}Linking paths from $R(A,B)$ of canonical genomes.
The symbol $\Gamma_A$ represents an artificial adjacency in $A$ and the symbol $\Gamma_B$ represents an artificial adjacency in $B$.
The value $\Delta d$ corresponds to $\Delta n - \Delta c - \Delta (2i)$. 
}
\begin{center}
\scriptsize
\begin{tabular}[t]{|c|c|r|r|r|r|}
\hline
{\bf ~~~~~~~~paths~~~~~~~~} & {\bf linking $\AB$-cycle} & \boldmath$\Delta n$ & \boldmath $\Delta c$ & \boldmath$\Delta(2i)$ & \boldmath$\Delta d$\\
\hline
\hline
$\AB$ & $(\AB)$ & $+0.5$ & $+1$ & $-0.5$ & $0$\\[1mm]
$\AA + \BB$ & $(\AA, \BB)$ & $+1$ & $+1$ & $0$ & $0$\\
\hline
\multicolumn{6}{c}{}\\
\hline
{\bf remaining paths} & {\bf linking $\AB$-cycle} & \boldmath$\Delta n$ & \boldmath $\Delta c$ & \boldmath$\Delta(2i)$ & \boldmath$\Delta d$\\
\hline
\hline
$\AA$ & $(\AA, \Gamma_B)$ & $+1$ & $+1$ & $0$ & $0$\\[1mm]
$\BB$ & $(\BB, \Gamma_A)$ & $+1$ & $+1$ & $0$ & $0$\\
\hline
\end{tabular}
\end{center}
\end{table}

Let $\kappa_{\!A}$ be the total number of linear chromosomes in $A$ and $\kappa_{\!B}$ be the total number of linear chromosomes in $B$. The difference between the number of $\AA$- or $\BB$-paths is equal to the difference between $\kappa_{\!A}$ and $\kappa_{\!B}$. In other words, if $R(A,B)$ has $a$ $\AA$-paths, $b$ $\BB$-paths and $i$ $\AB$-paths, 
the number of artificial adjacencies in such an optimal capping is exactly $a_*=|\kappa_{\!A}-\kappa_{\!B}|=|a-b|$. Moreover, 
the number of caps to be added is
\[
p_* = \max\{\kappa_{\!A},\kappa_{\!B}\} = \max\{a,b\} + \frac{i}{2}\,.
\]

We can show that the capping described above is optimal by verifying the corresponding DCJ-indel distances. Let the original genomes $A$ and $B$ have $n$ markers and $R(A,B)$ have $c$ $\AB$-cycles, besides the paths.
Then, after capping, the circular genomes $\AC$ and $\BC$ have $n' = n+p_*$ markers and
$R(\AC,\BC)$ has $c'=c+i+\max\{a,b\}$ $\AB$-cycles and no path, so that
\[
\ddcjid(\AC,\BC) = n'-c' = n+\max\{a,b\}+\frac{i}{2}-c-i-\max\{a,b\}=n-c-\frac{i}{2} = \ddcjid(A,B)\,.
\]

An example of an optimal capping of two canonical linear genomes is given in Figure~\ref{fig:capping-canonical}.

\begin{figure}
\begin{center}

\tiny
\setlength{\unitlength}{.8pt}
\begin{tabular}{c@{~~~~~~~~~~~~~~~~~~~~~~~~~~~~~~~~~}c}
\begin{picture}(210,80)

\put(0,0){\color{red}
\qbezier(10,70)(25,40)(40,10)
\qbezier(80,70)(65,40)(50,10)
\dottedline{1}(40,10)(50,10)
}
\put(0,0){\color{blue}
\qbezier(40,70)(25,40)(10,10)
\qbezier(50,70)(65,40)(80,10)
\dottedline{1}(40,70)(50,70)
}

\put(0,0){\color{orange}
\qbezier(140,70)(155,40)(170,10)
\qbezier(210,70)(195,40)(180,10)
\dottedline{1}(170,10)(180,10)
}
\put(0,0){\color{cyan}
\qbezier(170,70)(155,40)(140,10)
\qbezier(180,70)(195,40)(210,10)
\dottedline{1}(170,70)(180,70)
}

\put( 10,70){\circle*{4}}
\put( 40,70){\circle*{4}}
\put( 50,70){\circle*{4}}
\put( 80,70){\circle*{4}}
\put(140,70){\circle*{4}}
\put(170,70){\circle*{4}}
\put(180,70){\circle*{4}}
\put(210,70){\circle*{4}}

\put(-30,80){\makebox(0,0){$A$}}
\put( 10,80){\makebox(0,0){$\tel \tl{\T{2}}$}}
\put( 39,80){\makebox(0,0){$\hd{\T{2}}$}}
\put( 51,80){\makebox(0,0){$\tl{\T{1}}$}}
\put( 80,80){\makebox(0,0){$\hd{\T{1}}\tel$}}
\put(140,80){\makebox(0,0){$\tel\tl{\T{4}}$}}
\put(169,80){\makebox(0,0){$\hd{\T{4}}$}}
\put(181,80){\makebox(0,0){$\tl{\T{3}}$}}
\put(210,80){\makebox(0,0){$\hd{\T{3}}\tel$}}

\put( 10,10){\circle*{4}}
\put( 40,10){\circle*{4}}
\put( 50,10){\circle*{4}}
\put( 80,10){\circle*{4}}
\put(140,10){\circle*{4}}
\put(170,10){\circle*{4}}
\put(180,10){\circle*{4}}
\put(210,10){\circle*{4}}

\put(-30,-2){\makebox(0,0){$B$}}
\put( 10,-2){\makebox(0,0){$\tel \tl{\T{1}}$}}
\put( 39,-2){\makebox(0,0){$\hd{\T{1}}$}}
\put( 51,-2){\makebox(0,0){$\tl{\T{2}}$}}
\put( 80,-2){\makebox(0,0){$\hd{\T{2}}\tel$}}
\put(140,-2){\makebox(0,0){$\tel\tl{\T{3}}$}}
\put(169,-2){\makebox(0,0){$\hd{\T{3}}$}}
\put(181,-2){\makebox(0,0){$\tl{\T{4}}$}}
\put(210,-2){\makebox(0,0){$\hd{\T{4}}\tel$}}

\end{picture} &
\begin{picture}(210,80)
\put(0,0){\color{red}
\qbezier(10,70)(25,40)(40,10)
\qbezier(80,70)(65,40)(50,10)
\dottedline{1}(40,10)(50,10)
}
\put(0,0){\color{blue}
\qbezier(40,70)(25,40)(10,10)
\qbezier(50,70)(65,40)(80,10)
\dottedline{1}(40,70)(50,70)
}

\put(0,0){\color{orange}
\qbezier(130,70)(145,40)(160,10)
\qbezier(200,70)(185,40)(170,10)
\dottedline{1}(160,10)(170,10)
}
\put(0,0){\color{cyan}
\qbezier(160,70)(145,40)(130,10)
\qbezier(170,70)(185,40)(200,10)
\dottedline{1}(160,70)(170,70)
}

\put(0,0){\color{gray}
\drawline(0,70)(0,10)
\dottedline{1}( 0,70)(10,70)
\dottedline{1}( 0,10)(10,10)
\dottedline{1}(80,70)(90,70)
\dottedline{1}(80,10)(90,10)
\drawline(90,70)(90,10)

\drawline(120,70)(120,10)
\dottedline{1}(120,70)(130,70)
\dottedline{1}(120,10)(130,10)
\dottedline{1}(200,70)(210,70)
\dottedline{1}(200,10)(210,10)
\drawline(210,70)(210,10)
}

\put(  0,70){\circle*{4}}
\put( 10,70){\circle*{4}}
\put( 40,70){\circle*{4}}
\put( 50,70){\circle*{4}}
\put( 80,70){\circle*{4}}
\put( 90,70){\circle*{4}}
\put(120,70){\circle*{4}}
\put(130,70){\circle*{4}}
\put(160,70){\circle*{4}}
\put(170,70){\circle*{4}}
\put(200,70){\circle*{4}}
\put(210,70){\circle*{4}}

\put(-40,80){\makebox(0,0){$\AC$}}
\put( -1,80){\makebox(0,0){$\hd{\T{5}}$}}
\put( 11,80){\makebox(0,0){$\tl{\T{2}}$}}
\put( 39,80){\makebox(0,0){$\hd{\T{2}}$}}
\put( 51,80){\makebox(0,0){$\tl{\T{1}}$}}
\put( 79,80){\makebox(0,0){$\hd{\T{1}}$}}
\put( 91,80){\makebox(0,0){$\tl{\T{5}}$}}
\put(119,80){\makebox(0,0){$\hd{\T{6}}$}}
\put(131,80){\makebox(0,0){$\tl{\T{4}}$}}
\put(159,80){\makebox(0,0){$\hd{\T{4}}$}}
\put(171,80){\makebox(0,0){$\tl{\T{3}}$}}
\put(199,80){\makebox(0,0){$\hd{\T{3}}$}}
\put(211,80){\makebox(0,0){$\tl{\T{6}}$}}

\put(  0,10){\circle*{4}}
\put( 10,10){\circle*{4}}
\put( 40,10){\circle*{4}}
\put( 50,10){\circle*{4}}
\put( 80,10){\circle*{4}}
\put( 90,10){\circle*{4}}
\put(120,10){\circle*{4}}
\put(130,10){\circle*{4}}
\put(160,10){\circle*{4}}
\put(170,10){\circle*{4}}
\put(200,10){\circle*{4}}
\put(210,10){\circle*{4}}

\put(-40,-2){\makebox(0,0){$\BC$}}
\put( -1,-2){\makebox(0,0){$\hd{\T{5}}$}}
\put( 11,-2){\makebox(0,0){$\tl{\T{1}}$}}
\put( 39,-2){\makebox(0,0){$\hd{\T{1}}$}}
\put( 51,-2){\makebox(0,0){$\tl{\T{2}}$}}
\put( 79,-2){\makebox(0,0){$\hd{\T{2}}$}}
\put( 91,-2){\makebox(0,0){$\tl{\T{5}}$}}
\put(119,-2){\makebox(0,0){$\hd{\T{6}}$}}
\put(131,-2){\makebox(0,0){$\tl{\T{3}}$}}
\put(159,-2){\makebox(0,0){$\hd{\T{3}}$}}
\put(171,-2){\makebox(0,0){$\tl{\T{4}}$}}
\put(199,-2){\makebox(0,0){$\hd{\T{4}}$}}
\put(211,-2){\makebox(0,0){$\tl{\T{6}}$}}

\end{picture}
\end{tabular}
\end{center}
\caption{\label{fig:capping-canonical}Optimal capping of canonical genomes $A= \{\tel \T{2}\z\T{1}\tel, \tel \T{4}\z\T{3}\tel\}$ and $B=\{\tel \T{1}\z\T{2}\z\tel, \tel\T{3}\z\T{4}\tel \}$ into $\AC=\{(\T{2}\z\T{1}\z\T{5}),(\T{4}\z\T{3}\z\T{6})\}$ and $\BC=\{(\T{1}\z\T{2}\z\T{5}), (\T{3}\z\T{4}\z\T{6})\}$.
Each pair of $\AA$- + $\BB$-path is linked into a separate $\AB$-cycle.}
\end{figure}

\subsection{Singular genomes: correspondence between recombinations and capping}

When exclusive markers occur, we can obtain an optimal capping by simply finding caps that properly link the sources of each recombination group (listed in Table~\ref{tab:recomb-group})
into a single $\AB$-cycle. Indeed, in Table~\ref{tab:cap-group} we give a linking that achieves the optimal $\Delta d$ for each recombination group, followed by the optimal linking of remaining paths.
The remaining paths are treated exactly as the linking of paths in canonical genomes. 
By greedily linking the paths following a top-down order of the referred Table~\ref{tab:cap-group} we clearly obtain an optimal capping that transforms $A$ and $B$ into circular genomes $\AC$ and $\BC$ 
with~$\ddcjid(\AC,\BC) = \ddcjid(A,B)$.
See an example in Figure~\ref{fig:capping-singular}. 

\begin{figure}
\begin{center}
\scriptsize
\setlength{\unitlength}{.85pt}
\begin{picture}(350,80)

      \put(0,0){\color{red}
      \drawline(30,70)(60,70)
      \drawline(90,10)(120,10)
      \qbezier(70,70)(100,40)(130,10)
      \qbezier(140,70)(110,40)(80,10)
      \dottedline{1}(60,70)(70,70)
      \dottedline{1}(80,10)(90,10)
      \dottedline{1}(120,10)(130,10)

      }
      \put(0,0){\color{blue}
      \drawline(10,10)(40,10)
      \qbezier(100,70)(130,40)(160,10)
      \qbezier(110,70)(80,40)(50,10)
      \dottedline{1}(40,10)(50,10)
      \dottedline{1}(100,70)(110,70)
      }
      
      \put(0,0){\color{orange}
      \drawline(190,70)(220,70)
      \drawline(250,10)(280,10)
      \qbezier(230,70)(260,40)(290,10)
      \qbezier(340,70)(290,40)(240,10)
      \dottedline{1}(220,70)(230,70)
      \dottedline{1}(240,10)(250,10)
      \dottedline{1}(280,10)(290,10)
      }
      \put(0,0){\color{cyan}
      \drawline(270,70)(300,70)
      \qbezier(260,70)(290,40)(320,10)
      \qbezier(310,70)(260,40)(210,10)
      \dottedline{1}(260,70)(270,70)
      \dottedline{1}(300,70)(310,70)
      }

      \put( 30,70){\circle*{4}}
      \put( 60,70){\circle*{4}}
      \put( 70,70){\circle*{4}}
      \put(100,70){\circle*{4}}
      \put(110,70){\circle*{4}}
      \put(140,70){\circle*{4}}

      \put(190,70){\circle*{4}}
      \put(220,70){\circle*{4}}
      \put(230,70){\circle*{4}}
      \put(260,70){\circle*{4}}
      \put(270,70){\circle*{4}}
      \put(300,70){\circle*{4}}
      \put(310,70){\circle*{4}}
      \put(340,70){\circle*{4}}

\put(-10,80){\makebox(0,0){$A$}}
\put( 30,80){\makebox(0,0){$\tel\tl{\T{5}}$}}
\put( 59,80){\makebox(0,0){$\hd{\T{5}}$}}
\put( 71,80){\makebox(0,0){$\tl{\T{2}}$}}
\put( 99,80){\makebox(0,0){$\hd{\T{2}}$}}
\put(111,80){\makebox(0,0){$\tl{\T{1}}$}}
\put(140,80){\makebox(0,0){$\hd{\T{1}}\tel$}}
\put(190,80){\makebox(0,0){$\tel\tl{\T{5}}$}}
\put(219,80){\makebox(0,0){$\hd{\T{5}}$}}
\put(231,80){\makebox(0,0){$\tl{\T{4}}$}}
\put(259,80){\makebox(0,0){$\hd{\T{4}}$}}
\put(271,80){\makebox(0,0){$\tl{\T{5}}$}}
\put(299,80){\makebox(0,0){$\hd{\T{5}}$}}
\put(311,80){\makebox(0,0){$\tl{\T{3}}$}}
\put(340,80){\makebox(0,0){$\hd{\T{3}}\tel$}}

      \put( 10,10){\circle*{4}}
      \put( 40,10){\circle*{4}}
      \put( 50,10){\circle*{4}}
      \put( 80,10){\circle*{4}}
      \put( 90,10){\circle*{4}}
      \put(120,10){\circle*{4}}
      \put(130,10){\circle*{4}}
      \put(160,10){\circle*{4}}
      
      \put(210,10){\circle*{4}}
      \put(240,10){\circle*{4}}
      \put(250,10){\circle*{4}}
      \put(280,10){\circle*{4}}
      \put(290,10){\circle*{4}}
      \put(320,10){\circle*{4}}
      
      \put(-10,-2){\makebox(0,0){$B$}}
      \put( 10,-2){\makebox(0,0){$\tel\tl{\T{6}}$}}
      \put( 39,-2){\makebox(0,0){$\hd{\T{6}}$}}
      \put( 51,-2){\makebox(0,0){$\tl{\T{1}}$}}
      \put( 79,-2){\makebox(0,0){$\hd{\T{1}}$}}
      \put( 91,-2){\makebox(0,0){$\tl{\T{6}}$}}
      \put(119,-2){\makebox(0,0){$\hd{\T{6}}$}}
      \put(131,-2){\makebox(0,0){$\tl{\T{2}}$}}
      \put(160,-2){\makebox(0,0){$\hd{\T{2}}\tel$}}
      \put(210,-2){\makebox(0,0){$\tel\tl{\T{3}}$}}
      \put(239,-2){\makebox(0,0){$\hd{\T{3}}$}}
      \put(251,-2){\makebox(0,0){$\tl{\T{6}}$}}
      \put(279,-2){\makebox(0,0){$\hd{\T{6}}$}}
      \put(291,-2){\makebox(0,0){$\tl{\T{4}}$}}
      \put(320,-2){\makebox(0,0){$\hd{\T{4}}\tel$}}

    \end{picture} 
    
    \vspace{10mm}
    
    \begin{picture}(330,110)
       \put(0,0){\color{red}
      \drawline(10,100)(40,100)
      \drawline(90,10)(120,10)
      \qbezier(50,100)(90,55)(130,10)
      \qbezier(120,100)(100,55)(80,10)
      \dottedline{1}(40,100)(50,100)
      \dottedline{1}(80,10)(90,10)
      \dottedline{1}(120,10)(130,10)

      }
      \put(0,0){\color{blue}
      \drawline(10,10)(40,10)
      \qbezier(80,100)(120,55)(160,10)
      \qbezier(90,100)(70,55)(50,10)
      \dottedline{1}(40,10)(50,10)
      \dottedline{1}(80,100)(90,100)
      }
      
      \put(0,0){\color{orange}
      \drawline(170,100)(200,100)
      \drawline(250,10)(280,10)
      \qbezier(210,100)(250,55)(290,10)
      \qbezier(320,100)(280,55)(240,10)
      \dottedline{1}(200,100)(210,100)
      \dottedline{1}(240,10)(250,10)
      \dottedline{1}(280,10)(290,10)
      }
      \put(0,0){\color{cyan}
      \drawline(250,100)(280,100)
      \qbezier(240,100)(280,55)(320,10)
      \qbezier(290,100)(250,55)(210,10)
      \dottedline{1}(240,100)(250,100)
      \dottedline{1}(280,100)(290,100)
      }
      
      \put(0,0){\color{gray}
      \dottedline{1}(0,100)(10,100)
      \dottedline{1}(120,100)(130,100)
      \dottedline{1}(160,100)(170,100)
      \dottedline{1}(320,100)(330,100)
      \dottedline{1}(0,10)(10,10)
      \dottedline{1}(160,10)(170,10)
      \dottedline{1}(200,10)(210,10)
      \dottedline{1}(320,10)(330,10)

      \qbezier[150](0,100)(72,67)(165,55)
      \qbezier[150](165,55)(257,43)(330,10)
      \qbezier(130,100)(150,55)(170,10)
      \qbezier(160,100)(180,55)(200,10)
      \qbezier[100](330,100)(257,67)(165,55)
      \qbezier[100](165,55)(72,43)(0,10)
      }
      
      \put(  0,100){\circle*{4}}
      \put( 10,100){\circle*{4}}
      \put( 40,100){\circle*{4}}
      \put( 50,100){\circle*{4}}
      \put( 80,100){\circle*{4}}
      \put( 90,100){\circle*{4}}
      \put(120,100){\circle*{4}}
      \put(130,100){\circle*{4}}

      \put(160,100){\circle*{4}}
      \put(170,100){\circle*{4}}
      \put(200,100){\circle*{4}}
      \put(210,100){\circle*{4}}
      \put(240,100){\circle*{4}}
      \put(250,100){\circle*{4}}
      \put(280,100){\circle*{4}}
      \put(290,100){\circle*{4}}
      \put(320,100){\circle*{4}}
      \put(330,100){\circle*{4}}

\put(-20,110){\makebox(0,0){$\AC$}}
\put( -1,110){\makebox(0,0){$\tl{\T{7}}$}}
\put( 11,110){\makebox(0,0){$\tl{\T{5}}$}}
\put( 39,110){\makebox(0,0){$\hd{\T{5}}$}}
\put( 51,110){\makebox(0,0){$\tl{\T{2}}$}}
\put( 79,110){\makebox(0,0){$\hd{\T{2}}$}}
\put( 91,110){\makebox(0,0){$\tl{\T{1}}$}}
\put(119,110){\makebox(0,0){$\hd{\T{1}}$}}
\put(131,110){\makebox(0,0){$\tl{\T{8}}$}}
\put(159,110){\makebox(0,0){$\hd{\T{8}}$}}
\put(171,110){\makebox(0,0){$\tl{\T{5}}$}}
\put(199,110){\makebox(0,0){$\hd{\T{5}}$}}
\put(211,110){\makebox(0,0){$\tl{\T{4}}$}}
\put(239,110){\makebox(0,0){$\hd{\T{4}}$}}
\put(251,110){\makebox(0,0){$\tl{\T{5}}$}}
\put(279,110){\makebox(0,0){$\hd{\T{5}}$}}
\put(291,110){\makebox(0,0){$\tl{\T{3}}$}}
\put(319,110){\makebox(0,0){$\hd{\T{3}}$}}
\put(331,110){\makebox(0,0){$\hd{\T{7}}$}}

      \put(  0,10){\circle*{4}}
      \put( 10,10){\circle*{4}}
      \put( 40,10){\circle*{4}}
      \put( 50,10){\circle*{4}}
      \put( 80,10){\circle*{4}}
      \put( 90,10){\circle*{4}}
      \put(120,10){\circle*{4}}
      \put(130,10){\circle*{4}}
      \put(160,10){\circle*{4}}
      \put(170,10){\circle*{4}}
      \put(200,10){\circle*{4}}
      \put(210,10){\circle*{4}}
      \put(240,10){\circle*{4}}
      \put(250,10){\circle*{4}}
      \put(280,10){\circle*{4}}
      \put(290,10){\circle*{4}}
      \put(320,10){\circle*{4}}
      \put(330,10){\circle*{4}}
      
      \put(-20,-2){\makebox(0,0){$\BC$}}
      \put( -1,-2){\makebox(0,0){$\hd{\T{7}}$}}
      \put( 11,-2){\makebox(0,0){$\tl{\T{6}}$}}
      \put( 39,-2){\makebox(0,0){$\hd{\T{6}}$}}
      \put( 51,-2){\makebox(0,0){$\tl{\T{1}}$}}
      \put( 79,-2){\makebox(0,0){$\hd{\T{1}}$}}
      \put( 91,-2){\makebox(0,0){$\tl{\T{6}}$}}
      \put(119,-2){\makebox(0,0){$\hd{\T{6}}$}}
      \put(131,-2){\makebox(0,0){$\tl{\T{2}}$}}
      \put(159,-2){\makebox(0,0){$\hd{\T{2}}$}}
      \put(171,-2){\makebox(0,0){$\tl{\T{8}}$}}
      \put(199,-2){\makebox(0,0){$\hd{\T{8}}$}}
      \put(211,-2){\makebox(0,0){$\tl{\T{3}}$}}
      \put(239,-2){\makebox(0,0){$\hd{\T{3}}$}}
      \put(251,-2){\makebox(0,0){$\tl{\T{6}}$}}
      \put(279,-2){\makebox(0,0){$\hd{\T{6}}$}}
      \put(291,-2){\makebox(0,0){$\tl{\T{4}}$}}
      \put(319,-2){\makebox(0,0){$\hd{\T{4}}$}}
      \put(331,-2){\makebox(0,0){$\tl{\T{7}}$}}
      
    \end{picture}
 
  \end{center}
  \caption{\label{fig:capping-singular}Optimal capping of singular genomes $A= \{\tel \T{5}\z\T{2}\z\T{1}\tel, \tel \T{5}\z\T{4}\z\T{5}\z\T{3}\tel\}$ and $B=\{\tel \T{6}\z\T{1}\z\T{6}\z\T{2}\z\tel, \tel\T{3}\z\T{6}\z\T{4}\tel \}$ into $\AC=\{(\rev{\T{7}}\z\T{5}\z\T{2}\z\T{1\z}\T{8}\z\T{5}\z\T{4}\z\T{5}\z\T{3})\}$ and $\BC=\{(\T{7}\z\T{6}\z\T{1}\z\T{6}\z\T{2}\z\T{8}\z\T{3}\z\T{6}\z\T{4})\}$.
  This capping shows how to optimally link the four sources of the chained recombinations of Figure~\ref{fig:chained-recomb} into a single $\AB$-cycle.}
\end{figure}

Furthermore, similarly to the case of canonical genomes, the numbers of artificial adjacencies and caps in such a capping are respectively $a_*=|\kappa_{\!A}-\kappa_{\!B}|$ and~$p_* = \max\{\kappa_{\!A},\kappa_{\!B}\}$ as we will show in the following.

In Table~\ref{tab:cap-group} we can observe that there are two types of groups: (i) balanced, that contain the same number of $\AA$- and $\BB$-paths, and (ii) unbalanced, in which the numbers of $\AA$- and $\BB$-paths are distinct. Unbalanced groups require some extra elements to link the cycle. These elements can be indel-free $\AA$- or $\BB$-paths (of the type that is under-represented in the group) or, if these paths do not exist, artificial adjacencies either in genome $A$ or in genome $B$ (again, of the genome that is under-represented in the group). We then need to examine these unbalanced groups to determine the number of caps and of artificial adjacencies that are required for an optimal capping.

\begin{table}
\caption{Linking sources of chained recombination groups from Table~\ref{tab:recomb-group}.
The symbol $\Gamma_A$ represents an artificial adjacency in $A$ and the symbol $\Gamma_B$ represents an artificial adjacency in $B$. The notation $\AA_\no\!\prec\!\Gamma_A$ means that an $\AA_\no$-path is preferred to close the cycle, but if it does not exist, we take an artificial adjacency in $A$.
In order to give the correct order of linking, we sometimes need to represent a path $\ABab$ by $\BAba$ and a path $\ABba$ by $\BAab$. 
The value $\Delta d$ corresponds to $\Delta n - \Delta c - \Delta (2i) + \Delta \lambda$. Unbalanced groups over-represented in genome $A$ are marked with a ``$\cup$'', while unbalanced groups over-represented in genome $B$ are marked with a ``$\cap$'' .
}
\label{tab:cap-group}
\begin{center}
\scriptsize
\begin{tabular}[t]{|rc|l|l|c|r|r|r|r|r|}
\hline
& {\bf id} &  {\bf sources} & {\bf linking \boldmath$\AB$-cycle} & & \boldmath$\Delta n$ & \boldmath $\Delta c$ & \boldmath$\Delta(2i)$ & \boldmath$\Delta \lambda$ & \boldmath$\Delta d$\\
\hline
\hline
$\mathcal P$ & $\W\M$ & $\AAab + \BBab$ & $(\AAab, \BBba)$ & & $+1$ & $+1$ & $0$ & $-2$ & $-2$\\
\hline
\hline
& & & & && & &&\\[-3mm]
$\mathcal Q$ 
& $\W\W\MA\MB$ & $2\times\AAab + \BBaa + \BBbb$ & $(\AAab, \BBbb, \AAba, \BBaa)$ & & $+2$ & $+1$ & $0$ & $-4$ & $-3$\\[1mm]
& $\M\M\WA\WB$ & $2\times\BBab + \AAaa + \AAbb$ & $(\BBab, \AAbb, \BBba, \AAaa)$ & & $+2$ & $+1$ & $0$ & $-4$ & $-3$\\
\hline
\hline
& & & & && & &&\\[-3mm]
$\mathcal T$ 
& $\W\Z\MA$ & $\AAab + \BBaa + \ABab$ & $(\ABab, \AAba, \BBaa)$ & & $+1.5$ & $+1$ & $-0.5$ & $-3$ & $-2$\\
& $\W\W\MA$ & $2\times\AAab + \BBaa$ & $(\AAba, \BBaa, \AAab, \BB_\no\!\prec\!\Gamma_B)$ & $\cup$& $+2$ & $+1$ & $0$ & $-3$ & $-2$\\[1mm]
& $\W\N\MB$ & $\AAab + \BBbb + \ABba$ & $(\ABba, \AAab, \BBbb)$ & & $+1.5$ & $+1$ & $-0.5$ & $-3$ & $-2$\\
& $\W\W\MB$ & $2\times\AAab + \BBbb$ & $(\AAab, \BBaa, \AAab, \BB_\no\!\prec\!\Gamma_B)$ & $\cup$& $+2$ & $+1$ & $0$ & $-3$ & $-2$\\[1mm]
& $\M\N\WA$ & $\BBab + \AAaa + \ABba$ & $(\ABba, \AAaa, \BBab)$ & & $+1.5$ & $+1$ & $-0.5$ & $-3$ & $-2$\\
& $\M\M\WA$ & $2\times\BBab + \AAaa$ & $(\BBba, \AAaa, \BBab, \AA_\no\!\prec\!\Gamma_A)$ &$\cap$& $+2$ & $+1$ & $0$ & $-3$ & $-2$\\[1mm]
& $\M\Z\WB$ & $\BBab + \AAbb + \ABab$& $(\ABab, \AAbb, \BBba)$ & & $+1.5$ & $+1$ & $-0.5$ & $-3$ & $-2$\\
& $\M\M\WB$ & $2\times\BBab + \AAbb$& $(\BBab, \AAbb, \BBba, \AA_\no\!\prec\!\Gamma_A)$ &$\cap$& $+2$ & $+1$ & $0$ & $-3$ & $-2$\\[1mm]
\hline
\hline
$\mathcal S$ 
&$\Z\N$ & $\ABab + \ABba$  & $(\ABab, \ABba)$ && $+1$ & $+1$ & $-1$ & $-2$ & $-1$\\[1mm]
& $\WA\MA$ & $\AAaa + \BBaa$  & $(\AAaa, \BBaa)$ && $+1$ & $+1$ & $0$ & $-1$ & $-1$\\
& $\WB\MB$ & $\AAbb + \BBbb$  & $(\AAbb, \BBbb)$ && $+1$ & $+1$ & $0$ & $-1$ & $-1$\\[1mm]
& $\W\MA$ & $\AAab + \BBaa$  & $(\AAba, \BBaa)$ && $+1$ & $+1$ & $0$ & $-1$ & $-1$\\[1mm]
& $\W\MB$ & $\AAab + \BBbb$  & $(\AAab, \BBbb)$ && $+1$ & $+1$ & $0$ & $-1$ & $-1$\\[1mm]
& $\W\Z$ & $\AAab + \ABab$ & $(\AAba,\BB_\no\!\prec\!\Gamma_B, \ABab)$ &$\cup$&  $+1.5$ & $+1$ & $-0.5$ & $-2$ & $-1$\\[1mm]
& $\W\N$  & $\AAab + \ABba$ & $(\AAab,\BB_\no\!\prec\!\Gamma_B, \ABba)$ &$\cup$& $+1.5$ & $+1$ & $-0.5$ & $-2$ & $-1$\\[1mm]
& $\W\W$& $\AAab + \AAab$  & $(\AAab, \BB_\no\!\prec\!\Gamma_B, \AAba, \BB_\no\!\prec\!\Gamma_B)$ &$\cup$& $+2$ & $+1$ & $0$ & $-2$ & $-1$\\[1mm]
& $\M\WA$ & $\BBab + \AAaa$  & $(\AAaa, \BBab)$ && $+1$ & $+1$ & $0$ & $-1$ & $-1$\\[1mm]
& $\M\WB$ & $\BBab + \AAbb$  & $(\AAbb, \BBba)$ && $+1$ & $+1$ & $0$ & $-1$ & $-1$\\[1mm]
&$\M\Z$& $\BBab + \ABab$ & $(\BBba,\ABab,\AA_\no\!\prec\!\Gamma_A)$ &$\cap$& $+1.5$ & $+1$ & $-0.5$ & $-2$ & $-1$\\[1mm]
&$\M\N$& $\BBab + \ABba$ & $(\BBab,\ABba,\AA_\no\!\prec\!\Gamma_A)$ &$\cap$& $+1.5$ & $+1$ & $-0.5$ & $-2$ & $-1$\\[1mm]
&$\M\M$&  $\BBab + \BBab$ & $(\BBab, \AA_\no\!\prec\!\Gamma_A, \BBba, \AA_\no\!\prec\!\Gamma_A)$ &$\cap$& $+2$ & $+1$ & $0$ & $-2$ & $-1$\\
\hline
\hline
& & & & && & &&\\[-3mm]
$\mathcal M$ 
&$\Z\Z\WB\MA$& $2\times\ABab + \AAbb + \BBaa$  & $(\ABab, \AAbb, \BAba, \BBaa)$ && $+2$ & $+1$ & $-1$ & $-4$ & $-2$\\[1mm]
&$\N\N\WA\MB$& $2\times\ABba + \AAaa + \BBbb$ & $(\ABba, \AAaa, \BAab, \BBbb)$ && $+2$ & $+1$ & $-1$ & $-4$ & $-2$\\
\hline
\hline
& & & & && & &&\\[-3mm]
$\mathcal N$ 
&$\Z\WB\MA$& $\ABab + \AAbb + \BBaa$ & $(\ABab, \AAbb, \BBaa)$ && $+1.5$ & $+1$ & $-0.5$ & $-2$ & $-1$\\[1mm]
&$\Z\Z\WB$& $2\times\ABab + \AAbb$ & $(\ABab, \AAbb, \BAba, \BB_\no\!\prec\!\Gamma_B)$ &$\cup$& $+2$ & $+1$ & $-1$ & $-3$ & $-1$\\[1mm]
&$\Z\Z\MA$& $2\times\ABab + \BBaa$ & $(\BAba, \BBaa, \ABab, \AA_\no\!\prec\!\Gamma_A)$ &$\cap$& $+2$ & $+1$ & $-1$ & $-3$ & $-1$\\[1mm]
&$\N\WA\MB$& $\ABba + \AAaa + \BBbb$ & $(\ABba, \AAaa, \BBbb)$ && $+1.5$ & $+1$ & $-0.5$ & $-2$ & $-1$\\[1mm]
&$\N\N\WA$& $2\times\ABba + \AAaa$ & $(\ABba, \AAaa, \BAab, \BB_\no\!\prec\!\Gamma_B)$ &$\cup$& $+2$ & $+1$ & $-1$ & $-3$ & $-1$\\[1mm]
&$\N\N\MB$& $2\times\ABba + \BBbb$ & $(\BAab, \BBbb, \ABba, \AA_\no\!\prec\!\Gamma_A)$ &$\cap$& $+2$ & $+1$ & $-1$ & $-3$ & $-1$\\
\hline
\hline
& & {\bf remaining paths} & {\bf linking $\AB$-cycle} & & \boldmath$\Delta n$ & \boldmath $\Delta c$ & \boldmath$\Delta(2i)$ & \boldmath$\Delta \lambda$ & \boldmath$\Delta d$\\
\hline
\hline
& & $\AB_*$ & $(\AB_*)$ && $+0.5$ & $+1$ & $-0.5$ & $0$ & $0$\\
& & $\AA_* + \BB_*$ & $(\AA_*, \BB_*)$ && $+1$ & $+1$ & $0$ & $0$ & $0$\\
\hline
\hline
& & $\AA_*$ & $(\AA_*, \Gamma_B)$ &$\cup$& $+1$ & $+1$ & $0$ & $0$ & $0$\\[1mm]
& & $\BB_*$ & $(\BB_*, \Gamma_A)$ &$\cap$& $+1$ & $+1$ & $0$ & $0$ & $0$\\
\hline

\end{tabular}
\end{center}
\end{table}

\begin{proposition}\label{prop:one-type-unbalanced}
After identifying the recombination groups, either we have only unbalanced groups that are over-represented in genome $A$ or we have only unbalanced groups that are over-represented in genome $B$.
\end{proposition}
\begin{proof}
It is clear that, after $\mathcal P$ and until $\mathcal N$, we have either only groups $\W*$ (over-represented in $A$), or only groups $\M*$ (over-represented in $B$). The question is whether groups in $\mathcal N$ that are over-represented in $B$ are compatible with previous groups of type $\W*$ and, symmetrically, whether groups in $\mathcal N$ that are over-represented in $A$ are compatible with previous groups of type $\M*$. 
 
Let us examine the case of group $\Z\Z\WB$. (i) At a first glance one could think that this group is compatible with $\M\M\WB$. However, if all components of these two unbalanced groups would be in the diagram, we would instead have two times the group $\M\Z\WB$, that is balanced and located before the two other groups in the table (observe that $2\times\M\Z\WB$ has a smaller $\Delta d$ than $\Z\Z\WB+\M\M\WB$). (ii) When we test the compatibility of $\Z\Z\WB$ with $\M\M\WA$, we see that with the same components we would get $\M\M\WA\WB$, that is balanced and located before the two other groups in the table (observe that $\M\M\WA\WB$ has the same $\Delta d$ as $\Z\Z\WB+\M\M\WA$).
 
With a similar analysis we can show that for all cases either we have only unbalanced groups that are over-represented in genome $A$ or we have only unbalanced groups that are over-represented in genome $B$.
\qed
\end{proof}

\begin{proposition}\label{prop:close-unbalanced}
When an unbalanced group is being linked, either there is a remaining $\AA$- or $\BB$-path (of the genome that is under-represented), that is then used to link the group, or there is no remaining $\AA$- or $\BB$-path (of the genome that is under-represented) and an artificial adjacency links the group.
\end{proposition}
\begin{proof}
First we observe that, after distributing all paths of the relational diagram among the recombination groups, following the top-down greedy approach, there could be $\AA$- and/or $\BB$-paths remaining, that were not assigned to any group, and they might be useful to link unbalanced groups.
We will now examine the procedure of linking the unbalanced groups either with those remaining paths or with artificial adjacencies.

A particular case are the unbalanced groups from $\mathcal T$.
Since all unbalanced groups in $\mathcal T$ have analogous compositions, without loss of generality, suppose a group over-represented in genome $A$ of type $\W\W\MA$ is being linked.
If, at this point, there is a remaining indel-enclosing $\BB$-path, it cannot be $\BBab$ or $\BBbb$, otherwise with the components of the group being linked and the existing remaining path we could form a balanced group that appears in a higher position of the table, with at least the same $\Delta d$, which is a contradiction.
We could however have an extra $\BBaa$-path.
In this case we would take the alternative solution of linking each pair $\AAab$+$\BBaa$ into a separate cycle, that is twice group $\W\MA$ of $\mathcal S$, achieving the same $\Delta d$.
If no $\BBaa$-path remains, we would have the standard linking of the three paths into a single cycle including either an indel-free $\BB$-path or an artificial adjacency in $B$.

The unbalanced groups from $\mathcal S$ or $\mathcal N$ are easier to analyze:
if one of these groups,
over-represented in genome $A$ (respectively in genome $B$), is being linked, there cannot be any remaining indel-enclosing $\BB$-path (respectively $\AA$-path).
We can verify this by supposing, without loss of generality, that an unbalanced group over-represented in genome $A$ is being linked. If, at this point, there is a remaining indel-enclosing $\BB$-path, then with the components of the group being linked and the existing remaining path we could form a balanced group that appears in a higher position of the table, with at least the same $\Delta d$, which is a contradiction. 
\qed
\end{proof}

Propositions~\ref{prop:one-type-unbalanced} 
and~\ref{prop:close-unbalanced} prove the following result.

\begin{theorem}\label{thm:capping-singular}
Let $\kappa_{\!A}$ and $\kappa_{\!B}$ be, respectively, the total numbers of linear chromosomes in singular genomes $A$ and $B$. We can obtain an optimal capping of $A$ and $B$ with exactly
\[
p_* = \max\{\kappa_{\!A},\kappa_{\!B}\}
\]
caps and $a_*=|\kappa_{\!A}-\kappa_{\!B}|$ artificial adjacencies between caps.
\end{theorem}

\subsection{Capped multi-relational diagram}

We can transform $\MR(A,B)$ into the \emph{capped multi-relational diagram} $\MRC(A,B)$ as follows.
First we need to create $4p_*$ new vertices, named $\cext{A}{1},\cext{A}{2},\ldots,\cext{A}{2p_*}$ and $\cext{B}{1}, \cext{B}{2},\ldots,\cext{B}{2p_*}$, each one representing a \emph{cap extremity}. 
Each of the $2\kappa_{\!A}$ telomeres of $A$ is connected by an adjacency edge to a distinct cap extremity among $\cext{A}{1},\cext{A}{2},\ldots,\cext{A}{2\kappa_{\!A}}$.
Similarly, each of the $2\kappa_{\!B}$ telomeres of $B$ is connected by an adjacency edge to a distinct cap extremity among $\cext{B}{1},\cext{B}{2},\ldots,\cext{B}{2\kappa_{\!B}}$.
Moreover, if $\kappa_{\!A} < \kappa_{\!B}$, for $i = 2\kappa_{\!A}+1,2\kappa_{\!A}+3,\ldots,2\kappa_{\!B}-1$, connect $\cext{A}{i}$ to $\cext{A}{i+1}$ by an \emph{artificial adjacency edge}. Otherwise, if $\kappa_{\!B} < \kappa_{\!A}$, for $j = 2\kappa_{\!B}+1,2\kappa_{\!B}+3,\ldots,2\kappa_{\!A}-1$, connect $\cext{B}{j}$ to $\cext{B}{j+1}$ by an artificial adjacency edge. 
All these new adjacency edges and artificial adjacency edges are added to $\adjedges{A}$ and $\adjedges{B}$, respectively.
We also connect each $\cext{A}{i}$, $1\leq i\leq 2p_*$, by a \emph{cap extremity edge} to each $\cext{B}{j}$, $1\leq j\leq 2p_*$, and denote by $E_{\cp}$ the set of cap extremity edges. An example of a capped multi-relational diagram is given in Figure~\ref{fig:capped-multi-adj-graph}.

\begin{figure}
\begin{center}
\tiny
\setlength{\unitlength}{.7pt}
    \begin{picture}(570,160)
    
      \put(0,0){\color{cyan}
      \drawline(240,149)(280,149)
      \drawline(300,150)(340,150)
      }
     \put(0,0){\color{blue}
      \drawline(420,149)(460,149)
      }
      
      \put(0,0){\color{purple}
      \drawline(240,151)(280,151)
      \drawline(420,151)(460,151)
      \drawline(480,150)(520,150)
      }
      
      \put(0,0){\color{orange}
      \drawline(0,10)(40,10)
      \drawline(60,11)(100,11)
      }
      \put(0,0){\color{purple}
      \drawline(300,11)(340,11)
      \drawline(480,11)(520,11)
      \drawline(540,11)(580,11)
      }
      
      \put(0,0){\color{blue}
      \drawline(60,9)(100,9)
      \drawline(240,10)(280,10)
      \drawline(300,9)(340,9)
      }
      \put(0,0){\color{lightgray}
      \drawline(180,10)(220,10)
      \drawline(360,10)(400,10)
      }
      \put(0,0){\color{cyan}
      \drawline(480,9)(520,9)
      \drawline(540,9)(580,9)
      }
      
      \put(0,0){\color{gray}
      \qbezier[150](40,150)(140,95)(320,80)
      \qbezier[150](320,80)(500,65)(600,10)
      
      \qbezier(60,150)(270,80)(480,10)
      \qbezier(100,150)(300,77)(520,10)
      \qbezier(120,150)(340,83)(540,10)
      \qbezier(160,150)(370,80)(580,10)
      \qbezier(360,150)(450,80)(540,10)
      \qbezier(400,150)(490,80)(580,10)
      
      \qbezier[200](540,150)(440,100)(260,80)
      \qbezier[200](260,80)(100,60)(-20,10)
      }
      \put(0,0){\color{cyan}
      \qbezier(41,150)(11,80)(-19,10)
      \qbezier(60,150)(30,80)(0,10)
      }
      \put(0,0){\color{orange}
      \qbezier(39,150)(9,80)(-21,10)
      \qbezier(60,150)(150,80)(240,10)
      }
      \put(0,0){\color{blue}
      \qbezier(100,150)(70,80)(40,10)
      }
      \put(0,0){\color{purple}
      \qbezier(100,150)(190,80)(280,10)
      \qbezier(120,150)(150,80)(180,10)
      }
      \put(0,0){\color{blue}
      \qbezier(120,150)(240,80)(360,10)
      }
      
      \put(0,0){\color{orange}
      \qbezier(179,150)(149,80)(119,10)
      }
      \put(0,0){\color{blue}
      \qbezier(181,150)(151,80)(121,10)
      }
      \put(0,0){\color{orange}
      \qbezier(160,150)(190,80)(220,10)
      }
      \put(0,0){\color{blue}
      \qbezier(160,150)(280,80)(400,10)
      }
      
      \put(0,0){\color{purple}
      \qbezier(219,150)(189,80)(159,10)
      }
      \put(0,0){\color{cyan}
      \qbezier(221,150)(191,80)(161,10)
      }
      \put(0,0){\color{purple}
      \qbezier(300,150)(380,80)(460,10)
      }
      \put(0,0){\color{cyan}
      \qbezier(360,150)(270,80)(180,10)
      }
      \put(0,0){\color{purple}
      \qbezier(340,150)(380,80)(420,10)
      \drawline(360,150)(360,10)
      }
      
      \put(0,0){\color{blue}
      \qbezier(400,150)(310,80)(220,10)
      }
      \put(0,0){\color{purple}
      \drawline(400,150)(400,10)
      \qbezier(539,150)(569,80)(599,10)
      }
      \put(0,0){\color{blue}
      \qbezier(480,150)(450,80)(420,10)
      }
      
      \put(0,0){\color{cyan}
      \qbezier(520,150)(490,80)(460,10)
      \qbezier(541,150)(571,80)(601,10)
      }
      
      \put(0,0){\color{purple}
      \dottedline[$\cdot$]{2}(100,151)(120,151)
      \dottedline[$\cdot$]{2}(220,151)(240,151)
      \dottedline[$\cdot$]{2}(280,151)(300,151)
      \dottedline[$\cdot$]{2}(340,151)(360,151)
      \dottedline[$\cdot$]{2}(400,151)(420,151)
      \dottedline[$\cdot$]{2}(460,151)(480,151)
      \dottedline[$\cdot$]{2}(520,151)(540,151)
      \dottedline[$\cdot$]{2}(160,11)(180,11)
      \dottedline[$\cdot$]{2}(280,11)(300,11)
      \dottedline[$\cdot$]{2}(340,11)(360,11)
      \dottedline[$\cdot$]{2}(400,11)(420,11)
      \dottedline[$\cdot$]{2}(460,11)(480,11)
      \dottedline[$\cdot$]{2}(520,11)(540,11)
      \dottedline[$\cdot$]{2}(580,11)(600,11)
      }
      \put(0,0){\color{blue}
      \dottedline[$\cdot$]{2}(100,149)(120,149)
      \dottedline[$\cdot$]{2}(160,149)(180,149)
      \dottedline[$\cdot$]{2}(400,149)(420,149)
      \dottedline[$\cdot$]{2}(460,149)(480,149)
      \dottedline[$\cdot$]{2}(40,9)(60,9)
      \dottedline[$\cdot$]{2}(100,9)(120,9)
      \dottedline[$\cdot$]{2}(220,9)(240,9)
      \dottedline[$\cdot$]{2}(280,9)(300,9)
      \dottedline[$\cdot$]{2}(340,9)(360,9)
      \dottedline[$\cdot$]{2}(400,9)(420,9)
      }
      \put(0,0){\color{orange}
      \dottedline[$\cdot$]{2}(40,151)(60,151)
      \dottedline[$\cdot$]{2}(160,151)(180,151)
      \dottedline[$\cdot$]{2}(-20,11)(0,11)
      \dottedline[$\cdot$]{2}(40,11)(60,11)
      \dottedline[$\cdot$]{2}(100,11)(120,11)
      \dottedline[$\cdot$]{2}(220,11)(240,11)
      }
      \put(0,0){\color{cyan}
      \dottedline[$\cdot$]{2}(40,149)(60,149)
      \dottedline[$\cdot$]{2}(220,149)(240,149)
      \dottedline[$\cdot$]{2}(280,149)(300,149)
      \dottedline[$\cdot$]{2}(340,149)(360,149)
      \dottedline[$\cdot$]{2}(520,149)(540,149)
      \dottedline[$\cdot$]{2}(-20,9)(0,9)
      \dottedline[$\cdot$]{2}(160,9)(180,9)
      \dottedline[$\cdot$]{2}(460,9)(480,9)
      \dottedline[$\cdot$]{2}(520,9)(540,9)
      \dottedline[$\cdot$]{2}(580,9)(600,9)
      }
      
      \put( 40,150){\circle*{5}}
      \put( 60,150){\circle*{5}}
      \put(100,150){\circle*{5}}
      \put(120,150){\circle*{5}}
      \put(160,150){\circle*{5}}
      \put(180,150){\circle*{5}}
      \put(220,150){\circle*{5}}
      \put(240,150){\circle*{5}}
      \put(280,150){\circle*{5}}
      \put(300,150){\circle*{5}}
      \put(340,150){\circle*{5}}
      \put(360,150){\circle*{5}}
      \put(400,150){\circle*{5}}
      \put(420,150){\circle*{5}}
      \put(460,150){\circle*{5}}
      \put(480,150){\circle*{5}}
      \put(520,150){\circle*{5}}
      \put(540,150){\circle*{5}}

      \put(-50,162){\makebox(0,0){$A$}}
      \put( 40,162){\makebox(0,0){$\circ_A^1$}}
      \put( 60,162){\makebox(0,0){$\tl{\T{1}}$}}
      \put(100,162){\makebox(0,0){$\hd{\T{1}}$}}
      \put(120,162){\makebox(0,0){$\tl{\T{3}}$}}
      \put(160,162){\makebox(0,0){$\hd{\T{3}}$}}
      \put(180,162){\makebox(0,0){$\tl{\T{2}}$}}
      \put(220,162){\makebox(0,0){$\hd{\T{2}}$}}
      \put(240,162){\makebox(0,0){$\hd{\T{5}}$}}
      \put(280,162){\makebox(0,0){$\tl{\T{5}}$}}
      \put(300,162){\makebox(0,0){$\hd{\T{4}}$}}
      \put(340,162){\makebox(0,0){$\tl{\T{4}}$}}
      \put(360,162){\makebox(0,0){$\tl{\T{3}}$}}
      \put(400,162){\makebox(0,0){$\hd{\T{3}}$}}
      \put(420,162){\makebox(0,0){$\tl{\T{5}}$}}
      \put(460,162){\makebox(0,0){$\hd{\T{5}}$}}
      \put(480,162){\makebox(0,0){$\tl{\T{4}}$}}
      \put(520,162){\makebox(0,0){$\hd{\T{4}}$}}
      \put(540,162){\makebox(0,0){$\circ_A^2$}}
      
      \put(-20,10){\circle*{5}}
      \put(  0,10){\circle*{5}}
      \put( 40,10){\circle*{5}}
      \put( 60,10){\circle*{5}}
      \put(100,10){\circle*{5}}
      \put(120,10){\circle*{5}}
      \put(160,10){\circle*{5}}
      \put(180,10){\circle*{5}}
      \put(220,10){\circle*{5}}
      \put(240,10){\circle*{5}}
      \put(280,10){\circle*{5}}
      \put(300,10){\circle*{5}}
      \put(340,10){\circle*{5}}
      \put(360,10){\circle*{5}}
      \put(400,10){\circle*{5}}
      \put(420,10){\circle*{5}}
      \put(460,10){\circle*{5}}
      \put(480,10){\circle*{5}}
      \put(520,10){\circle*{5}}
      \put(540,10){\circle*{5}}
      \put(580,10){\circle*{5}}
      \put(600,10){\circle*{5}}

      \put(-50,-2){\makebox(0,0){$B$}}
      \put(-20,-2){\makebox(0,0){$\circ_B^1$}}
      \put(  0,-2){\makebox(0,0){$\tl{\T{1}}$}}
      \put( 40,-2){\makebox(0,0){$\hd{\T{1}}$}}
      \put( 60,-2){\makebox(0,0){$\tl{\T{6}}$}}
      \put(100,-2){\makebox(0,0){$\hd{\T{6}}$}}
      \put(120,-2){\makebox(0,0){$\tl{\T{2}}$}}
      \put(160,-2){\makebox(0,0){$\hd{\T{2}}$}}
      \put(180,-2){\makebox(0,0){$\tl{\T{3}}$}}
      \put(220,-2){\makebox(0,0){$\hd{\T{3}}$}}
      \put(240,-2){\makebox(0,0){$\tl{\T{1}}$}}
      \put(280,-2){\makebox(0,0){$\hd{\T{1}}$}}
      \put(300,-2){\makebox(0,0){$\tl{\T{7}}$}}
      \put(340,-2){\makebox(0,0){$\hd{\T{7}}$}}
      \put(360,-2){\makebox(0,0){$\tl{\T{3}}$}}
      \put(400,-2){\makebox(0,0){$\hd{\T{3}}$}}
      \put(420,-2){\makebox(0,0){$\tl{\T{4}}$}}
      \put(460,-2){\makebox(0,0){$\hd{\T{4}}$}}
      \put(480,-2){\makebox(0,0){$\tl{\T{1}}$}}
      \put(520,-2){\makebox(0,0){$\hd{\T{1}}$}}
      \put(540,-2){\makebox(0,0){$\tl{\T{3}}$}}
      \put(580,-2){\makebox(0,0){$\hd{\T{3}}$}}
      \put(600,-2){\makebox(0,0){$\circ_B^2$}}
    \end{picture}
  \end{center}
  \caption{\label{fig:capped-multi-adj-graph}Natural genomes $A=\tel\T{1}\z\T{3}\z\T{2}\z\rev{\T{5}}\z\rev{\T{4}}\z\T{3}\z\T{5}\z\T{4}\tel$ and~$B=\tel\T{1}\z\T{6}\z\T{2}\z\T{3}\z\T{1}\z\T{7}\z\T{3}\z\T{4}\z\T{1}\z\T{3}\tel$ and their capped multi-relational diagram $\MRC(A,B)$}
\end{figure}

A set $P \subseteq E_{\cp}$ is a \emph{capping set} if it does not contain any pair of incident edges.
A \emph{capped consistent decomposition} $Q[S,P]$ of $\MRC(A,B)$ is induced by a maximal sibling set $S \subseteq E_\ext$ and a maximal capping set $P \subseteq E_\cp$ and is composed of vertex disjoint cycles covering all vertices of $\MRC(A,B)$.  
We then have~$\ddcjid(Q[S,P]) = n_* + p_* - w(Q[S,P])$, where the weight of $Q[S,P]$ can be computed by the simple formula: 
\[
w(Q[S,P]) = c_Q -\sum_{C \in Q[S,P]}\!\! \lambda(C)\,.
\]

\begin{theorem}\label{thm:capped-decomposition}
Let $\mathbb{P}_\textsc{max}$ be the set of all maximal capping sets from $\MRC(A,B)$.
For each maximal sibling set $S$ of $\MR(A,B)$ and $\MRC(A,B)$, we have
\[
w(D[S]) = \max_{P \in \mathbb{P}_\textsc{max}}\{ w(Q[S,P])\}\,.
\]
\end{theorem}
\proof
 Recall that each maximal sibling set $S$ of $\MR(A,B)$ corresponds to a pair of matched singular genomes $\AMi{S}$ and $\BMi{S}$.
 Furthermore, in $\MRC(A,B)$, (i) each maximal capping set~$P$ corresponds to exactly $p_*$ caps, and (ii) all adjacencies, including the $|\kappa_A - \kappa_B|$ artificial adjacencies between cap extremities, are part of each consistent decomposition $Q[S,P]$.
Theorem~\ref{thm:capping-singular}~states that 
the pair of matched singular genomes $\AMi{S}$ and $\BMi{S}$
 can be optimally capped with $p_*$ caps and $|\kappa_A - \kappa_B|$ artificial adjacencies. 
 Therefore, it is clear that at least one optimal capping of $\AMi{S}$ and $\BMi{S}$ corresponds to a consistent decomposition of $\MRC(A,B)$, that is, $w(D[S]) = \max_{P \in \mathbb{P}_\textsc{max}}\{ w(Q[S,P])\}$.
\qed

As a consequence of Theorem~\ref{thm:capped-decomposition}, if $\mathbb{S}_\textsc{max}$ is the set of all maximal sibling sets and $\mathbb{P}_\textsc{max}$ is the set of all maximal caping-sets, we have
\[ 
\ddcjid(A, B) = n_* + p_* - \max_{S \in\mathbb{S}_\textsc{max},P \in\mathbb{P}_\textsc{max}}\{w(Q[S,P])\}\,.
\]  


Each decomposition $Q[S,P]$ corresponds to several capped versions of singular genomes $\AMi{S}$ and $\BMi{S}$, depending on how the cap extremities are paired. We do not need to identify the exact capped version we get, because all versions obtained with the same capping set $P$ give the same DCJ-indel distance. 

\section{\label{sec:ilp}An algorithm to compute the DCJ-indel distance of natural genomes}

\newcommand{\st}[1]{\{#1\}}

An ILP formulation for computing the distance of two balanced genomes $A$ and $B$ was given by \citet{SHA-LIN-MOR-2015}.
In this section we describe an extension of that formulation for computing the
DCJ-indel distance of natural genomes $A$ and $B$, based on consistent cycle
decompositions of $\MRC(A,B)$. The main difference is that here we need to
address the challenge of computing the indel-potential $\lambda(C)$ for each cycle $C$ of each decomposition. 
Note that a cycle $C$ of~$R(A,B)$ has either 0, or 1, or an even number of runs, therefore its indel-potential can be computed as follows:
\[
\lambda(C) = 
\begin{cases}
~~\Lambda(C)\:, & \mbox{ if $\Lambda(C)\leq1$}; \\[2mm]
\frac{\Lambda(C)}{2}+1\:, & \mbox{ if $\Lambda(C) \geq 2$}.
\end{cases}
\]

The formula above can be redesigned to a simpler one, that is easier to implement in the ILP. First, let a \emph{transition} in a decomposition $Q[S,P]$ be an indel-free path that is flanked by an indel edge from $\selfedges{A}$ and an indel-edge from $\selfedges{B}$. Each transition is part of some cycle $C$ of $Q[S,P]$ and we denote by $\aleph(C)$ the number of transitions in $C$. 
Observe that, if $C$ is indel-free, then obviously $\aleph(C) = 0$. If $C$ has a single run, then we also have $\aleph(C) = 0$. On the other hand, if $C$ has at least 2 runs, then $\aleph(C)=\Lambda(C)$. 
Our new formula is then split into a part that simply tests whether $C$ is indel-enclosing and a part that depends on the number of transitions $\aleph(C)$. 

\begin{proposition}
Given the function $r(C)$ defined as $r(C)=1$ if $\Lambda(C)\geq1$, otherwise $r(C)=0$, the indel-potential $\lambda(C)$ can be computed from the number of transitions $\aleph(C)$ with the formula 
\[
\lambda(C)=\frac{\aleph(C)}{2} + r(C)\,.
\]
\end{proposition}

Note that
$\sum_{C \in Q[S,P]}\! r(C) = c^{r}_Q + s_Q$, 
where $c^{r}_Q$ and $s_Q$ are the number of indel-enclosing $\AB$-cycles and
the number of circular singletons in $Q[S,P]$, respectively.
Furthermore, the number of transitions in $Q[S,P]$, given by the sum $\sum_{C \in Q[S,P]}\! \aleph(C)$ does not really need to be computed per cycle, that is, we can directly count the number of transitions in $Q[S,P]$ without keeping trace of which cycle each transition belongs to. We then denote by $\aleph_Q =\sum_{C \in Q[S,P]}\! \aleph(C)$ the number of transitions in $Q[S,P]$.

Now, we need to find a consistent decomposition $Q[S,P]$ of $\MRC(A,B)$ maximizing the weight
\[
w(Q[S,P]) = c_Q -\!\sum_{C \in Q[S,P]}\! \lambda(C) = c_Q - \!\left( c^{r}_Q + s_Q + \sum_{C \in Q[S,P]}\!\!\frac{\aleph(C)}{2}
\right)=c^{\tilde{r}}_Q-\frac{\aleph_Q}{2}-s_Q\,,
\]
where $c^{\tilde{r}}_Q=c_Q - c^{r}_Q$ is the number of indel-free $\AB$-cycles in $Q[S,P]$.

\subsection{\label{sec:ilp-details}ILP formulation}

Our formulation (shown in Algorithm~\ref{alg:ilp}) searches for an optimal consistent cycle decomposition of $\MRC(A,B) = (V,E)$, where the set of edges $E$ is the union of all disjoint sets of the distinct types of edges, $E=E_\ext\cup E_\cp\cup \adjedges{A}\cup\adjedges{B}\cup\selfedges{A}\cup\selfedges{B}$.

In the first part we use the same strategy as \citet{SHA-LIN-MOR-2015}.
A binary variable~$x_e$ (\texttt{D.01}) is introduced for every edge $e$, indicating whether $e$ is part of the computed decomposition. Constraint \texttt{C.01} ensures that adjacency edges are in all decompositions, Constraint \texttt{C.02} ensures that each vertex of each decomposition has degree 2, and Constraint \texttt{C.03} ensures that an extremity edge is selected only together with its sibling.
Counting the number of cycles in each
decomposition is achieved by assigning a unique identifier $i$ to each
vertex $v_i$ that is then used to label each cycle with the numerically
smallest identifier of any contained vertex (see Constraint \texttt{C.04}, Domain \texttt{D.02}). A vertex $v_i$ is then marked by variable $z_i$ (\texttt{D.03}) as representative of a cycle if its cycle label $y_i$ is equal to $i$ (\texttt{C.06}). However, unlike Shao~\textit{et al.}, we permit each variable $y_i$ to
take on value $0$ which, by Constraint~\texttt{C.05}, will be enforced whenever the
corresponding cycle is indel-enclosing. Since the smallest label of any vertex is $1$ (cf.\ \texttt{D.02}), any cycle with label 0 will not
be counted. 

The second part is our extension for counting transitions. We introduce binary variables $r_v$ (\texttt{D.04}) to label runs.
To this end, Constraint \texttt{C.07} ensures that each vertex $v$ is labeled
$0$ if $v$ is part of an $\ma$-run and otherwise it is labeled $1$ indicating its
participation in a $\mb$-run. Transitions between $\ma$- and $\mb$-runs in a
cycle are then recorded by binary variable $t_e$ (\texttt{D.05}). If a label change
occurs between any neighboring pair of vertices $u,v\in V$ of a cycle, Constraint \texttt{C.08} causes
transition variable $t_{\st{u, v}}$ to be set to 1.
We avoid an excess of co-optimal solutions by canonizing the
locations in which transitions are observed. More specifically, Constraint \texttt{C.09} prohibits
label changes in adjacencies not directly connected to an indel and Constraint \texttt{C.10} in edges other than adjacencies of genome $A$, resulting in the transition being observed as close to the $A$-run as possible.

In the third part we add a new constraint and a new domain to our ILP, so that we
can count the number of circular singletons. Let $K$ be the circular chromosomes in both genomes and $E^k_{id}$ be the set of indel edges of a circular chromosome $k\in K$.
For each circular chromosome we introduce a decision variable $s_k$ (\texttt{D.06}), that is 1 if $k$ is a circular singleton and 0 otherwise.
A circular chromosome is then a singleton if all its indel edges are set (see Constraint \texttt{C.11}). Only in that case the left side of the inequality will take on value $1$ and enforces $s_k$ to be set to $1$ as well.

The objective of our ILP is to maximize the weight of a consistent
decompositon, that is equivalent to maximizing the number of indel-free cycles,
counted by the sum over variables $z_i$, while simultaneously minimizing 
the number of transitions in indel-enclosing $\AB$-cycles,
calculated by half the sum over variables $t_e$, and
the number of circular singletons, calculated by the sum over variables $s_k$.

\begin{algorithm}
	\caption{\label{alg:ilp}ILP for the computation of the DCJ-indel distance of natural genomes}
    
    \small
    
    \medskip
    \noindent\textbf{Objective:}\\
    \hspace*{6em}\texttt{Maximize} $\displaystyle\sum_{1\leq i\leq|V|}\!\!\! z_i - \frac{1}{2} \displaystyle\sum_{e \in E} t_e - \sum_{k \in K}s_k$
    \medskip

        \setlength{\tabcolsep}{3pt}\renewcommand{\arraystretch}{1.3}
        \begin{tabular}{lll}
	        \multicolumn{3}{l}{\textbf{Constraints:}}\\
            \;(\texttt{C.01}) & $x_e = 1$                              & $\forall~e \in \adjedges{A} \cup \adjedges{B}$ \\
            \;(\texttt{C.02}) & $\displaystyle\sum_{\{u,v\}\in E}\!\!\! x_{\{u,v\}} = 2$ & $\forall~u \in V$ \\
            \;(\texttt{C.03}) & $x_e = x_d$                            & $\forall~e,d \in E_\ext \text{ such that }$\\[-1mm]
                              &                                        & $\text{$e$ and $d$ are siblings}$\\
            \;(\texttt{C.04}) & $y_i \leq y_j + i(1-x_{\{v_i, v_j\}})$ & $ \forall~ \{v_i, v_j\}\in E$ \\
            \;(\texttt{C.05}) & $y_i \leq i (1 - x_{\{v_i,v_j\}})$     & $\forall~\st{v_i,v_j} \in \selfedges{A}\cup \selfedges{B}$\\
            \;(\texttt{C.06}) & $i\cdot z_i \leq y_i$                  & $ \forall~1 \leq i \leq |V|$\\
            \;(\texttt{C.07}) & $r_v \leq 1 - x_{\st{u, v}}$           & $ \forall~\st{u, v}\in \selfedges{A}\,,$\\
                              & $r_{v'} \geq x_{\st{u',v'}}$           & $ \forall~\st{u',v'} \in \selfedges{B}$ \\
            \;(\texttt{C.08}) & $t_{\{u, v\}} \geq r_v - r_u - (1-x_{\{u, v\}})$  & $\forall~\st{u, v} \in E$ \\
            \;(\texttt{C.09}) & $\displaystyle\sum_{\substack{d \in \selfedges{A},\\ d \cap e \neq \varnothing}}
            \!\!\!x_{d} - t_{e} \geq 0$ & $\forall~e \in \adjedges{A}$\\
            \;(\texttt{C.10}) & $t_e = 0$                              & $\forall~e \in E\setminus \adjedges{A}$ \\ 
            \;(\texttt{C.11}) & $ \displaystyle\sum_{e \in E^k_{id}} x_e - |k| + 1\leq s_k$ & $\forall k \in K $ \medskip\\
            \multicolumn{3}{l}{\textbf{Domains:}}\\
            \;(\texttt{D.01}) & $x_e \in \{0, 1\}$                      & $\forall~e \in E$ \\
            \;(\texttt{D.02}) & $0~ \leq y_i \leq i$                    & $\forall~1 \leq i \leq |V|$ \\
            \;(\texttt{D.03}) & $z_i \in  \{0, 1\}$                     & $\forall~1 \leq i \leq |V|$ \\
            \;(\texttt{D.04}) & $r_v \in \{0, 1\} $                     & $\forall~v \in V$ \\
            \;(\texttt{D.05}) & $t_e \in \{0, 1\} $                     & $\forall~e \in E$ \\
            \;(\texttt{D.06}) & $s_k \in \{0,1\}$                       & $\forall~k \in K$ \\
        \end{tabular}
        \medskip
\end{algorithm}

\paragraph{Implementation.}
We implemented the construction of the ILP as a python application, available at \url{https://gitlab.ub.uni-bielefeld.de/gi/ding}.

\paragraph{Comparison to the approach by Lyubetsky~\textit{et al.}}
As mentioned in the Introduction, another ILP for the comparison of genomes with  unequal content and paralogs was presented by \citet{LYU-GER-GOR-2017}.
In order to compare our method to theirs, we ran our ILP using CPLEX on a single thread with the two small artificial examples given in that paper on page~8.
The results in terms of DCJ distance were the same. A comparison of running times is presented in Table~\ref{tab:lyubetsky}.

\begin{table}
\caption{\label{tab:lyubetsky}Comparison of running times and memory usage to the ILP in~\citep{LYU-GER-GOR-2017}.}
\begin{center}
\renewcommand{\tabcolsep}{5pt}
\footnotesize
\begin{tabular}{cccccc}
\toprule
 & & \bf\#marker & \bf running time as reported & \bf our & \bf our peak\\
\raisebox{1.3ex}[0ex]{\bf dataset} & 
\raisebox{1.3ex}[0ex]{\bf \#markers}
& \bf occurrences & \bf by~\citet{LYU-GER-GOR-2017} & \bf running time & \bf memory\\
\midrule
Example 1 & 5/5 & 9/9 & ``about 1.5h'' & .16s & 13200kb\\
Example 2 & 10/10 & 11/11 & ``about 3h'' & .05s& 13960kb\\
\bottomrule
\end{tabular}
\end{center}
\end{table}

\subsection{\label{sec:eva}Performance benchmark}

For benchmarking purposes, we used Gurobi 9.0 as solver. In all our experiments, we ran Gurobi on a single thread.

\paragraph{Generation of simulated data.}
Here we describe our simulation tool that is included in our software repository (\url{https://gitlab.ub.uni-bielefeld.de/gi/ding}) 
and used for evaluating the performance of our ILP implementation.

Our method samples marker order sequences over a user-defined phylogeny. However, here we restrict our simulations to pairwise comparisons generated over rooted, weighted trees of two leaves. Starting from an initial marker order sequence of user-defined length (i.e., number of markers), the simulator samples Poisson-distributed DCJ events with expectation equal to the corresponding edge weights. Likewise, insertion, deletion and duplication events of one or more consecutive markers are sampled, yet, their frequency is additionally dependent on a \emph{rate factor} that can be adjusted by the user. The length of each segmental insertion, deletion, and duplication is drawn from a Zipfian distribution, whose parameters can also be adjusted by the user. At each internal node of the phylogeny, the succession of mutational operations is performed in the following order: DCJ operations, duplications, deletions, insertions. To this end, cut points, as well as locations for insertions, deletions and duplications are uniformly drawn over the entire genome. 

In our simulations, we used $s=4$ for Zipfian distributions of insertions and deletions, and $s=6$ for duplications. Unless specified otherwise, insertion and deletion rates were set to be $0.1$ and $0.2$ respectively. We set the length of the root genome to 20{,}000 marker occurrences.

\paragraph{Evaluating the impact of the number of duplicate occurrences.}
In order to evaluate the impact of the number of duplicate occurrences on the running time,
we first keep the number of simulated DCJ events fixed to $10{,}000$ and  vary parameters that affect the number of duplicate occurrences.

Our ILP solves the decomposition problem efficiently for real-sized genomes under small to moderate numbers of duplicate occurrences: solving times for genome pairs with less than $10{,}000$ duplicate occurrences ($\sim 50\%$ of the genome size) shown in Figure~\ref{fig:solv_t}~(i) are with few exceptions below $5$ minutes and exhibit a linear increase, but solving time is expected to increase dramatically with higher numbers of duplicate occurrences. To further exploit the conditions under which the ILP is no longer solvable with reasonable compute resources we continued the experiment with even higher amounts of duplicate occurrences and instructed Gurobi to terminate within 1 hour of computation. We then partitioned the simulated data set into 8 intervals of length 500 according to the observed number of duplicate occurrences.  For each interval, we determined the average as well as the maximal multiplicity of any duplicate marker and examined the average \emph{optimality gap}, i.e., the difference in percentage between the best primal and the best dual solution computed within the time limit. The results are shown in Table~\ref{tab:gaps} and emphasize the impact of duplicate occurrences in solving time: below 14,000 duplicate occurrences, the optimality gap remains small and sometimes even the exact solution is computed, whereas above that threshold the gap widens very quickly.


\begin{figure}

\begin{minipage}{.45\linewidth}
\centering
{\bf (i)}\\
\includegraphics[height=5.5cm]{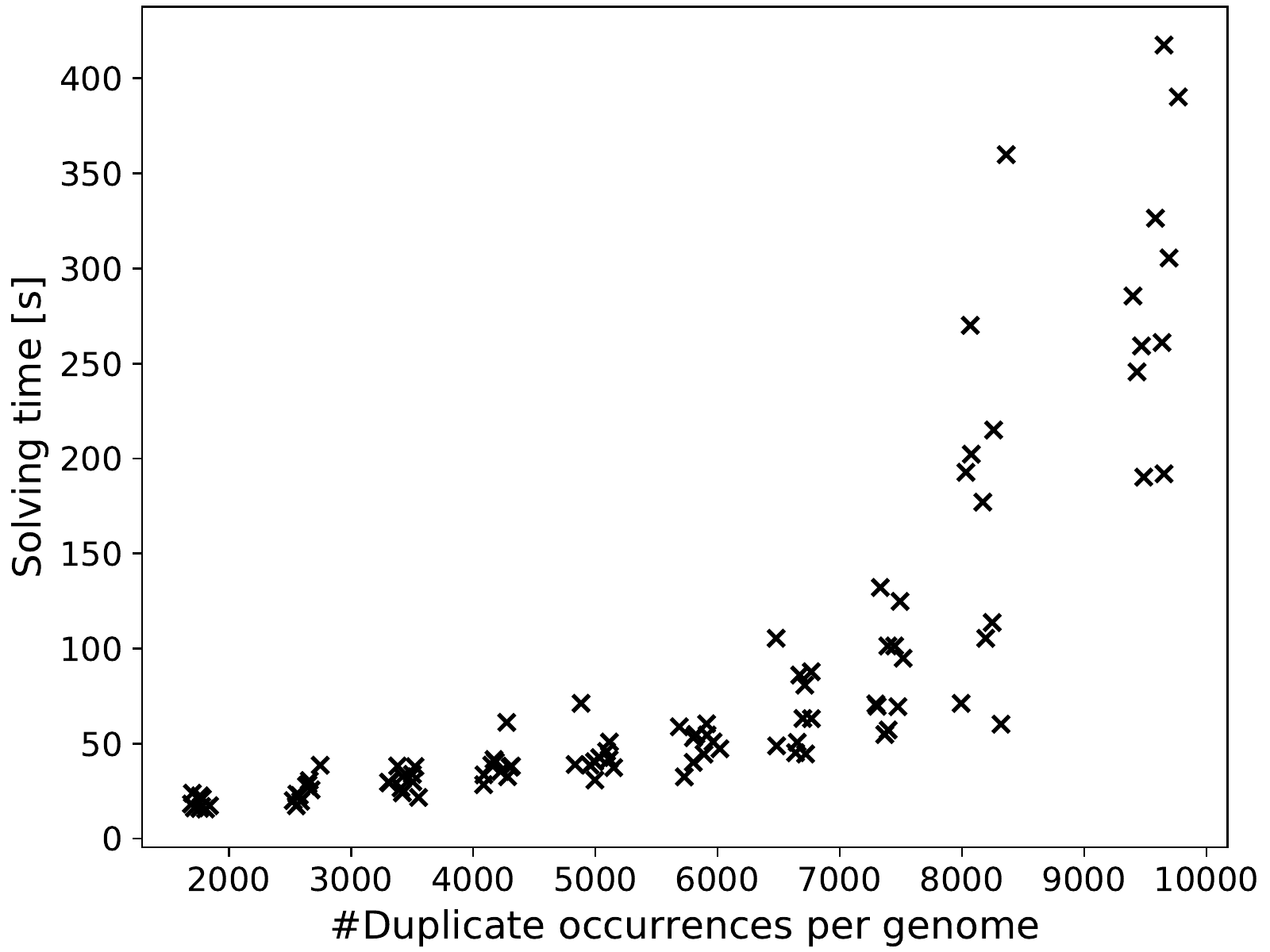}
\end{minipage}
\hspace{8mm}
\begin{minipage}{.45\linewidth}
\centering
{\bf (ii)}\\
\includegraphics[height=5.5cm]{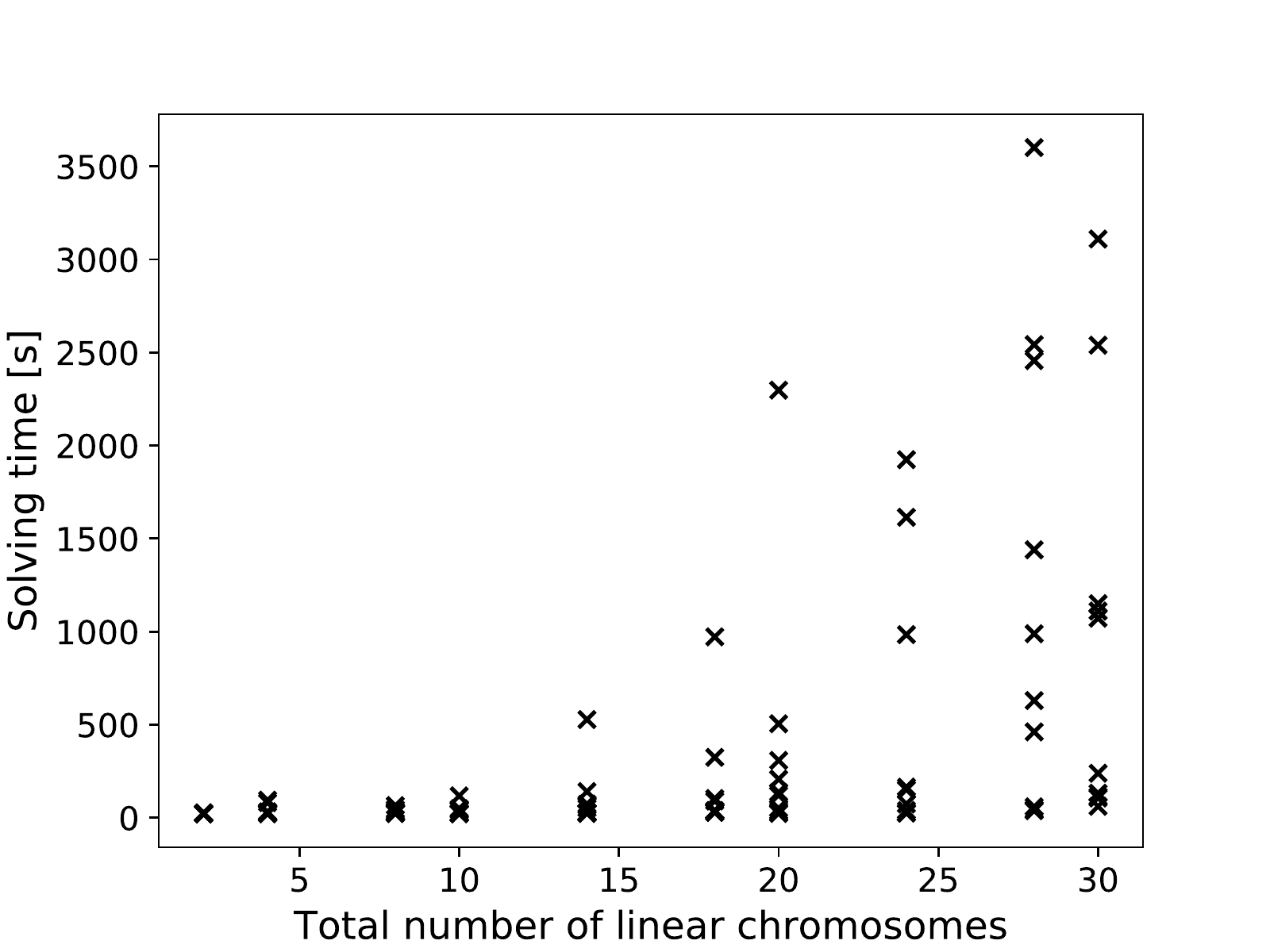}
\end{minipage}

\bigskip

\begin{minipage}{.45\linewidth}
\centering
{\bf (iii)}\\ 
\includegraphics[height=5.5cm]{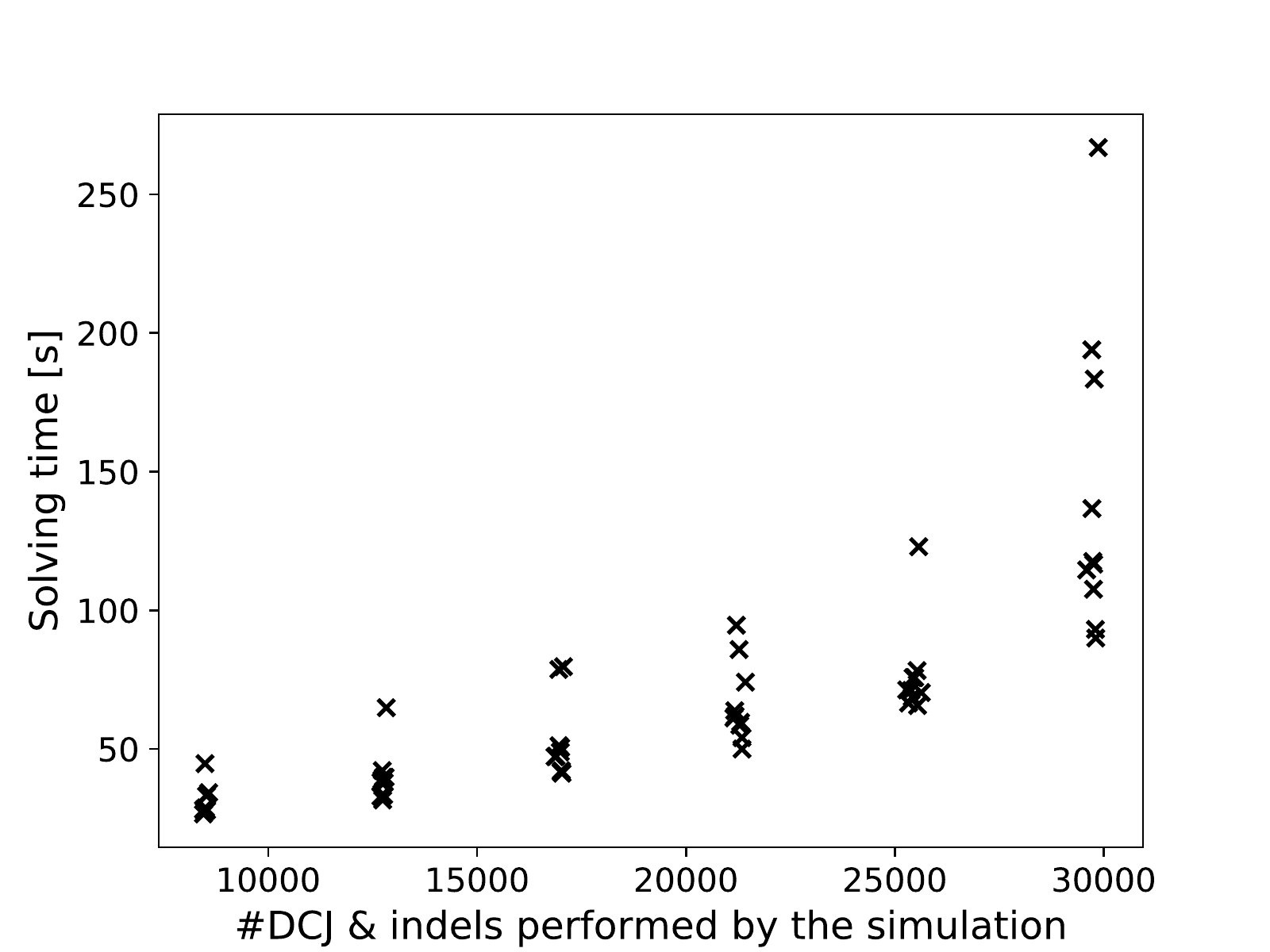}
\end{minipage}
\hspace{8mm}
\begin{minipage}{.45\linewidth}
\centering
{\bf (iv)}
\includegraphics[height=5.5cm]{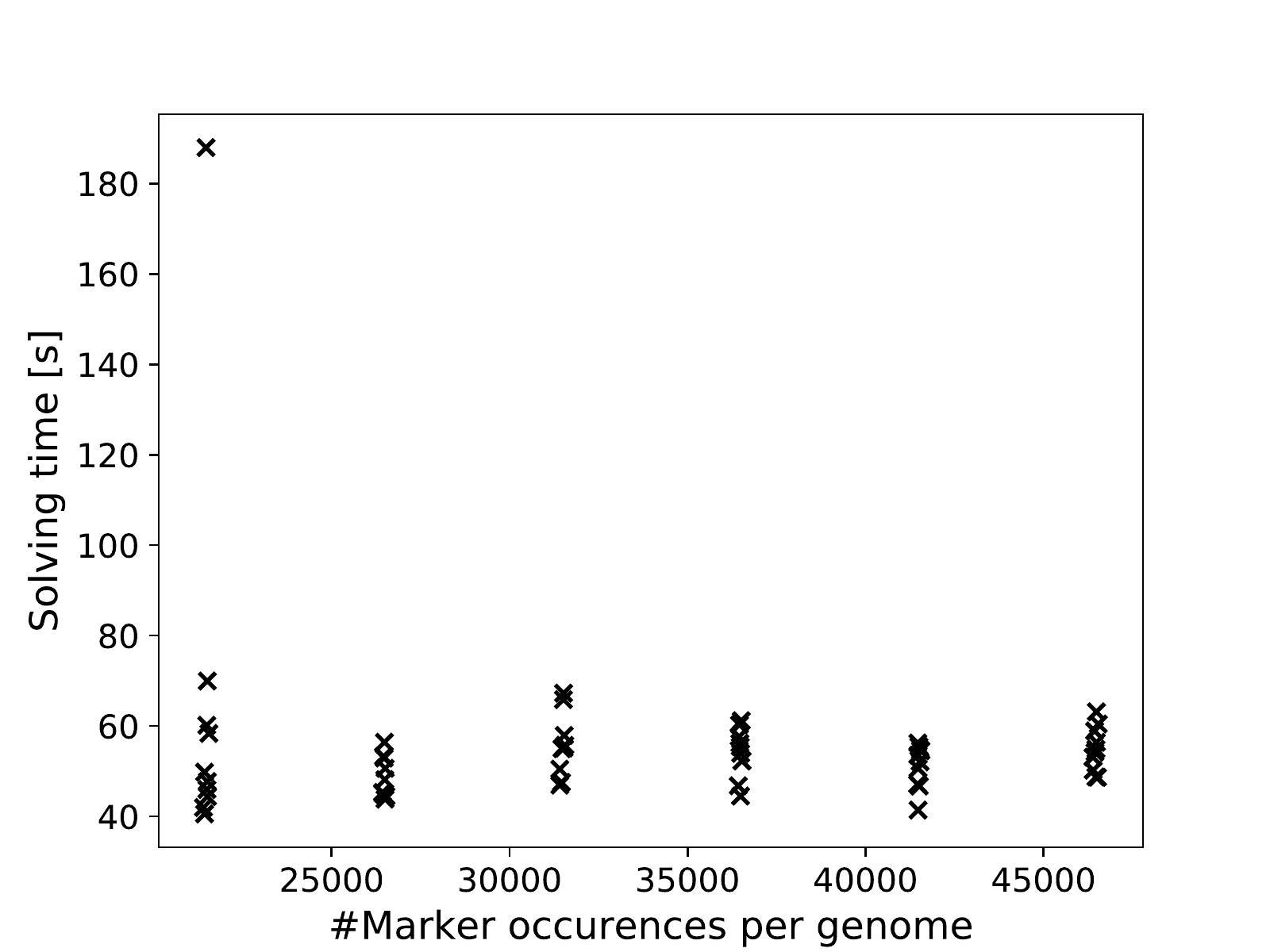}
\end{minipage}

\caption{\label{fig:solv_t}Solving times for: (i) genomes with varying number of duplicate occurrences, totaling 20,000 marker occurrences per genome; (ii) genome pairs with varying number of linear chromosomes with 20,000 marker occurrences per genome;
(iii) varying number of DCJs and indels applied by the simulation to genomes of  $\sim 35{,}000$ marker occurrences; and (iv)
genome pairs with varying total number of marker occurrences from both genomes.}
\end{figure}

\begin{table}
\caption{\label{tab:gaps}Average optimality gap for simulated genome pairs grouped by number of duplicate occurrences after 1h of running time.}
\begin{center}
\setlength{\tabcolsep}{5pt}
\footnotesize
\begin{tabular}{C{2.5cm}C{3cm}C{2.5cm}C{2.5cm}}
\toprule
\bf \#duplicate occurrences & \bf avg.\ mult.\ of dupl.\ markers & \bf max.\ multiplicity & \bf avg.\ opt.\ gap~(\%)\\
\midrule
11500..11999& 2.206&8& 0.000\\
12000..12499& 2.219&8& 0.031\\
12500..12999& 2.217&7& 0.025\\
13000..13499& 2.233&9& 0.108\\
13500..13999& 2.247&8& 0.812\\
14000..14499& 2.260&8& 1.177\\
14500..14999& 2.274&8& 81.865\\
15000..15499& 2.276&9& 33.102\\
\bottomrule
\end{tabular}
\end{center}
\end{table}

\paragraph{Evaluating additional parameters.}
So far we examined only the impact of duplicates on solving times of our program. However, other parameters of our experiment are expected to have an effect on the solving times, too. 
We ran three experiments, in each varying one of the following parameters while keeping the others fixed: (i) genome size, (ii) number of simulated DCJs and indels, and (iii) number of chromosomes. The duplication rate was fixed at $0.4$ for these experiments and the running time was limited to 1 hour.

The results, shown in Figure~\ref{fig:solv_t}, parts~(ii), (iii) and~(iv),
indicate that the number of linear chromosomes plays a major factor in the solving time. At the same time, solving times vary more widely with increasing chromosome number.
The latter has a simple explanation: telomeres, represented as caps in the multi-relational diagram, behave in the same way as duplicate occurrences of the same marker do. Increasing their number (by increasing the number of linear chromosomes) increases exponentially the search space of matching possibilities.

Conversely, the number of simulated DCJs and indels has a minor impact on the solving times of our simulation runs. However, while initially exhibiting collinearity, the solving times for higher numbers of DCJs and indels divert super-linearly. Lastly, the genome size has a negligible effect on solving time within the tested range of $20{,}000$ to $50{,}000$ marker occurrences.

\subsection{Real data analysis}
In order to demonstrate the applicability to real data sets, we compared the genomes of six \textit{Drosophila} species and reconstructed their phylogeny from pairwise DCJ-indel distances.
The species names and the NCBI accession numbers of the assemblies are listed in Table~\ref{tab:drosophilas}.

\begin{table}
\caption{\label{tab:drosophilas}List of genome assemblies used in our experiments}
\begin{center}
\footnotesize
\begin{tabular}{llcc}
\toprule
\bf Species & \bf NCBI Assembly & \bf \#genes & \bf \#segments\\
\midrule
\textit{Drosophila busckii} & ASM1175060v1 & 11371 & 23285\\
\textit{Drosophila melanogaster} & Release 6 plus ISO1 MT & 13048 & 62415\\
\textit{Drosophila pseudoobscura} & UCI\_Dpse\_MV25 & 13399 & 46692\\
\textit{Drosophila sechellia} & ASM438219v1 & 13037 & 60855\\
\textit{Drosophila simulans} & ASM75419v2 & 13023 & 59520\\
\textit{Drosophila yakuba}  & dyak\_caf1 & 12835 & 60946\\
\bottomrule
\end{tabular}
\end{center}
\end{table}

We used two types of markers.
In our first experiment, 
markers correspond to the longest annotated coding sequences (CDSs) per locus obtained from the respective NCBI annotations.
Their numbers are listed in Table~\ref{tab:drosophilas} as well.
Subsequently, we inferred \emph{hierarchical orthologous groups (HOGs)} of these markers with \textit{D.~busckii} being the outgroup species
running OMA standalone version 2.4.1~\citep{ALT-LEV-ZAR-TOM-XXX-2019} with default settings.
As before, we used Gurobi in computing pairwise DCJ-indel distances.
As can be seen in Table~\ref{tab:distdrosophila}, Gurobi was able to solve most instances within seconds with the exception of one pair, which took about 9 hours to compute, emphasizing again the sensitivity of the ILP's solving time to the number of duplicate occurrences.

\begin{table}
\caption{\label{tab:distdrosophila}Pairwise comparisons of the six \textit{Drosophila} species (\textit{busckii} (dbus), \textit{melanogaster} (dmel), \textit{pseudoobscura} (dpse), \textit{sechellia} (dsec),  \textit{simulans} (dsim) and \textit{yakuba} (dyak)). Genomes were constructed using genes as markers. All instances were solved by Gurobi on a single thread.}
\begin{center}
\footnotesize
\begin{tabular}{C{2cm}C{3.2cm}C{2cm}C{2cm}C{2cm}C{2cm}}
\toprule
\bf Genome pair&\bf Max. multiplicity of dupl. marker&\bf \#duplicate markers&\bf \#duplicate occ. &\boldmath$d_{DCJ}^{id}$ &\bf solving time [s] \\ 
\midrule
dbus-dmel&23&303&832&4661&6.02\\
dbus-dpse&17&361&934&4688&5.29\\
dbus-dsec&15&295&766&4710&5.64\\
dbus-dsim&13&281&721&4767&5.05\\
dbus-dyak&19&318&785&4756&5.00\\
dmel-dpse&23&469&1319&3799&32218.93\\
dmel-dsec&23&326&902&901&6.78\\
dmel-dsim&23&322&893&1093&5.73\\
dmel-dyak&23&362&972&1379&7.22\\
dpse-dsec&17&464&1227&3866&13.82\\
dpse-dsim&17&449&1198&3962&6.81\\
dpse-dyak&19&481&1259&3951&8.96\\
dsec-dsim&15&314&843&1138&5.67\\
dsec-dyak&19&354&903&1516&6.56\\
dsim-dyak&19&347&864&1661&23.07\\
\bottomrule
\end{tabular}
\end{center}
\end{table}

Using Neighbor Joining in \textit{MEGA X}~\citep{KUM-STE-LI-KNY-2018}, we constructed a phylogeny of the considered species.
The tree rooted by \textit{D.~busckii} is shown in Figure~\ref{fig:treedrosophila}~(i).
It is consistent with the knowledge on the \textit{Drosophila} phylogeny so far, except for the resolution of the subtree containing the taxa \textit{melanogaster}, \textit{sechellia} and \textit{simulans}.
Considering the corresponding Splits diagram constructed by \textit{NeighborNet} in \textit{SplitsTree}~\citep{HUS-BRY-2005} (see Figure~\ref{fig:treedrosophila}~(ii)), we observe that the distances in this subtree do not behave very tree-like. This suggests that, rather than an erroneous tree being computed, the resolution of the gene-based inference of markers simply does not provide distances for meaningfully clustering any two of the three taxa together.

\begin{figure}
\begin{minipage}{.45\linewidth}
\centering
{
\bf (i)}\\
\includegraphics[scale=0.32]{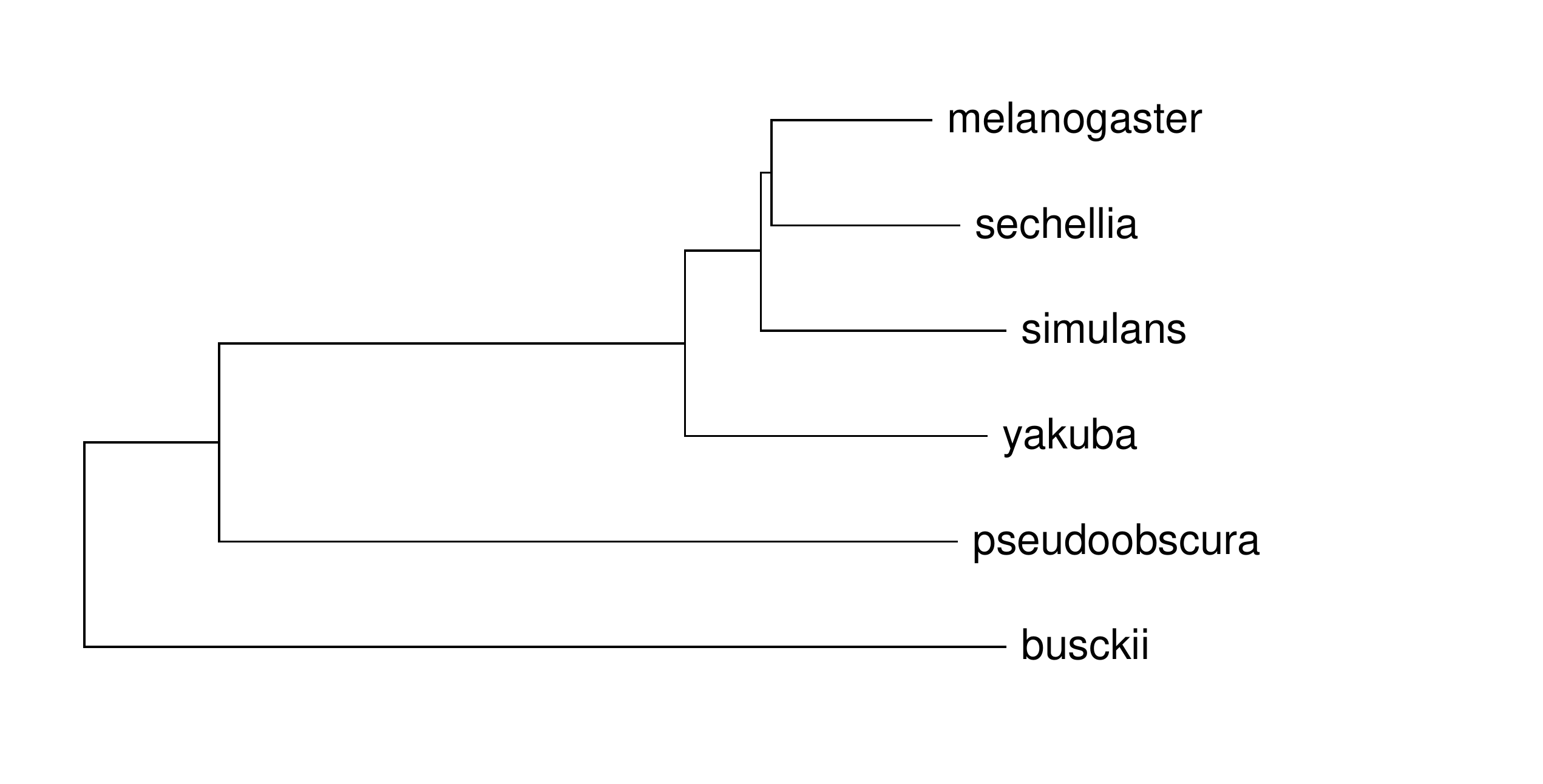}
\end{minipage}
\hspace{8mm}
\begin{minipage}{.45\linewidth}
\centering
{

\bf (ii)}\\
\includegraphics[scale=0.15]{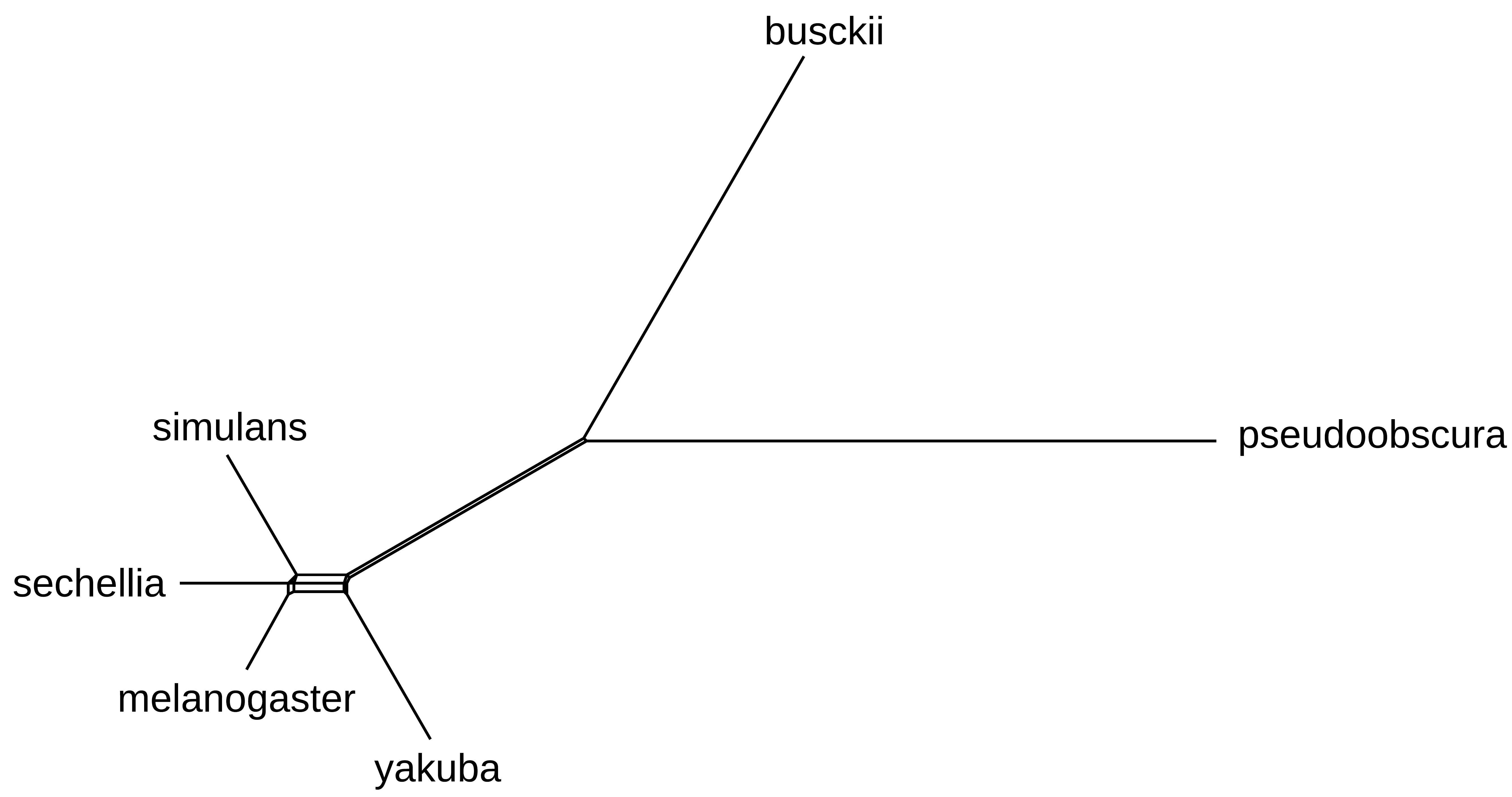}
\end{minipage}
\caption{\label{fig:treedrosophila}The gene-based distances in Table~\ref{tab:distdrosophila} are used as input to reconstruct \textit{Drosophila}-Phylogeny: (i) with Neighbor Joining; and (ii) as a Splits diagram.}
\end{figure}

In order to increase coverage and resolution, we also generated a second set of markers, directly from the genomic sequences and not restricted to genes or CDSs. We used GEESE~\citep{RUB-MAR-STO-DOE-2020} to construct genomic markers of length at least 500bp. GEESE implements a heuristic for the genome segmentation problem~\citep{Visnovska:2013ua} and takes as input local pairwise sequence alignments that we computed with LASTZ. The parameter settings used for both GEESE and LASTZ are detailed in our software repository.
The number of markers in each genome are shown in Table~\ref{tab:drosophilas}.
Again using Gurobi on a single thread, we were able to solve all corresponding instances of the ILP within a few minutes.
The distances as well as data concerning duplicates and solving times can be found in Table~\ref{tab:distdrosophila_seg}.

Using the same procedure as above to construct the Neighbor Joining tree and the Splits diagram (see Figure~\ref{fig:treedrosophila_seg}, parts~(i) and~(ii), respectively) we find that the segmentation-based approach not only produces the correct topology of the tree, but also improves the strength of all correct splits in the previously problematic subtree, including those involving \textit{D.~yakuba}.
We notice however that the branch length of \textit{D.~busckii} is comparatively short.
This is most likely due to the lack of markers which could be inferred on the \textit{D.~busckii} genome (see Table~\ref{tab:drosophilas}), thus leading to some rearrangements being missed.
One might attribute the fact that the segmentation did not infer many homologies in this case to more rapid sequence evolution in non-coding regions.

\begin{table}
\caption{\label{tab:distdrosophila_seg}Pairwise comparisons of the six \textit{Drosophila} species (\textit{busckii} (dbus), \textit{melanogaster} (dmel), \textit{pseudoobscura} (dpse), \textit{sechellia} (dsec),  \textit{simulans} (dsim) and \textit{yakuba} (dyak)). Genomes were constructed using segmentation. All instances were solved by Gurobi on a single thread.}
\begin{center}
\footnotesize
\begin{tabular}{C{2cm}C{3.2cm}C{2cm}C{2cm}C{2cm}C{2cm}}
\toprule
\bf Genome pair&\bf Max. multiplicity of dupl. marker&\bf \#duplicate markers&\bf \#duplicate occ. &\boldmath$d_{DCJ}^{id}$ &\bf solving time [s] \\ 
\midrule
dbus-dmel&15&582&1439&13965&31.77\\
dbus-dpse&46&675&1882&14329&29.94\\
dbus-dsec&15&578&1429&13877&41.97\\
dbus-dsim&15&545&1349&13822&23.69\\
dbus-dyak&15&615&1508&13801&12.51\\
dmel-dpse&43&952&2480&18660&20.38\\
dmel-dsec&43&1166&3126&5137&378.21\\
dmel-dsim&14&1045&2561&4791&23.85\\
dmel-dyak&34&1697&3896&7384&35.41\\
dpse-dsec&22&966&2511&18469&22.72\\
dpse-dsim&14&897&2264&18362&19.39\\
dpse-dyak&46&1174&3109&18602&19.64\\
dsec-dsim&23&1228&3151&3403&29.61\\
dsec-dyak&23&1701&3908&7361&27.95\\
dsim-dyak&14&1562&3492&7141&30.20\\

\bottomrule
\end{tabular}
\end{center}
\end{table}

\begin{figure}
\begin{minipage}{.45\linewidth}
\centering
{
\bf (i)}\\
\includegraphics[scale=0.32]{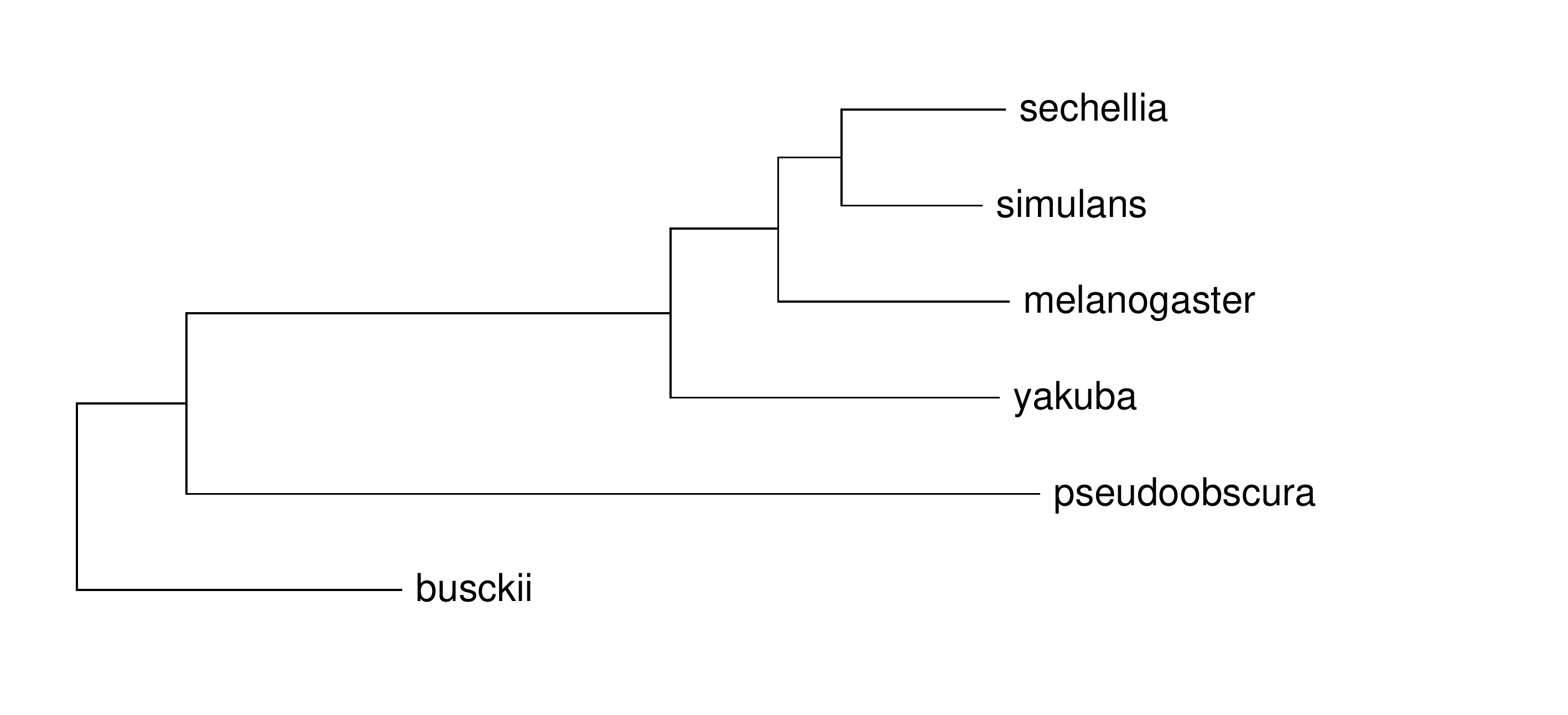}
\end{minipage}
\hspace{8mm}
\begin{minipage}{.45\linewidth}
\centering
{
\bf (ii)}\\
\includegraphics[scale=0.15]{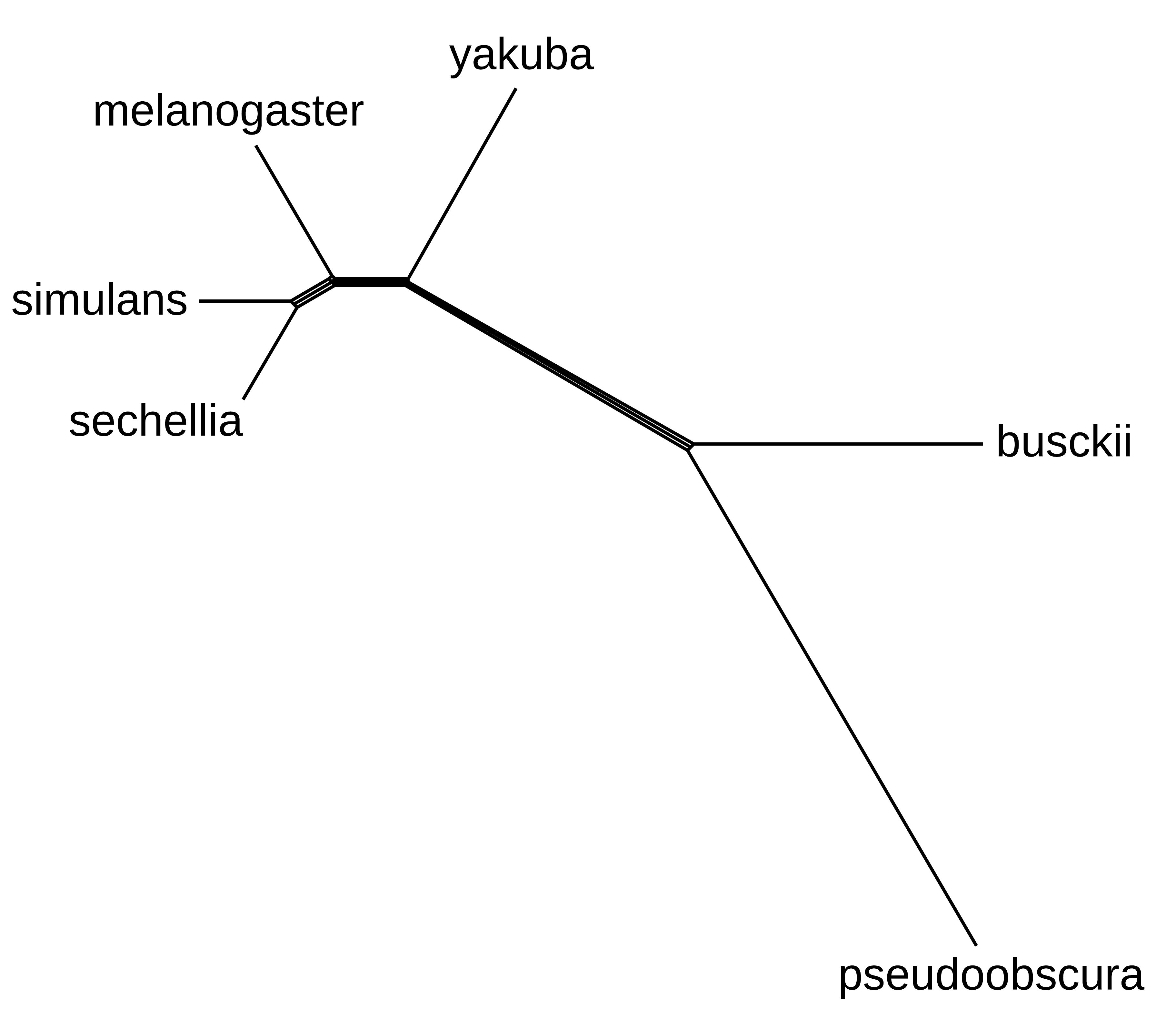}
\end{minipage}
\caption{\label{fig:treedrosophila_seg}The segmentation-based distances in Table~\ref{tab:distdrosophila_seg} are used as input to reconstruct \textit{Drosophila}-Phylogeny: (i) with Neighbor Joining; and (ii) as a Splits diagram.}
\end{figure}



\section{Conclusion}

By extending the DCJ-indel model to allow for duplicate markers, we introduced a rearrangement model that is capable of handling \emph{natural genomes}, i.e., genomes that contain shared, individual, and duplicated markers. In other words, under this model genomes require no further processing nor manipulation once genomic markers and their homologies are inferred.
The DCJ-indel distance of natural genomes being NP-hard, we presented a fast method for its calculation in form of an integer linear program. Our program is capable of handling real-sized genomes, as evidenced in simulation and real data experiments. It can be applied universally in comparative genomics and enables uncompromising analyses of genome rearrangements.

Our experiments on real data show that our approach is easily applicable to real world genomes with markers generated by different methods.
The power of the method however depends on the quality of the markers.
Genes as markers proved reliable in resolving distances and relations between further related taxa while not being expressive enough to resolve some closer relations.
In contrast, segmentation-based markers are better suited to resolve close distances, but might underestimate larger distances due to lack of markers.

We hope that similar analyses will provide further insights into the underlying mutational mechanisms of other, less well studied species.
Conversely, we expect the here presented model to be extended and specialized in future to reflect the insights gained by these analyses.
A follow up work with a \emph{family-free} version of our model has just appeared in the Proceedings of the Workshop on Algorithms in Bioinformatics (WABI 2020)~\citep{RUB-MAR-BRA-2020}. 


\begin{thebibliography}{25}
\expandafter\ifx\csname natexlab\endcsname\relax\def\natexlab#1{#1}\fi

\bibitem[Altenhoff {\em et~al.\/}, 2019]{ALT-LEV-ZAR-TOM-XXX-2019}
Altenhoff, A.~M., Levy, J. \& {Zarowiecki, et al.}, M. (2019).
\newblock {OMA standalone}: orthology inference among public and custom genomes
  and transcriptomes.
\newblock {\em Genome Res.\/}, {\bf 29}, 1152--1163.

\bibitem[Angibaud {\em et~al.\/}, 2009]{ANG-FER-RUS-THE-VIA-2009}
Angibaud, S., Fertin, G. \& {Rusu, et al.}, I. (2009).
\newblock On the approximability of comparing genomes with duplicates.
\newblock {\em J. Graph Alg. Appl.\/}, {\bf 13}, 19--53.
\newblock (A preliminary version appeared in Proc.~of WALCOM 2008.).

\bibitem[Bergeron {\em et~al.\/}, 2006]{BER-MIX-STO-2006}
Bergeron, A., Mixtacki, J. \& Stoye, J. (2006).
\newblock A unifying view of genome rearrangements.
\newblock In {\em Proceedings of the 6th International Conference on Algorithms
  in Bioinformatics (WABI 2006)\/}, volume 4175 of {\em LNBI\/}. Springer
  Verlag, pp. 163--173.

\bibitem[Boh\-nenk\"amper {\em et~al.\/}, 2020]{BOH-BRA-DOE-STO-2020a}
Boh\-nenk\"amper, L., Braga, M. D.~V. \& {Doerr, et al.}, D. (2020).
\newblock Computing the rearrangement distance of natural genomes.
\newblock In Schwartz, R., (ed.) {\em Proceedings of the 24th International
  Conference on Research in Computational Molecular Biology, RECOMB 2020\/},
  volume 12074 of {\em LNCS\/}. Springer Verlag, pp. 3--18.

\bibitem[Braga \& Stoye, 2010]{BRA-STO-2010}
Braga, M. D.~V. \& Stoye, J. (2010).
\newblock The solution space of sorting by {DCJ}.
\newblock {\em J. Comput. Biol.\/}, {\bf 17}, 1145--1165.
\newblock (A preliminary version appeared in Proc.~of RECOMB-CG 2009.).

\bibitem[Braga {\em et~al.\/}, 2011]{BRA-WIL-STO-2011}
Braga, M. D.~V., Willing, E. \& Stoye, J. (2011).
\newblock Double cut and join with insertions and deletions.
\newblock {\em J. Comput. Biol.\/}, {\bf 18}, 1167--1184.
\newblock (A preliminary version appeared in Proc.~of WABI 2010.).

\bibitem[Bryant, 2000]{BRY-2000}
Bryant, D. (2000).
\newblock The complexity of calculating exemplar distances.
\newblock In Sankoff, D. \& Nadeau, J.~H., (eds.) {\em Comparative Genomics\/}.
  Kluwer Academic Publishers, pp. 207--211.

\bibitem[Bulteau \& Jiang, 2013]{BUL-JIA-2013}
Bulteau, L. \& Jiang, M. (2013).
\newblock Inapproximability of (1,2)-exemplar distance.
\newblock {\em IEEE/ACM Trans. Comput. Biol. Bioinform.\/}, {\bf 10},
  1384--1390.
\newblock (A preliminary version appeared in Proc.~of ISBRA 2012.).

\bibitem[Compeau, 2013]{COM-2013}
Compeau, P. E.~C. (2013).
\newblock {DCJ}-indel sorting revisited.
\newblock {\em Alg. Mol. Biol.\/}, {\bf 8}, 6.
\newblock (A preliminary version appeared in Proc.~of WABI 2012.).

\bibitem[Friedberg {\em et~al.\/}, 2008]{FRI-DAR-YAN-2008}
Friedberg, R., Darling, A.~E. \& Yancopoulos, S. (2008).
\newblock Genome rearrangement by the double cut and join operation.
\newblock In Keith, J.~M., (ed.) {\em Bioinformatics, Volume I: Data, Sequence
  Analysis, and Evolution\/}, volume 452 of {\em Methods in Molecular
  Biology\/}. Humana Press, pp. 385--416.

\bibitem[Hannenhalli \& Pevzner, 1995]{HAN-PEV-1995}
Hannenhalli, S. \& Pevzner, P.~A. (1995).
\newblock Transforming men into mice (polynomial algorithm for genomic distance
  problem).
\newblock In {\em Proceedings of the 36th Annual Symposium of the Foundations
  of Computer Science (FOCS 1995)\/}. IEEE Press, pp. 581--592.

\bibitem[Hannenhalli \& Pevzner, 1999]{HAN-PEV-1999}
Hannenhalli, S. \& Pevzner, P.~A. (1999).
\newblock Transforming cabbage into turnip: polynomial algorithm for sorting
  signed permutations by reversals.
\newblock {\em J. ACM\/}, {\bf 46}, 1--27.
\newblock (A preliminary version appeared in Proc.~of STOC 1995.).

\bibitem[Huson \& Bryant, 2005]{HUS-BRY-2005}
Huson, D.~H. \& Bryant, D. (2005).
\newblock {Application of Phylogenetic Networks in Evolutionary Studies}.
\newblock {\em Mol. Biol. Evol.\/}, {\bf 23}, 254--267.

\bibitem[Kumar {\em et~al.\/}, 2018]{KUM-STE-LI-KNY-2018}
Kumar, S., Stecher, G., Li, M. \& {Knyaz, et al.}, C. (2018).
\newblock {MEGA X: Molecular Evolutionary Genetics Analysis across Computing
  Platforms}.
\newblock {\em Mol. Biol. Evol.\/}, {\bf 35}, 1547--1549.

\bibitem[Lyubetsky {\em et~al.\/}, 2017]{LYU-GER-GOR-2017}
Lyubetsky, V., Gershgorin, R. \& Gorbunov, K. (2017).
\newblock Chromosome structures: reduction of certain problems with unequal
  gene content and gene paralogs to integer linear programming.
\newblock {\em BMC Bioinform.\/}, {\bf 18}, 537.

\bibitem[Martinez {\em et~al.\/}, 2015]{MAR-FEI-BRA-STO-2015}
Martinez, F.~V., Feijão, P. \& {Braga, et al.}, M. D.~V. (2015).
\newblock On the family-free {DCJ} distance and similarity.
\newblock {\em Alg. Mol. Biol.\/}, {\bf 10}, 13.
\newblock (A preliminary version appeared in Proc.~of WABI 2014.).

\bibitem[Rubert {\em et~al.\/}, 2020{\natexlab{a}}]{RUB-MAR-BRA-2020}
Rubert, D.~P., Martinez, F.~V. \& Braga, M. D.~V. (2020{\natexlab{a}}).
\newblock {Natural Family-Free Genomic Distance}.
\newblock In Kingsford, C. \& Pisanti, N., (eds.) {\em Proceedings of the 20th
  International Workshop on Algorithms in Bioinformatics (WABI 2020)\/}, volume
  172 of {\em Leibniz International Proceedings in Informatics (LIPIcs)\/}.
  Schloss Dagstuhl--Leibniz-Zentrum f{\"u}r Informatik, Dagstuhl, Germany, pp.
  3:1--3:23.

\bibitem[Rubert {\em et~al.\/}, 2020{\natexlab{b}}]{RUB-MAR-STO-DOE-2020}
Rubert, D.~P., Martinez, F.~V. \& {Stoye, et al.}, J. (2020{\natexlab{b}}).
\newblock Analysis of local genome rearrangement improves resolution of
  ancestral genomic maps in plants.
\newblock {\em BMC Genomics\/}, {\bf 21}.

\bibitem[Sankoff, 1992]{SAN-1992}
Sankoff, D. (1992).
\newblock Edit distance for genome comparison based on non-local operations.
\newblock In Apostolico, A., Crochemore, M., Galil, Z. \& Manber, U., (eds.)
  {\em Proceedings of the Third Annual Symposium on Combinatorial Pattern
  Matching, CPM 1992\/}, volume 644 of {\em LNCS\/}. Springer Verlag, Berlin,
  pp. 121--135.

\bibitem[Sankoff, 1999]{SAN-1999}
Sankoff, D. (1999).
\newblock Genome rearrangement with gene families.
\newblock {\em Bioinformatics\/}, {\bf 15}, 909--917.

\bibitem[Shao {\em et~al.\/}, 2015]{SHA-LIN-MOR-2015}
Shao, M., Lin, Y. \& Moret, B. M.~E. (2015).
\newblock An exact algorithm to compute the double-cut-and-join distance for
  genomes with duplicate genes.
\newblock {\em J. Comput. Biol.\/}, {\bf 22}, 425--435.
\newblock (A preliminary version appeared in Proc.~of RECOMB 2014.).

\bibitem[Visnovsk{\'a} {\em et~al.\/}, 2013]{Visnovska:2013ua}
Visnovsk{\'a}, M., Vina{\v{r}}, T. \& Brejov{\'a}, B. (2013).
\newblock {DNA Sequence Segmentation Based on Local Similarity.}
\newblock {\em ITAT\/}, 36--43.

\bibitem[Yancopoulos {\em et~al.\/}, 2005]{YAN-ATT-FRI-2005}
Yancopoulos, S., Attie, O. \& Friedberg, R. (2005).
\newblock Efficient sorting of genomic permutations by translocation, inversion
  and block interchange.
\newblock {\em Bioinformatics\/}, {\bf 21}, 3340--3346.

\bibitem[Yancopoulos \& Friedberg, 2009]{YAN-FRI-2009}
Yancopoulos, S. \& Friedberg, R. (2009).
\newblock {DCJ} path formulation for genome transformations which include
  insertions, deletions, and duplications.
\newblock {\em J. Comput. Biol.\/}, {\bf 16}, 1311--1338.
\newblock (A preliminary version appeared in Proc.~of RECOMB-CG 2008.).

\bibitem[Yin {\em et~al.\/}, 2016]{YIN-TAN-SCH-BAD-2016}
Yin, Z., Tang, J. \& {Schaeffer et al.}, S.~W. (2016).
\newblock Exemplar or matching: modeling {DCJ} problems with unequal content
  genome data.
\newblock {\em J. Comb. Opt.\/}, {\bf 32}, 1165--1181.

\end{thebibliography}

\end{document}